\documentclass[pdflatex,sn-mathphys-num]{sn-jnl}

\usepackage{graphicx}%
\usepackage{multirow}%
\usepackage{amsmath,amssymb,amsfonts}%
\usepackage{amsthm}%
\usepackage{mathrsfs}%
\usepackage[title]{appendix}%
\usepackage{xcolor}%
\usepackage{textcomp}%
\usepackage{manyfoot}%
\usepackage{booktabs}%
\usepackage{algorithm}%
\usepackage{algorithmicx}%
\usepackage{algpseudocode}%
\usepackage{listings}%

\usepackage{graphicx,amssymb}
\usepackage{amsmath,fleqn}
\usepackage{proof}
\usepackage{colortbl}  
\usepackage{soul} 
\usepackage{comment} 
\usepackage{makecell}

\usepackage{centernot}

\makeatletter
\let\MYcaption\@makecaption
\makeatother

\usepackage{subcaption}
\makeatletter
\let\@makecaption\MYcaption
\makeatother

\makeatletter
\newcommand{\customlabel}[2]{%
  \protected@write \@auxout {}{\string \newlabel {#1}{{#2}{}}}}
\makeatother

\usepackage{etoolbox}
\newcounter{proofid}
\newcommand{\eqprefix}{\thesection}
\AtBeginEnvironment{proof}{%
  \stepcounter{proofid}%
  \setcounter{equation}{0}%
  \edef\eqprefix{\thetheorem}%
}
\renewcommand{\theequation}{E:\eqprefix.\arabic{equation}}
\makeatletter
\@addtoreset{equation}{section}
\makeatother

\usepackage{tikz}
\usetikzlibrary{matrix}
\usetikzlibrary{arrows}
\usetikzlibrary{calc}
\usepackage{keyval}
\usetikzlibrary{positioning}
\usetikzlibrary{shapes.geometric}
\usetikzlibrary{patterns}  

\colorlet{darkgreen}{black!60!green!100}

\newcommand{\commentout}[1]{}
\newcommand{\suppress}[1]{} 

\renewcommand{\labelenumi}{(\arabic{enumi})}

\newcounter{trs}
\newcommand{\settrs}[1]{%
  \refstepcounter{trs}\label{#1}\ref{#1}%
}

\newlength{\ulength}
\newlength{\llength}
\newlength{\tempheight}
\newlength{\tempdepth}

\newcommand{\Farrow}[6]{%
  \settoheight{\tempheight}{$\scriptstyle #2$}%
  \settodepth{\tempdepth}{$\scriptstyle #3$}%
  \mathrel{%
    \rule[-\dimexpr \llength \ifx\relax#3\relax\else + \tempdepth + 0.5ex \fi \relax]{0pt}{\dimexpr \ulength + \llength \ifx\relax#2\relax\else + \tempheight + 0.5ex \fi \ifx\relax#2\relax\else + \tempdepth + 0.5ex \fi \relax}%
    {\smash{%
      \ooalign{%
        \hfil$\textstyle #1$\hfil\cr%
        \ifx\relax#2\relax\else \hidewidth\raisebox{1.3ex}{$\scriptstyle #2$}\hidewidth\cr \fi%
        \ifx\relax#3\relax\else \hidewidth\raisebox{-1.1ex}{$\scriptstyle #3$}\hidewidth\cr \fi%
      }%
    }}#6^{#4}_{#5}%
  }%
}
\newcommand{\Rarrow}[4]{\Farrow{\rightarrow}{#1}{#2}{#3}{#4}{}}
\newcommand{\Larrow}[4]{\Farrow{\leftarrow}{#1}{#2}{#3}{#4}{}}
\newcommand{\Barrow}[4]{\Farrow{\leftrightarrow}{#1}{#2}{#3}{#4}{}}
\newcommand{\Barrowd}[4]{\Farrow{\leftrightarrow}{#1}{#2}{#3}{#4}{'}}

\newcommand{\rarrow}[2]{\rightarrow^{#1}_{#2}}
\newcommand{\larrow}[2]{\leftarrow^{#1}_{#2}}
\newcommand{\barrow}[2]{\leftrightarrow^{#1}_{#2}}
\newcommand{\pnf}{\mathrm{pnf}}

\newcommand{\class}[1]{[#1]}
\newcommand{\classB}[1]{\text{\bf [}#1\text{\bf ]}}
\newcommand{\classD}[1]{[\![#1]\!]}
\newcommand{\Sub}{\mathrm{sub}}
\newcommand{\Root}{\mathrm{root}}
\newcommand{\restRoot}[1]{|\!|_{#1}}

\newcommand{\Dom}{\mathrm{Dom}}
\newcommand{\T}{\mathrm{T}}
\newcommand{\Pos}{\mathrm{Pos}}
\newcommand{\tpp}{\mathrm{\varepsilon}}
\newcommand{\ntpp}{{\tpp\neq}}
\newcommand{\inn}{\mathrm{\varepsilon<}}
\newcommand{\Ind}{\mathit{Ind}}
\newcommand{\Vrt}{\mathit{Vrt}}
\newcommand{\Hol}{\mathit{Hol}}
\newcommand{\Var}{\mathrm{Var}}

\newcommand{\conserveq}{\mathrel{\mathchoice%
    {\subseteq\mkern-15mu{\raise0.17em\hbox{\scriptsize$\mathrm{c}$}\mkern5mu}}%
    {\subseteq\mkern-15mu{\raise0.17em\hbox{\scriptsize$\mathrm{c}$}\mkern5mu}}%
    {\subseteq\mkern-11mu{\raise0.12em\hbox{\tiny$\mathrm{c}$}\mkern4mu}}%
    {\subseteq\mkern-11mu{\raise0.12em\hbox{\tiny$\mathrm{c}$}\mkern4mu}}%
}}
\newcommand{\supconserv}{\mathrel{\mathchoice%
    {\supseteq\mkern-15mu{\raise0.17em\hbox{\scriptsize$\mathrm{c}$}\mkern5mu}}%
    {\supseteq\mkern-15mu{\raise0.17em\hbox{\scriptsize$\mathrm{c}$}\mkern5mu}}%
    {\supseteq\mkern-11mu{\raise0.12em\hbox{\tiny$\mathrm{c}$}\mkern4mu}}%
    {\supseteq\mkern-11mu{\raise0.12em\hbox{\tiny$\mathrm{c}$}\mkern4mu}}%
}}
\newcommand{\NF}{\mathrm{NF}} 
\newcommand{\ddiv}{/\!/}

\newcommand{\stcongstar}{%
  \mathrel{\mathchoice%
    {\mbox{\stcongstartext}}%
    {\mbox{\stcongstartext}}%
    {\mbox{\stcongstarscript}}%
    {\mbox{\stcongstarscriptscript}}%
  }%
}
\newcommand{\stcongstartext}{%
  \raise0.4ex\vtop{\baselineskip0ex \lineskip.25ex%
    \ialign{##\crcr {$\smash{\gets^\ast}\rule[-0ex]{0pt}{1.6ex}$}\crcr%
      {$\smash{\to^\ast}\rule[-0.2ex]{0pt}{1ex}$}\crcr }%
  }%
}
\newcommand{\stcongstarscript}{%
  \raise0.25ex\vtop{\baselineskip0ex \lineskip.2ex%
    \ialign{##\crcr {$\scriptstyle\smash{\gets^\ast}\rule[-0ex]{0pt}{1.2ex}$}\crcr 
      {$\scriptstyle\smash{\to^\ast}\rule[-0.15ex]{0pt}{.7ex}$}\crcr} 
  }%
}
\newcommand{\stcongstarscriptscript}{%
  \raise0.2ex\vtop{\baselineskip0ex \lineskip.15ex%
    \ialign{##\crcr {$\scriptscriptstyle\smash{\gets^\ast}\rule[-0ex]{0pt}{.8ex}$}\crcr 
      {$\scriptscriptstyle\smash{\to^\ast}\rule[-0.1ex]{0pt}{.5ex}$}\crcr}%
  }%
}

\newcommand{\stcongstarss}[1]{\stcongstar_{\hspace{-0.8ex}{#1}}}
\newcommand{\<}{\langle}
\renewcommand{\>}{\rangle}

\newcommand{\lb}{(}
\newcommand{\rb}{)}

\tikzset{
  NODE/.style={circle, minimum width=35pt, draw=black, semithick},  
  every text node part/.style={font=\scriptsize},
  LINE/.style={semithick}
}
\tikzset{
  pics/simplecircle/.style n args={0}{
    code={
      \node[circle,draw=gray, very thin, inner sep=0pt, minimum size=40] (A0) at (0,0) {};
}}}
\tikzset{
  pics/confclass/.style n args={1}{
    code={
      \node[circle,draw=gray, very thin, inner sep=0pt, minimum size=35] (A0) at (0,0) {};
      \node[circle, inner sep=0pt, minimum size=1] (A1) at (0,2) {};
      \node[circle, inner sep=0pt, minimum size=1] (A2) at (-2,0) {};
      \node[circle, inner sep=0pt, minimum size=1] (A3) at (2,0) {};
      \node[circle, inner sep=0pt, minimum size=1] (A4) at (0,-2) {};
      \node at (0,4) {#1};

      \path[->]
      (A2) edge (A4)
      (A3) edge (A4)
      ;
}}}

\theoremstyle{thmstyleone}%
\newtheorem{theorem}{Theorem}
\newtheorem{proposition}[theorem]{Proposition}%
\newtheorem{lemma}[theorem]{Lemma}%
\newtheorem{corollary}[theorem]{Corollary}%

\theoremstyle{thmstyletwo}%
\newtheorem{example}[theorem]{Example}%
\newtheorem{remark}[theorem]{Remark}%

\theoremstyle{thmstylethree}%
\newtheorem{definition}[theorem]{Definition}%

\begin{document}

\title[Proving confluence]{A new method for proving confluence on abstract reduction systems\\
--- Confluence of non-$E$-overlapping weakly-shallow TRSs ---}


\author*[1]{\fnm{Masahiko} \sur{Sakai}}\email{sakai@rewriting.jp}

\author*[1,2]{\fnm{Mizuhito} \sur{Ogawa}}\email{mizuhito@otn.or.jp}

\author[1]{\fnm{Michio} \sur{Oyamaguchi}}\email{michiooyamaguchi@gmail.com}

\affil*[1]{\orgdiv{Grad.\ School of Informatics}, \orgname{Nagoya Univ.},
\orgaddress{\street{Furo-cho, Chikusa-ku}, \city{Nagoya}, \postcode{4648601}, \state{Aichi}, \country{Japan}}}

\affil[2]{\orgdiv{Old Teachers Network}, \orgname{OTN},
\orgaddress{\street{127-5 Ryusuke}, \city{Komatsu}, \postcode{9230926}, \state{Ishikawa}, \country{Japan}}}


\abstract{This paper proposes a new method for proving the confluence of
  an abstract reduction system (ARS)
  by clarifying the sufficient conditions, called
  \emph{compatibility} and \emph{edge commutativity},
  for expanding a given finite sub-ARS into a confluent one
  by adding rewrite edges.
  This method can be regarded as an extension of our earlier work,
  which showed that
  a weakly non-overlapping, shallow, and non-collapsing term rewriting system (TRS) is confluent.

  Furthermore, we apply our method to demonstrate that a non-$E$-overlapping and
  weakly shallow TRS is confluent.
  Here, a term is \emph{weakly shallow} if each defined function symbol occurs
  either at the root or in the ground subterms, and a TRS is weakly shallow
  if both sides of all its rewrite rules are weakly shallow.
  This drops the non-collapsing condition assumed in our previous work
  on weakly shallow TRSs.
  Moreover, since a weakly shallow TRS is non-$E$-overlapping whenever it is
  non-$\omega$-overlapping, and the latter property is decidable,
  we also obtain a decidable sufficient condition for confluence:
  non-$\omega$-overlapping and weakly shallow TRSs are confluent.
}


\maketitle

\section{Introduction}\label{sec1}

Confluence of term rewriting systems (TRSs) is undecidable,
even for flat TRSs~\cite{Mitsuhashi06} or
length-two string rewrite systems~\cite{Sakai08}.
For some subclasses, confluence is decidable, e.g., for shallow right-linear
TRSs by tree automata techniques~\cite{Godoy05}, and for terminating TRSs by
reducing to a finite search~\cite{KB70}.
Outside such subclasses, one examines sufficient conditions, which
guarantee confluence rather than decide it.
Many such conditions have been proposed, and they fall
into three broad approaches.

\begin{itemize}
\item Local confluence is equivalent to confluence in terminating TRSs~\cite{Newman1942}; this
  result has been extended to TRSs with relative termination~\cite{HM11,KH12}.
  A further criterion decomposes a TRS into a linear part and a terminating
  non-linear part~\cite{Liu14}, and requires establishing the existence of
  a suitable well-founded \emph{ranking}.

\item A second approach defines an explicit well-founded measure that decreases
  when peaks are eliminated~\cite{Oostrom94}.
  Many works study left-linear TRSs under the non-overlapping condition and
  its extensions~\cite{Rosen73,Huet80,Toyama87,Oostrom95,Oyama97,Okui98},
  measuring progress by the number of parallel reduction steps and its
  variations.
  For non-left-linear TRSs, only a few results are known~\cite{Toyama95,Gomi98},
  which are based on the depth-preserving property for non-$E$-overlapping
  TRSs.

\item A third approach sets constraints on left-nonlinear variables, as in
  the classical membership conditional TRSs~\cite{Toyama87b}
  (recently extended to Generalized TRSs~\cite{Lucas25}),
  and sorted higher-order rewrite systems~\cite{FJ25}.
  The idea underlying the membership condition, for instance, is to restrict
  the reduction strategy to a call-by-value discipline on left-nonlinear
  variables: a rule such as $d(x, x) \to b$ is applied only when the term
  substituted for the nonlinear variable $x$ is already a normal form
  (i.e., fully evaluated).
\end{itemize}

  The TRS $R_{\settrs{trs:nonCR}}$ in Example~\ref{ex:nonCR} below is a
  well-known example of a non-$E$-overlapping TRS that is not confluent.

\begin{example}\label{ex:nonCR}
  Let
  $R_{\ref{trs:nonCR}} =
  \{f(x)\to d(x,f(x)),\, d(x,x)\to b,\, a\to f(a) \}$~\cite{Klop80}.
  This TRS has two divergent sequences $f(a)\to^\ast b$ and $f(a)\to^\ast f(b)$.
  Since $b$ is a normal form and every reduct of $f(b)$ has the form
  $d(b,\dots,d(b,f(b))\dots)$, the terms $b$ and $f(b)$ are not joinable.
\end{example}

Beyond this example, the following problem remains open.
``\emph{A non-$E$-overlapping right-linear TRS is confluent}'' is
conjectured as Problem 58 of
\emph{Open Problems in Rewriting}~\cite{RTAopen1991}.\footnote{%
Originally, Problem 58 assumes the stronger condition of
\emph{strongly non-overlapping} in place of \emph{non-E-overlapping}.}
To the best of our knowledge, the strongest partial solution is
``\emph{a non-$E$-overlapping simple right-linear TRS is confluent}'',
where a right-linear TRS is simple right-linear if no left-nonlinear
variable occurs
in the right-hand side of any rewrite rule~\cite{OyamaguchiOhtaToyama1995}.
This result, like those in the second approach above, is obtained by
designing a suitable well-founded measure on rewrite sequences;
despite further attempts, however, no such measure has yet been found
for the general case of Problem 58.

Our ultimate goal is to settle this conjecture.
The difficulty, compared with the left-linear case, lies in how outer and inner
reductions can be reordered within a rewrite sequence.
In a right-linear TRS, one can always move an outer step after the inner steps
that follow it (as in \emph{basic reductions}), but not conversely:
an inner step that precedes an outer step cannot, in general, be moved after it.
A left-linear TRS has exactly the opposite property, so its inner reductions
may be postponed, that is, performed lazily; this laziness is what keeps the
joining procedure terminating.
A right-linear TRS instead forces eager inner reductions, which may continue
indefinitely, thereby ruling out any well-founded measure argument.
Controlling these inner reductions is therefore the key to the joining procedure.

After these attempts to construct a suitable well-founded measure failed, we proposed
a different methodology using \emph{reduction graphs}~\cite{Sakai10};
starting from the initial divergent rewrite sequence,
we add rewrite edges one by one to join the two terms,
while keeping the graph finite throughout the construction.
The basic reduction nature in a right-linear TRS suggests constructing
confluent sequences in an innermost manner: starting from a confluent sequence
at the leaves, we extend it with a surrounding context and add outer reductions
at the root.
To avoid interference from other reductions at intermediate positions,
the function symbols occurring between the root and a variable in this context
are currently restricted to constructor symbols
--- a restriction that amounts to the \emph{weakly shallow} condition.
Under this restriction, we showed that
``\emph{non-$E$-overlapping, weakly shallow,
and non-collapsing TRSs are confluent}''~\cite{Sakai15}.
The key observation behind this result is that, under the non-$E$-overlapping
condition, peak elimination uses only ``\emph{copies}'' of the reductions already
present in the original rewrite sequence.
Note that this result assumes no linearity condition:
once the weakly shallow restriction is imposed, the bottom-up construction
goes through without right-linearity, and the right-linearity assumption is
dropped altogether.
The class of weakly shallow TRSs is thus incomparable with the class of
Problem~58 --- right-linearity is traded for the weakly shallow restriction ---
and the present work develops, within this class, the graph-based machinery
aimed at the general case.

Our notion of a \emph{reduction graph} bounds reductions within a finite set of
terms, so that inner reductions terminate once they reach the boundary of this
set.
This construction, however, is fragile: neither monotonicity nor closure under
substitution of reductions holds in general.
Instead, we rely on two structural facts.
The non-$E$-overlapping property ensures that $E$-normalization also works for
reduction graphs.
The weakly shallow property ensures that non-root reductions between redexes
correspond to reductions between variable positions.

This paper makes two contributions.
First, we propose a general framework for proving confluence of abstract
reduction systems, built on the notions of compatibility and edge commutativity,
which extend edge normalizability by means of edge expansion.
  Note that the finiteness of reduction graphs guarantees
  the termination of edge normalization.
Second, using this framework, we show that
``\emph{weakly shallow and non-$E$-overlapping TRSs are confluent}'',
thereby removing the non-collapsing restriction~\cite{Sakai15}.
Here, a term is \emph{weakly shallow} if every defined function symbol occurs
either at the root or within a ground subterm, and a TRS is \emph{weakly shallow}
if both sides of every rule are weakly shallow.
This result, together with an outline of its proof, was first presented
in the extended abstract~\cite{Sakai14}, and the proof given here follows
essentially the same line.
In this paper, we refine that line by abstracting those parts of it that can
be formulated for abstract reduction systems;
without this abstraction, the proof would be considerably more complicated and harder
to follow.
A full proof without this abstraction was given in~\cite{Sakai15},
but only for the restricted case of non-collapsing TRSs.

\paragraph*{Remarks.}

Weakly shallow TRSs are Turing complete.
Klop gave a TRS encoding that simulates a given Turing machine
(Example~\ref{ex:turing}).
It directly shows that the word problem and termination are
undecidable for weakly shallow TRSs, which makes a sharp contrast
with shallow TRSs, of which the word problem is decidable~\cite{CHJ94}.
Mitsuhashi et al.\ showed that confluence and the non-$E$-overlapping
property are undecidable for weak monadic TRSs~\cite{Mitsuhashi12},
which is equivalent to right weakly-shallow.
Their undecidability proofs use Example~\ref{ex:turing} as
a counterexample, and they work in the same way for weakly shallow TRSs.
\begin{example}[{\cite[Exercise~2.2.21]{Klop93}}]\label{ex:turing}
  Let $Q$ and $\Gamma$ be finite sets of
  states and tape symbols of a Turing machine $M$, respectively, and
  let
  $\delta:\ Q\times\Gamma\to Q\times\Gamma\times
            \{\mathrm{left},\mathrm{right}\}$
  be its transition function.
  Each configuration $a_1\cdots a_i q a_{i+1}\cdots a_n\in\Gamma^+ Q \Gamma^+$
  (where $q\in Q$) is represented by the term
  $q(a_i\cdots a_1(\$),a_{i+1}\cdots a_n(\$))$,
  where the arities of the function symbols $q$, $a_j$ ($1\leq j\leq n$) and
  $\$ $ are 2, 1 and 0, respectively.
  The corresponding TRS $R_M$ is weakly shallow and consists of the following
  rewriting rules:
  \[
  R_M =
  \left\{
  \begin{array}{ll}
    q(x,a(y)) \rightarrow p(b(x),y) &
    \text{ for } \delta(q,a)=(p,b,\mathrm{right}),\\
    q(a'(x),a(y)) \rightarrow p(x,a'(b(y))) &
    \text{ for } \delta(q,a) = (p,b,\mathrm{left})
  \end{array}
  \right\}
  \]
\end{example}

A decidable sufficient condition for the non-$E$-overlapping property is
\emph{strongly non-overlapping}~\cite{Ogawa89},
where a TRS is \emph{strongly non-overlapping}
if its linearization is non-overlapping~\cite{Klop80,Chew81}.
A weaker, still decidable, sufficient condition is
\emph{non-$\omega$-overlapping}~\cite{Ogawa89}, that is, non-overlapping
with respect to unification over infinite terms~\cite{MR84}:
every strongly non-overlapping TRS is non-$\omega$-overlapping, and every
non-$\omega$-overlapping TRS is
non-$E$-overlapping
(M.~Oyamaguchi, oral presentation at the 48th TRS Meeting, 2018;
N.~Hirokawa and M.~Oyamaguchi, personal communication, 2019).
Each of these conditions implies non-overlapping, and these inclusions
are strict, as shown in Examples~\ref{ex:strict} and~\ref{ex:nonCR2}.
For left-linear TRSs, all of these conditions coincide.

\begin{example}\label{ex:strict}
  The TRS $\{d(x,x)\to b\, d(g(x),h(x))\to c\}$ is
  non-$\omega$-overlapping but not strongly non-overlapping, since
  $d(x,x)$ and $d(g(y),h(y))$ are not unifiable even over infinite terms,
  whereas its linearization $\{d(x,y)\to b,\, d(g(x),h(y))\to c\}$ is
  overlapping.

  Likewise, $\{d(x,x)\to b,\, d(x,f(x))\to c\}$ is non-$E$-overlapping but
  not non-$\omega$-overlapping: $d(x,x)$ and $d(y,f(y))$ are unifiable over
  infinite terms by $f^\omega$, whereas no instances of them are convertible
  below the root, since this terminating and non-overlapping, hence confluent,
  TRS has $b$ and $c$ as distinct normal forms.
\end{example}

The non-$E$-overlapping property alone is not enough to guarantee
confluence,
and guarantees the unique normal form property only~\cite{MO01}.
The TRS in Example~\ref{ex:nonCR} is non-$E$-overlapping and not
confluent, though it has unique normal forms.

The $E$-overlapping condition depends on dynamic behavior of TRSs.
The TRS in Example~\ref{ex:nonCR2} is $E$-overlapping and not confluent.
However, if we drop the third rule $a \to f(a)$ from
Example~\ref{ex:nonCR2}, it becomes non-$E$-overlapping, since
the unifier is no more simulated by reductions.

\begin{example}\label{ex:nonCR2}
  $R_{\settrs{trs:nonCR2}} =
  \{d(x,x)\to b,\, d(x,f(x))\to c,\, a \to f(a)\}$
  is another nonoverlapping non-confluent example~\cite{Huet80}.
  It is $\omega$-overlapping and $f^\omega$ is a unifier with infinite terms
  between the first and the second rules;
  It is, however, $E$-overlapping,
  since the instances
  $d(a,a)$ and $d(a,f(a))$ of its left-hand sides are related by a step
  below the root.


\end{example}

Our result, and its non-collapsing special case~\cite{Sakai15}, give
  provers a new criterion for establishing confluence of
  TRSs.\footnote{Confluence Competition is held annually
  (See \url{https://project-coco.uibk.ac.at/}).
  Its problem database, the ARI database, is available at
  \url{https://ari-cops.uibk.ac.at/ARI/}, where each problem can also be
  run with any participating tool since 2015.}
  The tool \textsf{CSI} has proved Example~\ref{ex:1619} confluent
  since the 2015 competition~\cite{CoCo15}, using the criterion of~\cite{Sakai15};
  our result extends this criterion to collapsing TRSs.

\begin{example}[{ARI~\#1619}]\label{ex:1619}
  The TRS
  $\{ d(x,x) \rightarrow h(x),\, f(x) \rightarrow d(x,f(c)),\,
  h(x) \rightarrow f(g(x)) \}$
  is weakly shallow, non-$E$-overlapping, and non-collapsing;
  hence confluent by~\cite{Sakai15}.
\end{example}

\paragraph*{Organization.}

Section~\ref{sec:overview} gives an informal overview of the construction, and
Section~\ref{subsec:organization} lists which section contains what.

\section{Overview and intuition: A new method to prove confluence}\label{sec2}\label{sec:overview}

Traditional methods for proving confluence show how to build a join sequence
for each given divergent sequence, and guarantee that replacing a peak with a
smaller valley eventually stops with respect to some well-founded order.
This paper takes the opposite direction.
Rather than reasoning about the rewrite relation as a whole, we construct,
for each given divergent sequence, a \emph{finite} and confluent subgraph of
the rewrite relation that contains that sequence,
  which is called a reduction graph~\cite{Sakai10} and an $R$-graph in this paper
  for a TRS $R$.
Confluence of the system is then immediate, since any two divergent reducts
already meet inside the constructed finite graph.

The benefit is that confluence becomes a property of a finite graph.
We no longer have to design a well-founded measure on rewrite sequences;
what we need instead is a graph construction that terminates, and the
finiteness of the graph is what makes it terminate.
The price is that the construction itself has to be carried out with care,
and this is what the paper is about.

Such a finite subgraph is called an \emph{$R$-graph}; its formal definition is
given in Section~\ref{sec:R-reduction-graph}.
For the present section it suffices to read an $R$-graph as a finite set of
terms together with a finite set of rewrite steps of $R$ between them.
The construction was introduced by Sakai and Ogawa~\cite{Sakai10}, where it
was used for weakly non-overlapping, shallow, and non-collapsing term
rewriting systems (TRSs).
The present paper abstracts the construction to the setting of abstract
reduction systems (ARSs), which separates it from the syntactic details of
terms, and then instantiates the abstract construction to prove that
non-$E$-overlapping and weakly shallow TRSs are confluent.

Throughout this section the construction is described informally, and each
notion is accompanied by a pointer to the place where it is defined.
Section~\ref{subsec:organization} lists which section contains what.

\subsection{The construction in the former approach}\label{subsec:SO}

We first illustrate the idea of~\cite{Sakai10}, which builds, in a bottom-up
manner, a finite, confluent, and terminating $R$-graph that preserves the
connectivity of a given $R$-graph.
What matters here is the shape of the construction, not the class of TRSs to
which~\cite{Sakai10} applies it.

Consider the TRS
\(R_{\settrs{trs:simple-div}}=\{f(x,x)\to g(x), a\to b, b\to a\}\)
and the divergent sequence $G_0$ shown in Fig.~\ref{fig:graph-divergence}.
The sequence $G_0$ is a finite $R_{\ref{trs:simple-div}}$-graph.
We extract its direct-subterm graph $G$: $a \leftrightarrows b$, which is
obtained by collecting all arguments of $f$ and $g$ occurring in $G_0$,
with the edges induced by the non-root reductions of $G_0$.
Since $G$ is syntactically smaller than $G_0$, the induction hypothesis
applies to it and constructs a finite, confluent, and terminating
$R_{\ref{trs:simple-div}}$-graph, say $G'$: $a \to b$, that preserves the
connectivity of $G$.
\begin{figure}
  \centering
    \begin{minipage}{2em}
      $G_0$:
    \end{minipage}
    \begin{minipage}{0.4\textwidth}
      \(
      \begin{tikzpicture} [line width=0.07em]
        \matrix[matrix of math nodes,row sep=0.5em,column sep=1em,minimum width=0.5em] (m) {
          f(b,b) & f(a,b) & f(a,a) & g(a) \\
        };
        \path[->]
        (m-1-2) edge (m-1-1)
        (m-1-2) edge (m-1-3)
        ;
        \path[->,color=blue]
        (m-1-3) edge node [above] {\scriptsize $\varepsilon$} (m-1-4)
        ;
      \end{tikzpicture}
      \)
    \end{minipage}

  \caption{A given divergent sequence for $R_{\ref{trs:simple-div}}$}\label{fig:graph-divergence}
\end{figure}

We then lift $G'$ one layer up, obtaining the graph drawn with black edges in
Fig.~\ref{fig:graph-construction}.
Its vertices are the terms $k(s_1,\ldots,s_n)$ for each function symbol $k$ of
the TRS and vertices $s_1,\ldots,s_n$ of $G'$, and its edges are those
obtained by replacing one argument along the edge $a\to b$ of $G'$.
This graph is written $M_F(G')$ in this paper (Definition~\ref{def:mono-ext});
the rightmost column of the figure shows its two vertices coming from the
constants $a$ and $b$, which carry no edge since they have no argument.
Since $G'$ is convergent, so is $M_F(G')$.
\begin{figure}
  \begin{center}
    \mbox{}\\[-1em]
    \(
    \begin{tikzpicture} [line width=0.07em]
      \matrix[matrix of math nodes,row sep=0.8em,column sep=1em,minimum width=0.5em] (m) {
        f(a,b) & f(a,a) & g(a) & a\\
        f(b,a) & f(b,b) & g(b) & b\\
      };
      \path[<-] (m-1-1) edge (m-1-2) (m-2-1) edge (m-1-2);
      \path[->] (m-1-1) edge (m-2-2) (m-2-1) edge (m-2-2);
      \path[->] (m-1-3) edge (m-2-3);  

      \path[->,dashed,color=blue] (m-1-2) edge node [above] {\scriptsize $\varepsilon$} (m-1-3);
      \path[->,color=red]  (m-2-2) edge node [above] {\scriptsize $\varepsilon$} (m-2-3)  ;
    \end{tikzpicture}
    \)\\[-2em]\mbox{}
  \end{center}
  \caption{Graph construction for the divergence in Fig.~\ref{fig:graph-divergence}}\label{fig:graph-construction}
\end{figure}
What is left is to recover the connection provided by the root reduction
$f(a,a) \to g(a)$ of $G_0$, in such a way that the result stays confluent.
Adding that edge itself (the dashed blue edge in
Fig.~\ref{fig:graph-construction}) does not work: in the resulting graph,
$f(a,a)$ reduces to the two distinct terminal vertices $f(b,b)$ and $g(b)$.
Instead we add the edge $f(b,b) \to g(b)$ (the red edge in the figure),
whose source $f(b,b)$ is a normal form of $M_F(G')$.
Both edges are instances of the same rewrite rule, so the resulting graph is
still an $R_{\ref{trs:simple-div}}$-graph, and it is confluent.
We call the operation of replacing a root edge by one that starts from a
normal form an \emph{edge normalization}.

Summarizing, the property established in~\cite{Sakai10} in order to derive
confluence of a TRS $R$ is the following.
\begin{quote}\normalsize
For an arbitrary $R$-graph $G_0$, there exists an $R$-graph $G_0'$ such that
\begin{enumerate}
\item $G_0'$ is finite, confluent, and terminating, and
\item $V_{G_0} \subseteq V_{G_0'}$ and ${\barrow{\ast}{G_0}} \subseteq {\barrow{\ast}{G_0'}}$.
\end{enumerate}
\end{quote}

\subsection{The claim $C_P$ and the induction scheme}\label{subsec:framework}

Abstracting the above, this paper works with the following claim $C_P$,
parameterized by a property $P$ of $R$-graphs.
\begin{quote}\normalsize
$C_P$:
For an arbitrary $R$-graph $G_0$, there exists an $R$-graph $G_0'$ such that
\begin{enumerate}
\item[(1')] $G_0'$ is finite and confluent,
\item[(2')] $G_0 \subseteq G_0'$, and
\item[(3')] $P(G_0')$.
\end{enumerate}
\end{quote}
The property $P$ is chosen according to the class of TRSs under consideration;
for the classes treated in this paper it will be the \emph{subterm-closed}
property discussed in Section~\ref{subsec:collapsing}.

Two of the three conditions differ from those of~\cite{Sakai10}, and it is
worth explaining why.

Condition (1) is weakened to (1') by dropping termination.
The reason lies in the edge normalization of Section~\ref{subsec:SO}.
There, we could add the edge $f(b,b) \to g(b)$ because the normal form
$f(b,b)$ of $M_F(G')$ happened to be a redex.
This is a consequence of shallowness and fails in general: for a non-shallow
system, a normal form of the inner part of the graph need not be a redex, so
before normalizing an edge we must first \emph{expand} the graph by fresh
edges until a redex is reached.
Such expansions may destroy termination.
Termination could be restored afterwards by deleting suitable edges, but it
plays no role in the argument, so we simply drop it.

Condition (2) is strengthened to (2') for the same reason, read backwards.
The weak form (2) was needed in~\cite{Sakai10} precisely because edges were
deleted in order to keep the graph terminating.
Once we stop deleting edges, we may as well keep the plain inclusion
$G_0 \subseteq G_0'$, which is easier to work with in the later sections.

Dropping termination has one immediate consequence.
In a non-terminating graph there need not be any terminal vertex, and edge
normalization has nothing to start from.
We therefore replace normal forms by \emph{pseudo-normal forms}, that is,
vertices lying in a strongly connected component that has no leaving edge
(Definition~\ref{def:strongly-convergent} in
Section~\ref{subsec:looping}).
The resulting notion of edge expansion, and the conditions under which it
preserves confluence, are developed in Sections~\ref{subsec:looping}
and~\ref{subsec:Conv-Exp-ARS}.
It is more involved than the corresponding notion under termination, but it
relieves all later sections of termination arguments.

The proof of a claim $C_P$ follows the three steps below, as in~\cite{Sakai10}.
\begin{description}
\item[Step 1:] Construct $G$ as the direct-subterm graph of a given $G_0$.
\item[Step 2:] Let $G'$ be an $R$-graph satisfying (1')--(3') that is obtained
  from $G$ by the induction hypothesis.
\item[Step 3:] Construct from $G'$ and $G_0$ an $R$-graph $G_0'$ that
  satisfies (1')--(3').
\end{description}
Step 3 carries the whole weight of the proof, together with the choice of the
property $P$.
It starts from a finite $R$-graph $G_1$ containing all non-root reductions of
$G_0$; concretely we take either $M_F(G')$ or the subgraph of $G_0$ formed by
its non-root reductions.
We then select a set $U$ of root edges by normalizing the root edges that have
to be recovered, adding vertices and non-root edges to $G_1$ on demand.
Writing $G_1'$ for the result of these additions, we obtain
$G_0' := G_1' \cup U$, and the task is to make this graph confluent.

\subsection{Compatibility and edge commutativity}\label{subsec:key-construction}

The construction of Step 3 is carried out abstractly, on ARSs, by
Theorem~\ref{thm:seq-edge-norm} in Section~\ref{subsec:confluent-edge-norm}.
Its inputs are the finite ARS $G_1$ described above and two finite sets $T_1$
and $T_2$ of root edges to be recovered, consisting respectively of the
non-collapsing and the collapsing ones.
Its output is a finite confluent ARS $G_1' \cup T'$, where $G_1'$ extends
$G_1$ by the non-root edges added on demand, and $T'$ is the set $U$ of
normalized root edges.
The edges added to $G_1$ are taken from $\Vrt$, a confluent, possibly
infinite, ARS built from the graph $G'$ delivered by the induction hypothesis,
and the edges of $T'$ are taken from $\Hol$, an ARS consisting of root edges
on the vertices of $\Vrt$.
The theorem guarantees that $G_1' \cup T'$ retains all the connectivity of
$G_1 \cup T_1 \cup T_2$.

Theorem~\ref{thm:seq-edge-norm} rests on two properties of the root edges
relative to $\Vrt$ and $\Hol$: \emph{edge commutativity}
(Definition~\ref{def:edge-commutative}) and \emph{compatibility}
(Definition~\ref{def:compatible}).
Fig.~\ref{fig:abs-diag} shows what each of them contributes.
Each circle in the figure stands for a connected component of the non-root
reductions of $G_1$, and the blue edge is a non-collapsing root edge in
$T_1$ that we wish to normalize.
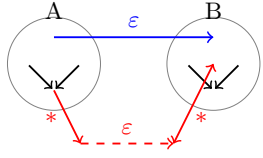
\begin{figure}
  \begin{center}
    \begin{tikzpicture}[line width=0.07em, x=5, y=5]
      \path pic at (0,0) {confclass={A}};  
      \path pic at (12,0) {confclass={B}}; 

      \path[->,color=blue] (0,2) edge node [above] {\scriptsize $\tpp$} +(12,0);
      \path[->,color=red]
      (0,-2) edge node[left] {$*$} +(2,-4)
      (2,-6) edge[dashed] node [above] {\scriptsize $\tpp$} +(7,0)
      ;

      \path[<->]
      (12,0) edge[color=red] node[right,xshift=-3,yshift=-5] {$*$} (9,-6);
    \end{tikzpicture}
  \end{center}
  \caption{Abstract presentation of edge normalization}\label{fig:abs-diag}
\end{figure}
\emph{Edge commutativity} is what makes the normalization possible at all.
It states that the root edge survives further reduction of its source:
however far the source of the blue edge is reduced within $\Vrt$, some root
edge of $\Hol$ is still reachable, and its target is joinable with the target
of the blue edge.
In the figure this is the existence of the red non-root edges, drawn inside
$\Vrt$, together with the red dashed root edge, drawn inside $\Hol$.
Note that the source of the red dashed edge must be a pseudo-normal form of
the extended $G_1$ for the result to be confluent; this is exactly what the
normalization step of Section~\ref{subsec:normalization} arranges.

\emph{Compatibility} is what bounds the amount of work.
It states that a reduction at a redex is deterministic modulo $\Vrt$: if two
root edges leave the same $\Vrt$-component, their targets are already
joinable.
Consequently a single edge normalization per component suffices.
In Fig.~\ref{fig:abs-diag}, normalizing the blue edge shown is enough even if
the component A has several outgoing root edges.

The design of $\Vrt$, and with it that of $\Hol$, is where the class of TRSs
enters the picture; $\Hol$ is in fact determined by $\Vrt$ as the set
$\Ind_R(V_{\Vrt})^{\tpp}$ of root edges on the vertices of $\Vrt$.
Suitable choices for shallow and for weakly shallow TRSs are given in
Sections~\ref{sec:shallow} and~\ref{sec:WS_system}, respectively.
The essential point in both cases is that $\Vrt$ must supply enough non-root
edges to keep the subterms at non-linear variable positions of a rewrite rule
in balance; this is why $\Vrt$ is required to be confluent.

\subsection{Extending the class, I: collapsing rules}\label{subsec:collapsing}

To reach the class of non-$E$-overlapping and weakly shallow TRSs, the
construction has to be extended in two independent directions: it must admit
collapsing rewrite rules, and it must survive the relaxation of shallowness.
We discuss the two in turn.

A collapsing rule breaks the layering between non-root and root reductions on
which the construction of Section~\ref{subsec:SO} relies, since a root
reduction may now expose a subterm.
The remedy is to choose, as the property $P$ of the claim $C_P$, the
\emph{subterm-closed} property (Definition~\ref{def:subterm-closed} in
Section~\ref{subsec:subterm-closed}), which forces the non-root reductions of
the graph to be present as root reductions as well.
The normalization of a collapsing edge is depicted in
Fig.~\ref{fig:collapsing-treatment}, where the blue edge from F to A is
induced by a collapsing rule and the blue edge from A to B by a
non-collapsing one.
\begin{figure}
  \begin{center}
    \mbox{}\\[-1em]
    \begin{tikzpicture}[line width=0.07em, x=5, y=5]
      \path pic at (0,0) {confclass={F}};  
      \path pic at (12,0) {confclass={A}}; 
      \path pic at (24,0) {confclass={B}}; 

      \path[->,color=blue] (0.2,2) edge node [above] {\scriptsize $\tpp$} +(11,0);
      \path[->,color=red] (1,-2) edge [bend right] node [above] {\scriptsize $\tpp$} +(22.5,-0.2);

      \path[->,color=blue] (12.2,2) edge node [above] {\scriptsize $\tpp$} +(12,0);
      \path[->,color=red] (12.2,-2) edge node [above right] {\scriptsize $\tpp$} +(11,0);
    \end{tikzpicture}
    \\[-2em]\mbox{}
  \end{center}
  \caption{Normalization of collapsing edge}\label{fig:collapsing-treatment}
\end{figure}
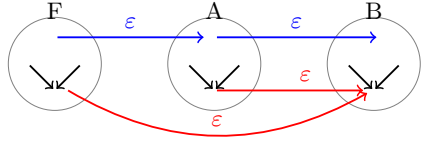
Normalizing a collapsing edge may thus require passing through a
non-collapsing one, so the two properties of
Section~\ref{subsec:key-construction} are needed in a relative form.
Both definitions are given with respect to an auxiliary ARS $N$: $T$ is
$N$-compatible over $\Vrt$ (Section~\ref{subsec:compatibility}) and
$N$-edge-commutative over $(\Vrt,\Hol)$
(Section~\ref{subsec:commutativity}) when the two targets are joinable in
$\Vrt \cup N$ rather than in $\Vrt$ alone.
For collapsing edges we take for $N$ the set of non-collapsing root edges,
and for non-collapsing edges we take $N$ empty, which recovers the plain
notions.
Under subterm-closedness of the non-root part of the graph, the collapsing
edges satisfy the relative form of edge commutativity once the redundant ones
among them have been trimmed away (Lemma~\ref{lem:NestedCollapsing}).

\subsection{Extending the class, II: beyond shallowness}\label{subsec:beyond-shallow}

The second direction is the one already announced in
Section~\ref{subsec:framework}: without shallowness, a pseudo-normal form of
the non-root part of the graph need not be a redex, and edge expansions
become unavoidable.

Expansions interact with the order in which edges are normalized.
Fig.~\ref{fig:right-order} shows the normalization of two non-collapsing,
non-shallow edges in a favourable order: the left blue edge is normalized
first, then the right one.
Performing them in the reverse order yields the graph of
Fig.~\ref{fig:wrong-order}, in which one further normalization step is
required before confluence is reached.
In other words, normalizing one edge may invalidate the normalization of
another edge that had already been normalized, so the process is not
obviously terminating.
That it does terminate is proved in
Section~\ref{subsec:norm-termination}
(Theorem~\ref{thm:termination-normalization}), by exhibiting a measure on
spanning graphs that strictly decreases at every normalization step.
\begin{figure}
  \begin{center}
    \begin{tikzpicture}[line width=0.07em, x=5, y=5]
      \path pic at (0,0) {confclass={A}};  
      \path pic at (12,0) {confclass={B}}; 
      \path pic at (24,0) {confclass={C}}; 

      \path[->,color=blue] (0,2) edge node [above] {\scriptsize $\tpp$} +(12,0);
      \path[->,color=red]
      (0,-2) edge[color=red] +(2,-2)
      (2,-4) edge node [above] {\scriptsize $\tpp$} +(8,-4)
      ;

      \path[->]
      (12,-2) edge[color=red] (11,-5)
      (11,-5) edge[color=red] (10,-8);
      \path[->]
      (10,-8) edge[color=red] (13,-11);
      \path[->,color=blue] (12,2) edge node [above] {\scriptsize $\tpp$} +(12,0);
      \path[->,color=red]  (13,-11) edge node [above] {\scriptsize $\tpp$} +(8,0);
      \path[->,color=red]
      (24,-2) edge (23,-5)
      (23,-5) edge (22,-8)
      (22,-8) edge (21,-11);
    \end{tikzpicture}
  \end{center}
  \caption{Normalizing the left edge first}\label{fig:right-order}
\end{figure}
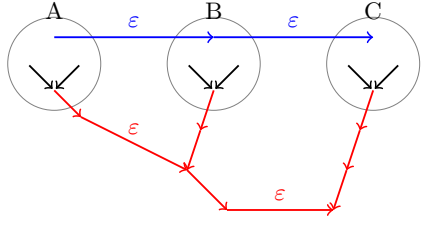
\begin{figure}
  \begin{center}
    \begin{tikzpicture}[line width=0.07em, x=5, y=5]
      \path pic at (0,0) {confclass={A}};  
      \path pic at (12,0) {confclass={B}}; 
      \path pic at (24,0) {confclass={C}}; 

      \path[->,color=blue] (0,2) edge node [above] {\scriptsize $\tpp$} +(12,0);
      \path[->,color=red]
      (0,-2) edge[color=red] +(2,-2)
      (2,-4) edge node [above] {\scriptsize $\tpp$} +(8,-4)
      ;

      \path[->]
      (13,-5) edge[color=red] (10,-8);
      \path[->]
      (12,-2) edge[color=red] (13,-5);
      \path[->,color=blue] (12,2) edge node [above] {\scriptsize $\tpp$} +(12,0);
      \path[->,color=red]  (13,-5) edge node [above] {\scriptsize $\tpp$} +(10,0);
      \path[->] (24,-2) edge[color=red] (23,-5);
    \end{tikzpicture}
  \end{center}
  \caption{Normalizing the right edge first, which needs one more step}\label{fig:wrong-order}
\end{figure}
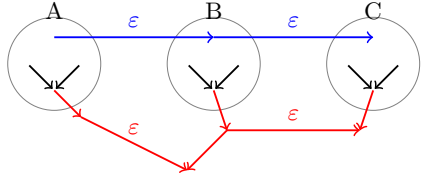

\subsection{Organization of the paper}\label{subsec:organization}

Sections~\ref{sec:preliminary}--\ref{Sec:CompatCommute} develop the
construction on abstract reduction systems, and
Sections~\ref{sec:R-reduction-graph}--\ref{sec:WS_system} instantiate it on
term rewriting systems.

Section~\ref{sec:preliminary} fixes the notation for ARSs and collects the
tools used throughout: conservative extension, unions that preserve
confluence, and quotient ARSs.
Section~\ref{sec:conversion_ARS} introduces pseudo-normal forms and edge
expansion, and characterizes when an edge expansion preserves confluence
(Lemma~\ref{lem:expansion-PNF}); it also shows how a sub-ARS of a confluent
ARS can be extended to a finite confluent one.
Section~\ref{sec:edge_properties} defines the \emph{spanning graph}
$(G,T_1,T_2)$, which records the graph under construction together with its
not-yet-normalized root edges, defines the normalization step $\vdash_n$ on
spanning graphs, and proves that $\vdash_n$ terminates.
Section~\ref{Sec:CompatCommute} defines compatibility and edge commutativity
and combines the above into Theorem~\ref{thm:seq-edge-norm}, the main result
on the ARS side.

Section~\ref{sec:R-reduction-graph} sets up TRSs and their $R$-graphs, splits
the edges of an $R$-graph into root and non-root ones, and introduces the
direct-subterm graph $\Sub(G)$, the monotonic extension $M_F(G)$, and the
subterm-closed property.
Section~\ref{sec:shallow} applies Theorem~\ref{thm:seq-edge-norm} to
non-$E$-overlapping and shallow TRSs, taking $\Vrt := M_F(G)$;
the case without collapsing rules gives Theorem~\ref{thm:shallow-noncollapsing} and the
general case Theorem~\ref{thm:NOV-Shallow}.
Section~\ref{sec:WS_system} treats non-$E$-overlapping and weakly shallow
TRSs, for which a single layer of extension no longer suffices: in this case, $\Vrt$
is obtained by first closing the graph under constructor contexts of
unbounded depth and only then adding one layer of function symbols.
This leads to Theorem~\ref{thm:main} and Corollary~\ref{cor:main}.
Section~\ref{sec10} concludes.

\section{Abstract reduction system} \label{sec3}\label{sec:preliminary}

\subsection{Notation and terminology of ARS}

For a binary relation $\rightarrow$ on a set $V$,
the \emph{domain of $\rightarrow$}, denoted by $\Dom({\to})$, is
$\{ v \mid (v,v') \in {\to} \text{ or } (v',v) \in {\to}\}$.
Clearly, $\Dom({\to}) \subseteq V$.
We use $\leftarrow$, $\leftrightarrow$,
$\rarrow{+}{}$ and $\rarrow{\ast}{}$
for the inverse relation, the symmetric closure,
the transitive closure, and the reflexive and transitive closure
of $\rightarrow$, respectively.  We use $\cdot $
for the composition operation of two relations.

An \emph{abstract reduction system} (ARS) is a directed graph
$G=\< V, \rightarrow\>$ with a reduction relation ${\rightarrow} \subseteq V \times V$.
If $(u,v) \in \mathord{\rightarrow}$,
we write it as $u\rightarrow v$.
For $V', V'' \subseteq V$, the edge restriction $\mathord{\rightarrow}|_{V'\times V''}$ stands
for $\mathord{\rightarrow} \cap (V'\times V'')$.
Moreover, we use $\mathord{\rightarrow}|_{V'}$ in the case that $V'=V''$.
We refer the set $V$ of vertices (resp.~the set $\rightarrow$ of edges) of
a graph $G$ by $V_G$ (resp.~${\rightarrow}_G$).
For a graph $G$ and $V' \subseteq V_G$,
we refer $\< V_G, \mathord{\rightarrow}|_{V'}\>$ by $G|_{V'}$.
Hence, $G^\emptyset$ is a graph that has no edges but retains the same vertex set
as $G$.

For an ARS $G=\<V,{\rightarrow}\>$, an element $u$ of $V$ is
\emph{($\rightarrow$-)normal} if there exists no $v\in V$ such that
$u\rightarrow v$.
We sometimes call a normal element a \emph{normal form (of $G$)}
and use \emph{$\NF_{G}$} to represent the set of normal forms of $G$.
We define the set \emph{$\NF_{G}(u)$} as
$\{v\in \NF_{G} \mid u \rarrow{\ast}{} v\}$.
We say $v$ is a \emph{reduct} of $u$ if $u\rarrow{\ast}{}v$.

We follow the standard terminology in graph theory.
An edge $v \rightarrow u$ is an \emph{outgoing edge} of $v$ and
an \emph{incoming edge} of $u$, and $v$ is the \emph{initial vertex} of
$v \rightarrow u$.
A vertex $v$ is $\rightarrow$-normal, if and only if it has no outgoing-edges.
A $\rightarrow$-normal vertex is \emph{isolated}
if it also has no incoming edges.
We use $(u\to v)$ for a single edge graph $\<\{u,v\}, \{(u,v)\}\>$.

For graphs $G=\<V,\mathord{\rightarrow}\>$,
$G_1=\<V_1,\mathord{\rightarrow}_1\>$ and $G_2=\<V_2,\mathord{\rightarrow}_2\>$,
we define the following graph operations:
\begin{description}
\item [Union:] $G_1 \cup G_2 = \<V_1 \cup V_2, \mathord{\rightarrow}_1\cup\mathord{\rightarrow}_2\>$.
\item [Intersection:] $G_1 \cap G_2 = \<V_1 \cap V_2, \mathord{\rightarrow}_1\cap\mathord{\rightarrow}_2\>$.
\item [Difference:] $G_1 \setminus G_2 = \<V_1, {{\to}_1} \setminus {{\to}_2} \>$.
\item [Inclusion:] $G_2$ includes $G_1$ (denoted by $G_1 \subseteq G_2)$
  if $V_1 \subseteq V_2$ and ${\rightarrow}_1 \subseteq {\rightarrow}_2$.
\item [Vertex restriction:] $G|_{V'} = \<V\cap V', {\rightarrow}|_{V'}\>$.
\end{description}
We introduce the following edge operations on graph $G$:
\begin{itemize}
\item $G \setminus (u \to v)$ denotes an \emph{edge deletion}, defined as
  $(V_G, \mathord{\rarrow{}{G}} \setminus \{(u,v)\})$.
\item $G \cup (u\to v)$ denotes an \emph{edge expansion}, defined as
  $(V_G \cup \{u,v\}, \mathord{\rarrow{}{G}} \cup \{(u,v)\})$.
\item If a sequence $u = u_0 {\to} u_1 {\to} \cdots {\to} u_n = v$ is clear
  from the context, we denote
  \( G \cup (u_0 {\to} u_1) \cup \cdots \cup (u_{n-1} {\to} u_n) \)
  by $G \cup \lb u {\rarrow{\ast}{}} v \rb$.
\end{itemize}

\begin{definition}\label{def:graphiso}
  Two ARSs $G_1 = \<V_1, \mathord{\rightarrow}_1\>$ and
  $G_2 = \<V_2, \mathord{\rightarrow}_2\>$
  are \emph{isomorphic} (denoted by $G_1 \cong G_2$)
  if there exists an bijection $\varphi: V_1 \to V_2$ such that
  $(u,u') \in \mathord{\rightarrow}_1$ if and only if
  $(\varphi(u),\varphi(u')) \in \mathord{\rightarrow}_2$.
\end{definition}

An isomorphism preserves the following properties:
\begin{itemize}
\item $G$ is \emph{finite} if $V$ is finite.
\item $G$ is \emph{cyclic} if $v\rarrow{+}{G}v$ for some $v \in V_{G}$.
\item $G$ is \emph{deterministic} if $(u,v)\in\mathord{\rarrow{}{G}}$ and $(u,v')\in\mathord{\rarrow{}{G}}$
  imply $v = v'$.
\item $G$ is \emph{connected} if $(u,v)\in\mathord{\barrow{\ast}{G}}$
  for all $u,v\in V_{G}$.
\end{itemize}
We sometimes present an ARS visually as a graph;
for example,
$a \leftrightarrows b$ represents an ARS $\<\{a,b\},\{(a,b),(b,a)\}\>$

\begin{proposition}\label{Prop:FiniteApproxARS}
  Let $A$ and $B$ be ARSs. Assume that $A$ is finite and
  ${\barrow{\ast}{A}} \subseteq {\barrow{\ast}{B}}$.
  Then, there exists a finite ARS $C$ with $C \subseteq B$ such that
  ${\barrow{\ast}{A}} \subseteq {\barrow{\ast}{C}}$.
\end{proposition}

\begin{proof}
  Let us enumerate edges $(u,v)$ in $A$, which are finite in number.
  Since ${\barrow{\ast}{A}} \subseteq {\barrow{\ast}{B}}$, for each edge
  $(u,v)\in \to_A$ there exists
  a sequence $u = u_0 \barrow{}{B} u_1 \cdots \barrow{}{B} u_n = v$.
  By collecting all edges $(u_i, u_{i+1})$ appearing in these sequences
  for all edges in $A$, we define $C$ as the union of these edges.
  Since $A$ is finite and each sequence is finite, $C$ is a finite ARS satisfying
  ${\barrow{\ast}{A}} \subseteq {\barrow{\ast}{C}}$.

\end{proof}

Let $G=\< V,\rightarrow\>$ be an ARS.
We write $\downarrow$ for the joinability relation
${\rarrow{\ast}{}} \cdot {\larrow{\ast}{}}$.
We  say $G$ is \emph{confluent} if
$\mathord{\larrow{\ast}{}}\cdot \mathord{\rarrow{\ast}{}}
\subseteq {\downarrow}$,
\emph{Church-Rosser} (\emph{CR})
if $\mathord{\barrow{\ast}{}} \subseteq {\downarrow}$,
\emph{terminating} if it does not admit an infinite reduction sequence
from any element, and
\emph{looping} if every infinite reduction sequence contains multiple
occurrences of some element.
Note that confluence and CR are equivalent, and every finite ARS is
looping.

\subsection{Edge expansion as conservative extension}
\label{subsec:conserv-extension}

%
\begin{definition}\label{def:conserv}
  $G'=\<V',{\to}'\>$ \emph{conservatively extends} $G=\<V,{\to}\>$,
  denoted by $G \conserveq G'$, if $G\subseteq G'$ and
  ${\barrow{\ast}{}} = {\Barrowd{}{}{\ast}{}}|_{V\times V}$.
\end{definition}
Note that ${\barrow{\ast}{}} \subseteq {\Barrowd{}{}{\ast}{}}|_{V\times V}$
follows immediately from $G \subseteq G'$. Thus, it suffices to require
${\barrow{\ast}{}} \supseteq {\Barrowd{}{}{\ast}{}}|_{V\times V}$
instead of the equality condition in Definition~\ref{def:conserv}.

Next proposition holds immediately.
\begin{proposition} \label{prop:ConservUnion}
  Let $T,G,G'$ be ARSs with $G \conserveq G'$.
  If $V_T \subseteq V_G$, $T \cup G \conserveq T \cup G'$.
\end{proposition}

We remark that an edge expansion $G':=G \cup (v\to v')$ is a conservative
extension of $G$ except when $v,v'\in V_G$ and
$v \not\barrow{\ast}{G} v'$.
%

\begin{proposition}\label{prop:conserv2}
  Let $G_1$ be a finite ARS and $G_2$ be an ARS such that $G_1 \subseteq G_2$,
  then there exists a finite ARS $G$ such that $G_1\subseteq G \conserveq G_2$.
\end{proposition}

\begin{proof}
  Without loss of generality, we can assume that $G_1$ is contained in a single
  connected component with respect to $\barrow{\ast}{G_2}$ (otherwise, we can apply
  the following argument to each connected component of $G_2$ separately and then
  take the union of the resulting ARSs).
  We prove the claim by induction on the number $k$ of connected components of
  $G_1$ with respect to $\barrow{\ast}{G_1}$.
  Since $G_1$ is finite, $k$ is well-defined.

  If $k \le 1$, then $G_1 \conserveq G_2$ holds trivially, and the claim
  follows by setting $G := G_1$.

  Assume $k > 1$.
  Then, there must exist a edge sequence $\eta$ of edges in $G_2$ that connects a pair of vertices
  belonging to two different connected components of $G_1$.
  Let $G_1' := G_1 \cup \eta$.
  Clearly, $G_1'$ is a finite ARS satisfying $G_1 \subseteq G_1' \subseteq G_2$,
  and the number of connected components of $G_1'$ with respect to
  $\barrow{\ast}{G_1'}$ is strictly less than $k$.
  Therefore, by the induction hypothesis, there exists a finite ARS $G$ such that
  $G_1' \subseteq G \conserveq G_2$. Since $G_1 \subseteq G_1'$, this $G$ also
  satisfies $G_1 \subseteq G \conserveq G_2$.
\end{proof}

\begin{proposition}\label{prop:conserv}
  Let $G_1$, $G_2$ be ARSs with $G_1 \conserveq G_2$.
  For $G := G_1 \cup (v \rarrow{}{G_2} v')$, $G_1 \conserveq G$ holds.
  Moreover, if $v \in V_{G_1}$ or $v' \in V_{G_1}$ then $G \conserveq G_2$.
\end{proposition}

\begin{proof}
  The inclusions $G_1 \subseteq G \subseteq G_2$ follow immediately from the construction of $G$.
  Since $G_1 \conserveq G_2$, we have
  \begin{align}
    & {\barrow{\ast}{G_2}}|_{V_{G_1}\times V_{G_1}} = {\barrow{\ast}{G_1}}.
    \label{eq:conserv:assump1}
  \end{align}
  Then, ${\barrow{\ast}{G}}|_{V_{G_1}\times V_{G_1}}\subseteq {\barrow{\ast}{G_1}}$
  follows from ${\barrow{\ast}{G}} \subseteq {\barrow{\ast}{G_2}}$ and
  \eqref{eq:conserv:assump1}.
  Thus, $G_1 \conserveq G$ holds.

  Next, we show that $G \conserveq G_2$, i.e.,
  ${\barrow{\ast}{G_2}}|_{V_G\times V_G} \subseteq {\barrow{\ast}{G}}$.
  Let $(u,u') \in {\barrow{\ast}{G_2}}|_{V_G\times V_G}$.
  We consider the following three cases based on the locations of $u$ and $u'$.
  \begin{description}
  \item[\textbf{Case 1}: $u,u'\in V_{G_1}$.]  In this case,
    we have $(u,u') \in {\barrow{\ast}{G_2}}|_{V_{G_1}\times V_{G_1}}$.
    Thus, we obtain $(u,u') \in { \barrow{\ast}{G} }$ from
    \eqref{eq:conserv:assump1} and $G_1 \subseteq G$.
  \item[\textbf{Case 2}: Exactly one of $u$ of $u'$ is in $V_G\setminus V_{G_1}$.]
    Without loss of generality, we assume $u\in V_{G_1}$ and $u'\in V_G\setminus V_{G_1}$.
    From the construction of $G$, we have
    $V_{G} \setminus V_{G_1} \subseteq \{v,v'\}$, which implies $u'=v$ or $u'=v'$.
    \begin{itemize}
    \item If $u'=v$,
      we have $u \barrow{\ast}{G_2} v' \gets_G v = u'$.
      Since $u' \notin V_{G_1}$, the assumption $v \in V_{G_1}$ or
      $v' \in V_{G_1}$ ensures $v' \in V_{G_1}$.
      By \eqref{eq:conserv:assump1}, we obtain $(u,v') \in { \barrow{\ast}{G} }$,
      and hence, $(u,u') \in { \barrow{\ast}{G} }$.
    \item If $u'=v'$, we have $u \barrow{\ast}{G_2} v \to_G v' = u'$.
      Here, $v \in V_{G_1}$ from the assumption $v \in V_{G_1}$ or
      $v' \in V_{G_1}$.
      Thus, $(u,u') \in { \barrow{\ast}{G} }$ follows from
      \eqref{eq:conserv:assump1}.
    \end{itemize}
  \item[\textbf{Case 3}: $u,u' \in V_{G}\setminus V_{G_1}$.]
    If $u=u'$, the claim is trivial.
    Assume $u \neq u'$.  Since $u,u' \in V_{G_1}\setminus V_{G_1}$, we obtain $u,u' \in \{v,v'\}$
    from the construction of $G$.
    This contradicts the assumption $v \in V_{G_1}$ or $v' \in V_{G_1}$.
  \end{description}
\end{proof}

In Proposition~\ref{prop:conserv}, the condition $v \in V_{G_1}$ or
$v' \in V_{G_1}$ is necessary for $G \conserveq G_2$.
A counter example: $G_1=\{ a\to b \}$, $G_2= G_1 \cup \{b\to c,\ c\to d\}$,
and $G=G_1\cup\{c\to d\}$.
Then, $G_1 \conserveq G$ but $G \not\conserveq G_2$.

\subsection{Union of ARSs that preserves confluence}

First, we show a simple confluence preserving condition for union.

\begin{lemma}\label{lem:inc-confluence}
  Let $G_1$ be an ARS and $G_2$ be a confluent ARS such that
  ${\to_{G_1}} \subseteq {\barrow{\ast}{G_2}}$.
  Then, $G_1 \cup G_2$ is confluent.
\end{lemma}
\begin{proof}
  Suppose $u \barrow{\ast}{G_1\cup G_2} u'$.
  Since ${\to_{G_1}} \subseteq {\barrow{\ast}{G_2}}$,
  we have $u \barrow{\ast}{G_2} u'$.
  From the confluence of $G_2$, we obtain
  $u \rarrow{\ast}{G_2} u'' \larrow{\ast}{G_2} u'$ for some $u''\in V_{G_2}$.

\end{proof}

Next, we introduce the notion of hierarchical combination as another confluence
preserving condition for union, which will be used to guarantee confluence in
Section~\ref{sec:WS_system}.

\begin{definition}\label{def:hc}
  For ARSs $G_1$ and $G_2$,
  we say that $G_1 \cup G_2$
  is a \emph{hierarchical combination of $G_2$ with $G_1$},
  denoted by $G_1 \gtrdot G_2$,
  if $\mathord{\rightarrow_{G_1}} \subseteq (V_{G_1}\setminus
  V_{G_2})\times V_{G_1} \cup \mathord{\leftrightarrow^\ast_{G_2}}$.
\end{definition}

The following lemma is a simpler version of those~\cite{Sakai15}.

\begin{lemma} \label{lem:hc-confluence-union}
  Let $G_1, G_2$ be confluent ARSs such that
  $G_1 \gtrdot G_2$ is a hierarchical combination.
  Then, $G_1 \gtrdot G_2$ is confluent.
\end{lemma}

\begin{proof}
  We denote
  \[
  \left\{
  \begin{array}{lll}
    {\rightarrow_{1'}}  & \text{ \ for \ } &
    {\rightarrow_{G_1}}|_{(V_{G_1}\setminus V_{G_2}) \times V_{G_1}} \\
    {\rightarrow_{1''}} & \text{ \ for \ } &
    {\rightarrow_{G_1}}|_{(V_{G_1} \cap V_{G_2}) \times (V_{G_1} \cap V_{G_2})},\
    \text{ and}\\
    \mathord{\rightarrow_{2'}} & \text{ \ for \ } &
    {\rightarrow_{1''}} \cup \mathord{\rightarrow_{G_2}} \\
  \end{array}
  \right.
  \]

  First, we will show that ${\rarrow{\ast}{G_1 \cup G_2}} =
  \mathord{\rightarrow_{1'}^\ast} \cdot \mathord{\rightarrow_{2'}^\ast}$.
  Since $G_1 \gtrdot G_2$ is a hierarchical combination
  (i.e., ${\rightarrow_{G_1}} \subseteq
  (V_{G_1}\setminus V_{G_2})\times V_{G_1} \cup {\barrow{\ast}{G_2}}$),
  for any step $u \rarrow{}{G_1} v$, we obtain that $u \rarrow{}{1'} v$ if $u\in V_{G_1}\setminus V_{G_2}$,
  and $u \rarrow{}{1''} v$ if $u\in V_{G_2}$.
  Consider a sequence $s\rarrow{\ast}{G_1\cup G_2}t$.
  If it does not contain any element in $V_{G_2}$, then $s \rarrow{\ast}{1'} t$.
  Otherwise, letting the leftmost element in $V_{G_2}$ be $u$, we have $s\rarrow{\ast}{1'} u \rarrow{\ast}{2'} t$.
  Therefore, ${\rarrow{\ast}{G_1 \cup G_2}} =
  \mathord{\rightarrow_{1'}^\ast} \cdot \mathord{\rightarrow_{2'}^\ast}$.

  Now we will show that the confluence of $G = G_1 \cup G_2$.
  Assume $t \larrow{\ast}{G_1 \cup G_2} s \rarrow{\ast}{G_1 \cup G_2} u$.
  Then we have $t \larrow{\ast}{2'} t' \larrow{\ast}{1'} s
  \rarrow{\ast}{1'} u' \rarrow{\ast}{2'} u$.
  If $t' \not\in V_{G_2}$ (resp. $u' \not\in V_{G_2}$), $t = t'$ (resp. $u = u'$).
  \begin{itemize}
  \item If $t', u' \not\in V_{G_2}$,
    $t = t' \downarrow_{G_1} u' = u$.
  \item If $t' \not\in V_{G_2}$ and $u' \in V_{G_2}$, the confluence of $G_1$
    implies $t = t' \rarrow{\ast}{G_1} s' \larrow{\ast}{G_1} u'$ for some $s'$.
    Since $u' \in V_{G_2}$, we have
    $s' \larrow{\ast}{1''} u' \rarrow{\ast}{2'} u$.
    Combining this with
    ${\rarrow{}{1''}} \subseteq {\rarrow{}{2'}} \subseteq {\barrow{\ast}{2}}$,
    the confluence of $G_2$ impiles $s' \downarrow_{G_2} u$.
    Hence $t \rarrow{\ast}{G_1} s' \downarrow_{G_2} u$.
  \item If $t', u' \in V_{G_2}$. The confluence of $G_1$ implies
    $t' \downarrow_{G_1} u'$.
    Furthermore, $t', u' \in V_{G_2}$ implies
    $t \larrow{\ast}{2'} t' \downarrow_{1''} u' \rarrow{\ast}{2'}$.
    Combining this with
    ${\rarrow{}{1''}} \subseteq {\rarrow{}{2'}} \subseteq {\barrow{\ast}{2}}$,
    the confluence of $G_2$ implies $t \downarrow_{G_2} u$.
  \end{itemize}
\end{proof}

\subsection{Quotient ARS}

This subsection introduces quotient ARSs of an ARS $G$
with respect to the connectivity $\barrow{\ast}{G}$ and the strong connectivity ${\rarrow{\ast}{G}}\cap {\larrow{\ast}{G}}$,
which is denoted by $\stcongstarss{G}$.

\begin{definition} \label{def:quotientARS}
  For ARSs $T$ and $G$ with $V_T \subseteq V_G$,
  let $\sim$ be an equivalence relation on $V_G$.
  We define the equivalence class of $u \in V_G$ with respect to $\sim$
  \[
  \class{u}_{\sim} = \{ v \in V_G \mid u \sim v \}
  \]
  A \emph{quotient ARS} $T/{\sim}$ is defined as:
  \[
  \left\{
  \begin{array}{l}
  V_{T/{\sim}} = \{ \class{u}_{\sim} \mid u\in V_G \}, \\
  U_1 \rarrow{}{T/{\sim}} U_2\,\, \Leftrightarrow \,\,
  U_1 \neq U_2 \text{ and } u_1\rarrow{}{T} u_2
  \text{ for some $u_1\in U_1$ and $u_2\in U_2$.}
  \end{array}
  \right.
  \]
  For $U \subseteq V_G$, $\class{U}_\sim = \bigcup_{u \in U}\,\class{u}_\sim$.
\end{definition}

We prepare notation for the equivalences of the connectivity
$\barrow{\ast}{G}$ and the strong connectivity
${\stcongstarss{G}}$.
For $u \in V_G$ and an ARS $T$ with $V_T \subseteq V_G$,
\begin{itemize}
\item Their equivalent classes are denoted by
  $\classB{u}_{\barrow{\ast}{G}}$ and $\classB{u}_{\stcongstarss{G}}$, respectively.
  (If clear from the context, we use $\classB{u}$ and $\classD{u}$ for them.)
\item The quotient ARSs $T/{\barrow{\ast}{G}}$ and $T/{\stcongstarss{G}}$ are denoted by $T/G$ and $T \ddiv G$, respectively.
  Moreover, $G \ddiv G$, i.e., $G/{\stcongstarss{G}}$ is denoted by $G_{\ddiv}$.
\end{itemize}
\begin{definition}
  A quotient ARS $T/{\sim}$ is \emph{strongly deterministic}
  if the number of outgoing edges in $T$ from each equivalence class is at most one,
  i.e.,
  $u' \gets_{T} u \sim v \to_{T} v'$ implies that $u=v$ and $u'=v'$.
\end{definition}

The following proposition is trivial.
\begin{proposition}\label{prop:strong-lynondeterministic}
  A strongly deterministic quotient $T/{\sim}$ is deterministic.
\end{proposition}

\medskip
The next two propositions show the property of connectivity.

\begin{proposition}\label{Prop:NestedQuotient}
  $T_2 / (T_1 \cup G) \cong (T_2 / G) / (T_1 / G)$ for ARSs $G, T_1,
  T_2$ with $V_{T_1 \cup T_2} \subseteq V_G$,
\end{proposition}

\begin{proof}
  Let $\varphi: V_{T_2/(T_1 \cup G)} \to V_{(T_2/G)/(T_1/G)}$ be
  $\varphi(\classB{u}_{\barrow{\ast}{T_1 \cup G}}) =
  \classB{\classB{u}_{\barrow{\ast}{G}}}_{\barrow{\ast}{T_1 / G}}$ for $u \in V_G$.
  The mapping is well-defined.

  If $\varphi(\classB{u}_{\barrow{\ast}{T_1\cup G}}) =
  \varphi(\classB{v}_{\barrow{\ast}{T_1\cup G}})$,
  then $\classB{\classB{u}_{\barrow{\ast}{G}}}_{\barrow{\ast}{T_1 / G}} =
  \classB{\classB{v}_{\barrow{\ast}{G}}}_{\barrow{\ast}{T_1 / G}}$, thus
  $\classB{u}_{\barrow{\ast}{G}} \barrow{\ast}{T_1/G} \classB{u}_{\barrow{\ast}{G}}$,
  which means
  $u\mathrel{(\barrow{\ast}{G}\cdot\barrow{\ast}{T_1}\cdot\barrow{\ast}{G})^\ast}v$.

  Hence $u \barrow{\ast}{T_1 \cup G} v$ and
  $\classB{u}_{\barrow{\ast}{T_1\cup G}} = \classB{v}_{\barrow{\ast}{T_1\cup G}}$ follows.
  Therefore, $\varphi$ is injective. Obviously, $\varphi$ is surjective.

  Let $(U, U') \in \mathord{\rightarrow}_{T_2/(T_1 \cup G)}$.
  There exists $(u,u') \in \mathord{\rightarrow}_{T_2}$ with
  $U = \classB{u}_{\barrow{\ast}{T_1\cup G}}$ and
  $U' = \classB{u'}_{\barrow{\ast}{T_1\cup G}}$.
  Thus, $\varphi(U) = \classB{\classB{u}_{\barrow{\ast}{G}}}_{\barrow{\ast}{T_1/G}}$
  and $\varphi(U') = \classB{\classB{u'}_{\barrow{\ast}{G}}}_{\barrow{\ast}{T_1/G}}$.
  Moreover, $\classB{u}_{\barrow{\ast}{G}} \to_{T_2/G} \classB{u'}_{\barrow{\ast}{G}}$
  from $u \to_{T_2} u'$.
  Hence, we have
  $(\varphi(U),\varphi(U')) =
  (\classB{\classB{u}_{\barrow{\ast}{G}}}_{\barrow{\ast}{T_1/G}},
   \classB{\classB{u'}_{\barrow{\ast}{G}}}_{\barrow{\ast}{T_1/G}})
  \in \mathord{\rightarrow}_{(T_2/G)/(T_1/G)}$.
  The reverse direction is similar.

\end{proof}

\begin{proposition}\label{Prop:CycleInDeterministic}
  Let $G_1, G_2$ be ARSs with $\rarrow{}{G_1} \cap \rarrow{}{G_2} = \emptyset$
  and $V_{G_2} \subseteq V_{G_1}$.
  If $G_1 \cup G_2$ is deterministic,
  then the acyclicity of $G_1 \cup G_2$ implies that of $G_2/G_1$.
\end{proposition}

\begin{proof}
  Suppose that there is a cycle
  $\classB{u_0} \rarrow{}{G_2/G_1} \classB{u_1} \rarrow{}{G_2/G_1}
  \cdots \rarrow{}{G_2/G_1} \classB{u_n} = \classB{u_0}$
  in $G_2/G_2$ with $n > 0$.
  Then $u_0 \barrow{\ast}{G_1}\cdot\rarrow{}{G_2}\cdot\barrow{\ast}{G_1} u_1
  \barrow{\ast}{G_1}\cdot\rarrow{}{G_2}\cdot\barrow{\ast}{G_1} \cdots
  \barrow{\ast}{G_1}\cdot\rarrow{}{G_2}\cdot\barrow{\ast}{G_1} u_n
  \barrow{\ast}{G_1} u_0$.
  Since $G_1$ is confluent from its deterministic property,
  we can replace $\barrow{\ast}{G_1}$ with
  $\rarrow{\ast}{G_1} \cdot \larrow{\ast}{G_1}$, which results in
  \[
  u_0 \rarrow{\ast}{G_1}\cdot\larrow{\ast}{G_1}\cdot\rarrow{}{G_2}
  \cdots
  \rarrow{\ast}{G_1}\cdot\larrow{\ast}{G_1}\cdot\rarrow{}{G_2} u_n
  \rarrow{\ast}{G_1} v \larrow{\ast}{G_1} u_0 \]
  Hence,
  \[
  v \larrow{\ast}{G_1} u_0
  \rarrow{\ast}{G_1}\cdot\larrow{\ast}{G_1}\cdot\rarrow{}{G_2}
  \cdots
  \rarrow{\ast}{G_1}\cdot\larrow{\ast}{G_1}\cdot\rarrow{}{G_2} u_n
  \rarrow{\ast}{G_1} v
  \]
  Since $G_1 \cup G_2$ is deterministic and
  $\rarrow{}{G_1} \cap \rarrow{}{G_2} = \emptyset$,
  $\larrow{+}{G_1}\cdot\rarrow{}{G_2}$ cannot occur.
  Therefore, the sequence above is
  \[
  v \rarrow{\ast}{G_1}\cdot\rarrow{}{G_2}
  \cdots
  \rarrow{\ast}{G_1}\cdot\rarrow{}{G_2} u_n
  \rarrow{\ast}{G_1} v
  \]
  which is a cycle in $G_1 \cup G_2$,
  contradicting the acyclicity of $G_1 \cup G_2$.

\end{proof}

\medskip
Next two porpositions show the properties of the strong connectivity.

\begin{proposition}\label{prop:bridge-equiv}
  For ARSs $T$, $G$ with $V_T \subseteq V_G$ and
  the quotient ARS $T \ddiv G$,
  if $U\rarrow{\ast}{T \ddiv G} U'$ for $U,U'\in V_{T \ddiv G}$,
  then $u \rarrow{\ast}{T \cup G} u'$ for any $u\in U$ and $u'\in U'$.
\end{proposition}
\begin{proof}
  We prove it by induction on the number $n$ of steps of $U\rarrow{\ast}{T \ddiv G} U'$.
  Let $u\in U$ and $u'\in U'$.
  For the case $n=0$, we have $U = U'$, and thus, $u \rarrow{\ast}{G} u'$
  from $u \stcongstarss{G} u'$.
  Consider the case that $n>0$, and let $u\in U$ and $u'\in U'$.
  Then we can write $U\to_{T \ddiv G} U'' \rarrow{\ast}{T \ddiv G} U'$
  for some $U'' \in V_{T \ddiv G}$.
  From $U\to_{T \ddiv G} U''$, there exist some $v\in U$ and $v''\in U''$
  such that $v \to_T v''$.
  Combining this with $u\rarrow{\ast}{G}v$ induced by $u \stcongstarss{G} v$,
  we obtain $u\rarrow{\ast}{T \cup G}v''$.
  On the other hand, $v''\rarrow{\ast}{T \cup G}u'$ by induction hypothesis.
  Therefore, $u\rarrow{\ast}{T \cup G}u'$.

\end{proof}

\begin{proposition}\label{prop:termination-looping}
  For ARSs $T$ and $G$ with $V_T \subseteq V_G$,
  \begin{enumerate}
  \item\label{prop:termination-looping-1}
    $T \ddiv G$ is confluent if $T \cup G$ is confluent,
  \item\label{prop:termination-looping-2}
    $T \ddiv G$ is confluent only if peaks in $T$ close in $T \cup G$
    (i.e., $\barrow{\ast}{T}\, \subseteq\, \downarrow_{T \cup G}$), and
  \item\label{prop:termination-looping-3}
    $G \ddiv G$ is terminating if $G$ is looping.
  \end{enumerate}
\end{proposition}

\begin{proof}
  \mbox{}
\begin{enumerate}
\item
  Suppose $U_2 \larrow{\ast}{T \ddiv G} U_1 \rarrow{\ast}{T \ddiv G} U_3$.
  Let $u_i \in U_i$ ($i=1,2,3$).
  Then $u_2 \larrow{\ast}{T \cup G} u_1 \rarrow{\ast}{T \cup G} u_3$
  by Proposition~\ref{prop:bridge-equiv}.
  From confluence of $T \cup G$, there exists $u_4$ such that
  $u_2 \rarrow{\ast}{T \cup G} u_4 \larrow{\ast}{T \cup G} u_3$.
  Therefore,
  $U_2\rarrow{\ast}{T \ddiv G}\classD{u_4}\larrow{\ast}{T \ddiv G} U_3$.

\item
  Suppose $u_2 \larrow{\ast}{T} u_1 \rarrow{\ast}{T} u_3$.
  Then, $\classD{u_2} \larrow{\ast}{T \ddiv G} \classD{u_1}
  \rarrow{\ast}{T \ddiv G} \classD{u_3}$.
  From confluence of $T \ddiv G$, we obtain
  $\classD{u_2}\rarrow{\ast}{T \ddiv G} U_4\larrow{\ast}{T \ddiv G}\classD{u_3}$
  for some $U_4 \in V_{T \ddiv G}$.
  Hence, $u_2 \rarrow{\ast}{T \cup G} u_4 \larrow{\ast}{T \cup G} u_3$

\item
  Suppose $G \ddiv G$ is not terminating.
  Then we have an infinite sequence
  $U_1 \to_{G \ddiv G} U_2 \to_{G \ddiv G} \cdots$.
  If $U_i = U_j$ for $i < j$ occurs, $U_i \cup U_{i+1} \cup \cdots \cup U_j$
  is a $\stcongstarss{G}$-class, and coincides with $U_i$, and violates that
  $U_i \to_{G \ddiv G} U_{i+1}$ requires $U_i \neq U_{i+1}$.
  Therefore, $U_i \neq U_j$ for $i < j$.
  Hence, by Proposition~\ref{prop:bridge-equiv} (with $T := G$),
  we obtain an infinite sequence $u_1 \to_{G} u_2 \to_{G} \cdots$
  with $u_j \mathrel{\not\rarrow{\ast}{G}} u_i$ for each $i < j$
  (since $u_j \rarrow{\ast}{G} u_i$ and $u_i \rarrow{\ast}{G} u_j$
  would give $u_i \stcongstarss{G} u_j$, contradicting $U_i \neq U_j$).
  This contradicts the assumption that $G$ is looping.
\end{enumerate}
\end{proof}

Note that
\begin{itemize}
\item The confluence of $G_{\ddiv}$ is equivalent to that of $G$.
\item A finite ARS $G$ is always looping. Hence, $G_{\ddiv}$ is terminating.
\end{itemize}

\section{Edge expansion of finite ARS to make it confluent} \label{sec4} \label{sec:conversion_ARS}

In this section, we investigate how to obtain and extend a finite ARS while preserving the
confluence.
We introduce an edge expansion by adding a new edge and show its properties.
We also introduce a \emph{hierarchical combination}~\cite{Sakai15} for the
confluence-preserving union of ARSs.

\subsection{Pseudo-normal form and confluence of ARS}\label{subsec:looping}

\begin{definition}\label{def:strongly-convergent}
  Let $G$ be an ARS and $u, v \in V_G$.
  \begin{itemize}
  \item \emph{$u$ is convergent to $v$},
    denoted by $u \Rrightarrow_{G} v,$
    if $u\rarrow{\ast}{G} w \implies w\rarrow{\ast}{G} v$ for any $w \in V_G$.
  \item \emph{The set of pseudo-normal forms is}
   $\pnf_{G} = \{u \in V_G \mid u \Rrightarrow_{G} u \}$.
  \end{itemize}
\end{definition}

Intuitively, $u \Rrightarrow_G v$ means that $u$ is confluent, $u \rarrow{\ast}{G} v$, and $v \in \pnf_G$.
Here, the confluence of $u$ is defined as $v_1\gets_G u \to_R v_2 \implies v_1 \to_R \cdot \gets_R v_2$
for any $v_1, v_2 \in V_G$.
The next proposition shows the relationships among $\pnf_G$,
$\Rrightarrow_G$, and $\classD{u}_{\stcongstarss{G}}$.

\begin{proposition} \label{prop:Gddiv}
  Suppose $u, v \in V_G$.
  \begin{enumerate}
  \item \label{subprop:Gddiv-1}
    The following statements are equivalent:
    \begin{enumerate} \renewcommand{\labelenumii}{(\roman{enumii})}
      \item $u \in \pnf_G$,
      \item $u \Rrightarrow_G u$, and
      \item $\classD{u}_{\stcongstarss{G}} \in \NF_{G_{\ddiv}}$.
    \end{enumerate}
    \item \label{subprop:Gddiv-2}
    $u \rarrow{\ast}{G} v$ and $u \in \pnf_G$ if and only if
    $\classD{u}_{\stcongstarss{G}} = \classD{v}_{\stcongstarss{G}} \in \NF_{G_{\ddiv}}$.
    \item \label{subprop:Gddiv-3}
    $u \Rrightarrow_G v$ and $v \Rrightarrow_G u$ if and only
    if $\classD{u}_{\stcongstarss{G}} = \classD{v}_{\stcongstarss{G}} \in\NF_{G_{\ddiv}}$.
  \end{enumerate}
\end{proposition}

\begin{proposition}\label{prop:looping_pnf}
  Let $G$ be a looping ARS.
  For any $t\in V_G$, there exists $s\in\pnf_G$ such that $t\rarrow{\ast}{} s$.
\end{proposition}
\begin{proof}
  There exists $U \in V_{G_{\ddiv}}$ such that
  $[t]_{\stcongstarss{G}} \rarrow{\ast}{G_{\ddiv}} U \in \NF_{G_{\ddiv}}$ from
  Proposition~\ref{prop:termination-looping}\eqref{prop:termination-looping-3}.
  By taking some $s \in U$, we have $t\rarrow{\ast}{G}s \in \pnf_{G}$ by
  Proposition~\ref{prop:bridge-equiv} and the definition of pseudo normal
  forms.
\end{proof}

\begin{proposition}\label{prop:pnf}
  For an ARS $G$ and $u,v \in V_G$,
  \begin{enumerate}
  \item \label{subprop:pnf-1}
    $u\Rrightarrow_G v$ implies $u \rarrow{\ast}{G} v$ and $v \Rrightarrow_G v$.
  \item \label{subprop:pnf-2}
    For each $w \in V_G$,
    $u \Rrightarrow_G v \rarrow{\ast}{G} w$ implies $u \Rrightarrow_G w$.
  \item \label{subprop:pnf-3}
    If $G$ is confluent, then for each $w \in V_G$,
    $w \rarrow{\ast}{G} u \Rrightarrow_G v$ implies $w \Rrightarrow_G v$.
    Moreover, when $G$ is looping, the converse holds,
    i.e., if this implication holds for all $w \in V_G$
    then $G$ is confluent.
  \end{enumerate}
\end{proposition}

\begin{proof}
  Since \eqref{subprop:pnf-1} and \eqref{subprop:pnf-2} are immediate,
  we consider \eqref{subprop:pnf-3}.

  Suppose first that $G$ is confluent.
  Let $w' \in V_G$ with $w \rarrow{\ast}{G} w'$.
  Since $G$ is confluent, there is $w'' \in V_G$ such that
  $w' \rarrow{\ast}{G} w'' \larrow{\ast}{G} u$.
  Since $u \Rrightarrow_G v$, we have $w'' \rarrow{\ast}{G} v$.
  Hence $w \Rrightarrow_G v$.

  To prove the converse,
  assume that the implication holds for all $w \in V_G$.
  Suppose $w \rarrow{\ast}{G} u$ and $w \rarrow{\ast}{G} u'$.
  Since $G$ is looping, there exist $v, v' \in \pnf_G$ such that
  $u \rarrow{\ast}{G} v$ and $u' \rarrow{\ast}{G} v'$
  by Proposition~\ref{prop:looping_pnf}.
  Because $v \in \pnf_G$, we have $v \Rrightarrow_G v$ by
  Proposition~\ref{prop:Gddiv}~(\ref{subprop:Gddiv-1}).
  This, together with $u \rarrow{\ast}{G} v$,
  gives $u \Rrightarrow_G v$ (by Proposition~\ref{prop:pnf}~(\ref{subprop:pnf-1})).
  Using $w \rarrow{\ast}{G} u$ and $u \Rrightarrow_G v$,
  the assumption implies $w \Rrightarrow_G v$.
  Since we have $w \rarrow{\ast}{G} u' \rarrow{\ast}{G} v'$,
  the relation $w \Rrightarrow_G v$ implies $v' \rarrow{\ast}{G} v$.
  Consequently, we obtain $u' \rarrow{\ast}{G} v' \rarrow{\ast}{G} v$.
  Since $u \rarrow{\ast}{G} v$ also holds, $v$ is indeed a common reduct of $u$ and $u'$.

  \suppress{
  For (4) {\bf if} part, $w \barrow{\ast}{G} u$ implies
  $w,u \rarrow{\ast}{G} w'$ for some $w'$.
  Hence $w \barrow{\ast}{G} u \Rrightarrow_G v$ implies
  $w' \larrow{\ast}{G} u \Rrightarrow_G v$. Thus $w' \rarrow{\ast}{G} v$.
  Thus $v \Rrightarrow_G v$ from (2) and $w' \Rrightarrow_G v$
  and $w \Rrightarrow_G v$ from (3).

  {\bf Only-if} part when $G$ is looping is from (3)
  since $w \barrow{\ast}{G} u$ includes $w \rarrow{\ast}{G} u$.
  }
\end{proof}

\suppress{
\begin{proposition}\label{prop:pnfX}
  \begin{enumerate}
  For an ARS $G$,
  \item\label{subprop:pnfX1}
    $u\in \pnf_{G}$ if and only if $u \Rrightarrow_{G} u$.
  \item\label{subprop:pnfX2}  $\Rrightarrow_{G}$ is transitive.
  \item\label{subprop:pnfX3}
    If $u\rarrow{\ast}{G}u'$ and $u\in\pnf_{G}$,
    then $u' \in [\![ u ]\!]_{\stcongstarss{G}}$ and thus $u'\in\pnf_{G}$.
  \item\label{subprop:pnfX4}
    If $u' \Rrightarrow_G u$, then $u'\rarrow{\ast}{G} u$ and $u \in \pnf_{G}$.
  \end{enumerate}
\end{proposition}

\begin{proof}
  Except for the last statement, the others are immediate.

  For (\ref{subprop:pnf4}), assume $u' \Rrightarrow_G u$.
  Then $u' \rarrow{\ast}{G} u'$ implies $u' \rarrow{\ast}{G} u$.
  Let $u \rarrow{\ast}{G} w$.
  Since $u' \rarrow{\ast}{G} u \rarrow{\ast}{G} w$ and $u' \Rrightarrow_G u$,
  it holds that $w \rarrow{\ast}{G} u$.
  Thus $u \Rrightarrow_G u$, i.e., $u \in \pnf_{G}$ by (\ref{subprop:pnf1}) of
  this proposition.
\end{proof}

Note that the converse of Proposition~\ref{prop:pnf}(\ref{subprop:pnf4}) does not hold for
non-confluent ARSs. For example, consider an ARS $u\gets u'\to u''$.
}


We deal with confluent ARSs, and clarify the conditions for an edge expansion
that preserves confluence, which is a generalization of confluent expansions of
ARSs in~\cite{Sakai10} that assume termination.
\medskip

\begin{lemma}\label{lem:expansion-PNF}
  For an edge expansion $G' := G \cup (v\to v')$ of a confluent ARS $G$,
  the ARS $G'$ is confluent if and only if
  one of the following holds: (1) $v \barrow{\ast}{G} v'$, (2) $v \notin V_G$, or
  (3) $v \in \pnf_{G}$.
\end{lemma}
\begin{proof}
  We first show the {\bf if}-part.
  (1) is immediate since ${\barrow{\ast}{G}} = {\barrow{\ast}{G'}}$.

  For (2), $v \notin V_G$ implies that no edge in $G'$ enters $v$, and $v \rarrow{}{G'} v'$ is the unique edge leaving $v$.
  Hence, any divergence of $G'$ starting from $u \in V_G$ never visits $v$, which means it joins by the confluence of $G$.
  Furthermore, any divergence of $G'$ starting from $v$ trivially reduces to one starting from $v'$.

  For (3), if $v' \notin V_G$, no edge leaves $v'$ in $G'$, and $v \rarrow{}{G'} v'$ is the unique edge entering $v'$.
  Thus, any divergence of $G'$ containing $v'$ can be written as
  $u_1 \larrow{\ast}{G} u \rarrow{\ast}{G} v \rarrow{}{G'} v'$.
  By the confluence of $G$ and $v \in \pnf_G$, we have $u_1 \rarrow{\ast}{G} v$,
  and hence $u_1$ and $v'$ join in $G'$.
  Next, we assume $v' \in V_G$ and let $u_1 \larrow{\ast}{G'} u \rarrow{\ast}{G'} u_2$
  be a diverging sequence. We have three cases according to the occurrences of
  $v \rarrow{}{G'} v'$ in it.
  \begin{itemize}
  \item If $(v,v')$ does not occur,
    $u_1 \larrow{\ast}{G} u \rarrow{\ast}{G} u_2$ and $u_1 \downarrow_{G} u_2$
    from confluence of $G$.
  \item If $(v,v')$ occurs both sides in
    $u_1 \larrow{\ast}{G'} u \rarrow{\ast}{G'} u_2$,
    then $u_1 \larrow{\ast}{G} v' \larrow{\ast}{G'} v \larrow{\ast}{G'} u
    \rarrow{\ast}{G'} v \rarrow{\ast}{G'} v' \rarrow{\ast}{G} u_2$.
    Hence, $u_1 \larrow{\ast}{G} v' \rarrow{\ast}{G} u_2$ and
    $u_1 \downarrow_{G} u_2$ from confluence of $G$.
  \item If $(v,v')$ occurs only on one side, say
    $u_1 \larrow{\ast}{G} u \rarrow{\ast}{G}
    v \rarrow{\ast}{G'} v' \rarrow{\ast}{G'} u_2$.
    Since $G$ is confluent and $v\in \pnf_G$,
    $u \rarrow{\ast}{G} v \Rrightarrow_G v$ implies $u \Rrightarrow_G v$
    by Proposition~\ref{prop:pnf}~(\ref{subprop:pnf-3}).
    Combining $u \Rrightarrow_G v$ and $u \rarrow{\ast}{G} u_1$, we have $u_1 \rarrow{\ast}{G} v$,
    and hence, $u_1 \rarrow{\ast}{G'} u_2$.
  \end{itemize}
  Therefore, in each case, we obtain $u_1 \downarrow_{G'} u_2$.
  \smallskip

  For the {\bf only-if} part, assume that $G'$ is confluent.
  If $v \notin V_G$, then condition (2) holds trivially;
  thus, it remains to consider the case where $v \in V_G$.

  If $v' \notin V_G$,
  we will show $v \Rrightarrow_G v$, which implies (3) $v \in \pnf_G$ by
  Proposition~\ref{prop:Gddiv}~(\ref{subprop:Gddiv-1}).
  Suppose $v \rarrow{\ast}{G} u$, then $u \larrow{\ast}{G'} v \rarrow{}{G'} v'$,
  and hence $u \downarrow_{G'} v'$ by the confluence of $G'$.
  Since $v'$ is reachable only via $v \rarrow{}{G'} v'$, we have $u \rarrow{\ast}{G} v$.
  Hence, $v \Rrightarrow_G v$.

  If $v' \in V_G$, we will show that (2) $v \notin\pnf_G$ implies (3) $v \barrow{\ast}{G} v'$.
  Since $v \notin\pnf_G$, there exists $v''$ such that $v \rarrow{\ast}{G} v''$ and
  $v'' \not\rarrow{\ast}{G} v$ by Proposition~\ref{prop:Gddiv}~(\ref{subprop:Gddiv-1}).
  By $v'' \larrow{\ast}{G} v \rarrow{}{G'} v'$ and the confluence of $G'$,
  there exists $w$ such that $v'' \rarrow{\ast}{G'} w \larrow{\ast}{G'} v'$.
  Now, we will show $v'' \rarrow{\ast}{G} w \larrow{\ast}{G} v'$,
  then we are done since $v \rarrow{\ast}{G} v''$.
  Since $v'' \not\rarrow{\ast}{G} v$, the path $v'' \rarrow{\ast}{G'} w$
  cannot use the edge $v \rarrow{}{G'} v'$, which gives $v'' \rarrow{\ast}{G} w$.
  Furthermore, extracting the suffix of $v' \rarrow{\ast}{G'} w$
  after the last occurrence of $v \rarrow{}{G'} v'$ (or the entire path if none exists)
  gives $v' \rarrow{\ast}{G} w$.
\end{proof}

\subsection{Pseudo-normal forms after edge expansion}
\label{subsec:Conv-Exp-ARS}

We clarify how edge expansions affect confluence by analyzing transitions of
quotient graphs $\NF_{G_{\ddiv}}$.

\begin{lemma} \label{lem:exp-class}
  Let an edge expansion $G' := G \cup (v\to v')$ of an ARS $G$.
  \begin{enumerate}
  \item\label{sublem:exp-class_1}
    Let $u\in\pnf_{G}$.
    Then, $u\in\pnf_{G'}$ if and only if
    $u \not\rarrow{\ast}{G} v \lor v' \Rrightarrow_{G} v$.
  \item\label{sublem:exp-class_2}
    Let $u\in V_G \setminus \pnf_{G}$.
    Then, $u\in\pnf_{G'}$ if and only if
    $v' \rarrow{\ast}{G} u \wedge v' \Rrightarrow_{G} v$.
  \item\label{sublem:exp-class_3}
    Let $u\notin V_{G}$.
    Then, $u\in\pnf_{G'}$ if and only if $u = v' \notin V_G$.
  \end{enumerate}
\end{lemma}

\begin{proof}
  Let $G' := G \cup (v\to v')$.
  \begin{enumerate}
  \item (Only-if): Assuming $u \rarrow{\ast}{G} v$,
    we show $v' \Rrightarrow_G v$.
    Let $v' \rarrow{\ast}{G} u'$.
    Then, $u \rarrow{\ast}{G'} u'$ as
    $u \rarrow{\ast}{G} v \to_{G'} v' \rarrow{\ast}{G} u'$.
    Since $u \in \pnf_{G'}$, $u' \rarrow{\ast}{G'} u$.
    Hence, we can show $u' \rarrow{\ast}{G} v$ regardless of the occurence of
    $v \to v'$ as follows.
    If $u' \rarrow{\ast}{G'} u$ contains $v \to v'$,
    we have $u' \rarrow{\ast}{G} v$ at its first occurrence,
    otherwise $u' \rarrow{\ast}{G} u \rarrow{\ast}{G} v$.
    \smallskip

    \noindent
    (If): We show $u \in \pnf_{G'}$ for either of cases.
    \begin{itemize}
    \item Let $u \not\rarrow{\ast}{G} v$.
      If $u \rarrow{\ast}{G'} u'$,
      the sequence does not contain the introduced edge $v\to v'$, and thus
      it must be $u \rarrow{\ast}{G} u'$.
      Since $u \in \pnf_{G}$, $u' \rarrow{\ast}{G} u$.

    \item Let $v' \Rrightarrow_G v$. Assume $u \rarrow{\ast}{G'} u'$.

      If it contains $v \to v'$,
      we have $u \rarrow{\ast}{G} v \rarrow{}{G'} v' \rarrow{\ast}{G} u'$,
      which is obtained by connecting the first and the last occurrences
      even if it contains multiple $v \rarrow{}{G'} v'$ steps.
      Here, $u' \rarrow{\ast}{G} v$ from $v' \rarrow{\ast}{G} u'$ and
      $v' \Rrightarrow_G v$, and also $v \rarrow{\ast}{G} u$
      from $u \rarrow{\ast}{G} v$ and $u \in \pnf_G$.
      Hence $u' \rarrow{\ast}{G} v \rarrow{\ast}{G} u$.

      Otherwise, $u \rarrow{\ast}{G} u'$.
      Combining this with $u \in \pnf_G$, we obtain $u' \rarrow{\ast}{G} u$.
    \end{itemize}

  \item (Only-if):
    Since $u \in V_G \setminus \pnf_{G}$, we have $u \rarrow{\ast}{G} u' \not\rarrow{\ast}{G} u$ for some $u'$.
    Since $u \in \pnf_{G'}$, we have $u' \rarrow{\ast}{G'} u$, which must contain the augmented edge $v \to v'$.
    Thus, this sequence can be decomposed into $u' \rarrow{\ast}{G} v \rarrow{\ast}{G'} v' \rarrow{\ast}{G} u$.
    By combining this with $u \rarrow{\ast}{G} u'$, we obtain a sequence:
    \begin{align}
      & u \rarrow{\ast}{G} v \rarrow{}{G'} v' \rarrow{\ast}{G} u.
      \label{eq:vdTov}
    \end{align}
    From this sequence~\eqref{eq:vdTov}, we immediately obtain the first required property:
    $v' \rarrow{\ast}{G} u$.

    We show the remaining property: $v' \Rrightarrow_{G} v$.
    Let $v' \rarrow{\ast}{G} u''$.
    Then, we have $u \rarrow{\ast}{G'} v' \rarrow{\ast}{G'} u''$
    from \eqref{eq:vdTov}.
    Thus $u \in \pnf_{G'}$ yeilds $u'' \rarrow{\ast}{G'} u$.
    If it contains $v \rarrow{}{G'} v'$, we obtain $u'' \rarrow{\ast}{G} v$.
    Otherwise, $u'' \rarrow{\ast}{G} u$.
    Combining this with $u \rarrow{\ast}{G} v$, we have $u'' \rarrow{\ast}{G} v$.
    \smallskip

    \noindent
    (If): Assume $v' \rarrow{\ast}{G} u$ and $v' \Rrightarrow_{G} v$.
    To show $u \in \pnf_{G'}$, let $u'$ be any vertex such that $u \rarrow{\ast}{G'} u'$.
    We first show that $v' \rarrow{\ast}{G} u'$ holds.
    If the sequence $u \rarrow{\ast}{G'} u'$ contains the augmented edge $v \to v'$,
    extracting the suffix after its last occurrence induces $v' \rarrow{\ast}{G} u'$.
    Otherwise, the sequence lies entirely in $G$,
    and combining it with the assumption $v' \rarrow{\ast}{G} u$
    gives $v' \rarrow{\ast}{G} u \rarrow{\ast}{G} u'$.
    In either case, we obtain $v' \rarrow{\ast}{G} u'$.

    Since $v' \Rrightarrow_{G} v$, this sequence $v' \rarrow{\ast}{G} u'$ implies $u' \rarrow{\ast}{G} v$.
    Thus, we obtain $u' \rarrow{\ast}{G} v \rarrow{}{G'} v' \rarrow{\ast}{G} u$.
    Hence, $u \in \pnf_{G'}$ holds.

  \item (Only-if):
    Since $u \notin V_G$ and $u \in \pnf_{G'}$, $u$ must be a newly introduced vertex.
    However, $u=v$ is ruled out from  $u \in \pnf_{G'}$.
    Thus, $u = v' \notin V_G$ holds.
    \smallskip

  \noindent
    (If): If $u = v' \notin V_G$, then $u$ has no outgoing edge,
    which immediately implies $u\in\pnf_{G'}$.
  \end{enumerate}
\end{proof}

Note that $v' \Rrightarrow_{G} v$ implies $v' \in \pnf_{ G\cup(v\to v') }$ by applying
  \eqref{sublem:exp-class_1} and \eqref{sublem:exp-class_2} of
Lemma~\ref{lem:exp-class} with $u := v'$.

The following Lemma~\ref{lem:NF-number} shows how an edge expansion affects
the number of normal forms in $G_{\ddiv}$.
This is also indicated by the colors in Table~\ref{tab:expansion-PNF};
in the red cases, the number increases, whereas in the blue cases, it decreases.
\begin{table}
  \caption{Classification of cases on an edge expansion $G \cup (v \to v')$}\label{tab:expansion-PNF}
  \small
  \setlength{\aboverulesep}{0pt}
  \setlength{\belowrulesep}{0pt}
    \begin{tabular}{c|c|c|c|c}
      \toprule
      & \multicolumn{2}{|c|}{$v \notin \pnf_{G}$}
      & \multicolumn{2}{c}{$v \in \pnf_{G}$}  \\
      \midrule 
      $v'\notin V_G$
      & \multicolumn{2}{>{\columncolor{red!15}}c|}{(a)}
      & \multicolumn{2}{c}{(b)}  \\ \hline
      $v'\in V_G$
      & \multicolumn{2}{c|}{(c)}
      & (d) if $v' \Rrightarrow_{G} v$
      & \multicolumn{1}{>{\columncolor{blue!18}}c}{(e)
        if $v' \not\Rrightarrow_{G} v$} \\
      \Xhline{1pt}
    \end{tabular}

  \vspace{3pt}
  \noindent\makebox[\linewidth][c]{%
      \footnotesize
      Red color (resp.\ blue) color indicates $|\NF_{G_{\ddiv}}|$ increases (resp.\ decreases).
  }
\end{table}

\begin{lemma}\label{lem:NF-number}
  Let $G' := G \cup (v\to v')$ be an edge expansion of an ARS $G$. Then,
  \[
  |\NF_{G'_{\ddiv}}| = \begin{cases}
    |\NF_{G_{\ddiv}}| - 1
    & \text{if } v\in \pnf_{G} \land v'\in V_G \land v'\not\Rrightarrow_{G} v, \\
    |\NF_{G_{\ddiv}}| + 1
    & \text{if } v \notin \pnf_{G} \land v' \notin V_G, \\
    |\NF_{G_{\ddiv}}|
    & \text{otherwise.}
  \end{cases}
  \]
\end{lemma}

\begin{proof}
  We consider the case analysis, i.e., the combination of
  whether $v \in \pnf_G$ and whether $v' \in V_G$.
  Among them, the case $v \in \pnf_G$ and $v' \in V_G$ is further decomposed
  with $v' \Rrightarrow_G v$ and $v' \not\Rrightarrow_G v$.
  \begin{description}
  \item[\textbf{Case (a):} $v \notin \pnf_{G}$ and $v' \notin V_G$.]
    From Lemma~\ref{lem:exp-class}(\ref{sublem:exp-class_3}) with $u:=v'$, a new
    equivalence class $\classD{v'}_{\stcongstarss{G'}} \in\NF_{G'_{\ddiv}}$ is created.

    Moreover, for any $u$, $\classD{u}_{\stcongstarss{G}} \in \NF_{G_{\ddiv}}$ implies
    $\classD{u}_{\stcongstarss{G'}} \in \NF_{G'_{\ddiv}}$.
    Otherwise, there would exist some $u$ such that $u \in \pnf_{G}$
    but $u \notin \pnf_{G'}$.
    By Lemma~\ref{lem:exp-class}(\ref{sublem:exp-class_1}), we have $u \rarrow{\ast}{G} v$,
    which implies $v \in \pnf_{G}$, contradicting the case assumption
    $v \notin \pnf_{G}$.
    Thus, $|\NF_{G'_{\ddiv}}| = |\NF_{G_{\ddiv}}| + 1$. 

    \item[\textbf{Case (b):} $v \in \pnf_{G}$ and $v' \notin V_G$.]
    Similarly to Case (a),
    a new equivalence class $\classD{v'}_{\stcongstarss{G'}} \in \NF_{G'_{\ddiv}}$ is
    created.
    On the other hand, since $v' \notin V_G$, we trivially have
    $v' \not\Rrightarrow_{G} v$.
    Thus, applying
    Lemma~\ref{lem:exp-class}(\ref{sublem:exp-class_1}) with $u:=v$
    shows $v \notin \pnf_{G'}$, i.e., $\classD{v}_{\stcongstarss{G'}} \notin \NF_{G'_{\ddiv}}$,
    while $\classD{v}_{\stcongstarss{G}} \in \NF_{G_{\ddiv}}$ holds by the case assumption.
    Hence, these two changes cancel each other out, resulting in
    $|\NF_{G'_{\ddiv}}| = |\NF_{G_{\ddiv}}|$. 

\item[\textbf{Case (c):} $v \notin \pnf_{G}$ and $v' \in V_G$.]
    First, for any $u \in V_G$, $u \in \pnf_{G}$ implies $u \in \pnf_{G'}$.
    Otherwise, there would exist some $u$ such that $u \in \pnf_{G}$ but
    $u \notin \pnf_{G'}$.
    By Lemma~\ref{lem:exp-class}(\ref{sublem:exp-class_1}),
    we obtain $u \rarrow{\ast}{G} v$, which induces $u \Rrightarrow_G u \rarrow{\ast}{G} v$
    since $u \in \pnf_{G}$.
    Applying Proposition~\ref{prop:pnf}~(\ref{subprop:pnf-2}) gives $v \Rrightarrow_G v$,
    contradicting the case assumption $v \notin \pnf_{G}$.

    Conversely, for any $u \in V_G$, $u \notin \pnf_{G}$ implies $u \notin \pnf_{G'}$.
    Otherwise, there would exist some $u$ such that $u \notin \pnf_{G}$
    but $u \in \pnf_{G'}$, which gives $v' \Rrightarrow_{G} v$
    by Lemma~\ref{lem:exp-class}~(\ref{sublem:exp-class_2}).
    By Proposition~\ref{prop:pnf}~(\ref{subprop:pnf-1}),
    we obtain $v \in \pnf_G$, again contradicting the case assumption.
    Therefore $|\NF_{G'_{\ddiv}}| = |\NF_{G_{\ddiv}}|$ holds.

  \item[\textbf{Case (d):} $v\in\pnf_{G}$, $v'\in V_G$, and $v' \Rrightarrow_G v$.]
    First, for any $u \in V_G$, $u \in \pnf_{G}$ implies $u \in \pnf_{G'}$ by
    Lemma~\ref{lem:exp-class}(\ref{sublem:exp-class_1}).
    In particular, $v \in \pnf_{G'}$.

    On the other hand, there may exist some $u$
    such that $u \notin \pnf_G$ but
    $u \in \pnf_{G'}$.
    By Lemma~\ref{lem:exp-class}(\ref{sublem:exp-class_2}), any such $u$ satisfies $v' \rarrow{\ast}{G} u$.
    Since $v' \Rrightarrow_G v$ by the case assumption,
    this sequence induces $u \rarrow{\ast}{G} v$.
    Thus, we obtain a cycle $u \rarrow{\ast}{G} v \rarrow{}{G'} v' \rarrow{\ast}{G} u$,
    and hence
    $\classD{u}_{\stcongstarss{G'}} = \classD{v}_{\stcongstarss{G'}}$.
    Combining this with $u,v \in \pnf_{G'}$,
    it follos that any such $u$ is simply absorbed into the existing normal form class of $v$,
    meaning no new normal equivalence classes are created.
    Hence, $|\NF_{G'_{\ddiv}}| = |\NF_{G_{\ddiv}}|$ holds.

  \item[\textbf{Case (e):} $v \in \pnf_{G}$, $v' \in V_G$, and $v' \not\Rrightarrow_{G} v$.]
    First, since $v' \not\Rrightarrow_{G} v$,
    applying Lemma~\ref{lem:exp-class}(\ref{sublem:exp-class_2})
    shows that for any $u \in V_G$, $u \notin \pnf_{G}$ implies $u \notin \pnf_{G'}$.
    Thus, no new normal form classes are created.

    On the other hand,
    at least one normal equivalence class is lost. Indeed, taking $u:=v$ in
    Lemma~\ref{lem:exp-class}(\ref{sublem:exp-class_1}) gives $v \notin \pnf_{G'}$,
    whereas $v \in \pnf_{G}$ holds by the case assumption.

    Furthermore, if there exists any $u$ such that $u \in \pnf_{G}$ but $u \notin \pnf_{G'}$,
    Lemma~\ref{lem:exp-class}(\ref{sublem:exp-class_1}) implies $u \rarrow{\ast}{G} v$.
    Combining this with $u \in \pnf_{G}$ shows
    $\classD{u}_{\stcongstarss{G}} = \classD{v}_{\stcongstarss{G}}$.
    This indicates that the loss of normal forms is
    uniquely restricted to the single class $\classD{v}_{\stcongstarss{G}}$.

    Therefore, exactly one equivalence class in normal form is removed, resulting in
    $|\NF_{G'_{\ddiv}}| = |\NF_{G_{\ddiv}}| - 1$.
\end{description}
\end{proof}

Note that Lemma~\ref{lem:NF-number} is used only in
Lemma~\ref{lem:abs-edge-extension2}, where $|\NF_{G_{\ddiv}}|$ serves as a measure for induction.
Specifically, the lemma is used to show that $|\NF_{G'_{\ddiv}}| \leq |\NF_{G_{\ddiv}}|$ under
condition that $v\in\pnf_{G}$ and $v'\in V_G$, i.e., Cases~(d) or (e).

\subsection{Finite confluent extension of a sub-ARS in a confluent ARS}

\begin{lemma} \label{lem:exp-class2}
  Let $G$ be a confluent ARS, and
  let $G' := G \cup (v \to v')$ be an edge expansion such that $G \conserveq G'$.
  Then, the following hold:
  \begin{enumerate}
  \item\label{lem:exp-class2-a}
    $v \in \pnf_G$ implies $v' \in \pnf_{G'}$,
  \item\label{lem:exp-class2-b}
    if $v' \in V_G$, then for any $w \in \pnf_G$ with $v' \rarrow{\ast}{G} w$,
    we have $w \in \pnf_{G'}$, and
  \item\label{lem:exp-class2-c}
    $G'$ is confluent whenever $v \in \pnf_G$ or $v' \in V_G$.
  \end{enumerate}
\end{lemma}

\begin{proof}
  \begin{enumerate}
  \item[~(\ref{lem:exp-class2-a})]
    Depending on whether $v' \in V_G$, we consider the following two cases.
    \begin{description}
    \item[\textbf{Case 1:} $v'\in V_G$.]
      Suppose $v \in \pnf_G$, i.e., we have $v \Rrightarrow_G v$.
      Since $G \conserveq G'$ and $v,v' \in V_G$,
      the edge $v \to v'$ in $G'$ implies $v \barrow{\ast}{G} v'$.
      By the confluence of $G$, there exists some $v'' \in V_G$ such that
      $v \rarrow{\ast}{G} v'' \larrow{\ast}{G} v'$.
      Furthermore, since $v \Rrightarrow_G v$, the sequence $v \rarrow{\ast}{G} v''$
      ensures $v'' \rarrow{\ast}{G} v$.
      This induces a path $v' \rarrow{\ast}{G} v'' \rarrow{\ast}{G} v$.
      Applying Proposition~\ref{prop:pnf}~(\ref{subprop:pnf-3}) with $w:=v'$ and $u:=v$,
      it follows that $v' \Rrightarrow_G v$.
      Finally, applying (\ref{sublem:exp-class_1}) and (\ref{sublem:exp-class_2})
      of Lemma~\ref{lem:exp-class} with $u:=v'$, we conclude $v' \in \pnf_{G'}$.

    \item[\textbf{Case 2:} $v'\notin V_G$.]
      Applying Lemma~\ref{lem:exp-class}(\ref{sublem:exp-class_3}) with
      $u := v'$, we obtain $v'\in\pnf_{G'}$.
    \end{description}

  \item[(\ref{lem:exp-class2-b})]
    We prove that $w \in \pnf_{G'}$ assuming $v' \rarrow{\ast}{G} w$ and $w \in \pnf_G$.
    \begin{description}
    \item[\textbf{Case 1:} $w \not\rarrow{\ast}{G} v$.]
      Applying Lemma~\ref{lem:exp-class}(\ref{sublem:exp-class_1}) with $u:=w$,
      we immediately obtain $w \in \pnf_{G'}$
      from the case assumption $w \not\rarrow{\ast}{G} v$.

    \item[\textbf{Case 2:} $w \rarrow{\ast}{G} v$.]
      First, since $w \in \pnf_G$, (i.e., $w \Rrightarrow_G w$)
      and $w \rarrow{\ast}{G} v$, we obtain $w \Rrightarrow_G v$
      by Proposition~\ref{prop:pnf}~(\ref{subprop:pnf-2}).
      Furthermore, since $G$ is confluent, the sequence $v' \rarrow{\ast}{G} w \Rrightarrow_G v$
      implies $v' \Rrightarrow_G v$ by Proposition~\ref{prop:pnf}~(\ref{subprop:pnf-3}).
      Therefore, since $v' \Rrightarrow_G v$, applying Lemma~\ref{lem:exp-class}(\ref{sublem:exp-class_1})
      with $u:=w$ implies $w \in \pnf_{G'}$.
    \end{description}

  \item[(\ref{lem:exp-class2-c})]
    The assumption $v \in \pnf_G$ or $v' \in V_G$ can be divided into three cases:
    (a) $v' \in V_G$ and $v \in V_G$, (b) $v' \in V_G$ and $v \notin V_G$, and (c) $v \in \pnf_G$.
    In Case (a), since $v, v' \in V_G$ and $G \conserveq G'$, we obtain $v \barrow{\ast}{G} v'$,
    which satisfies condition~(1) in Lemma~\ref{lem:expansion-PNF}.
    In Cases (b) and (c), the case assumptions themselves are sufficient
    for conditions~(2) and (3) in Lemma~\ref{lem:expansion-PNF}, respectively.
    Hence, the confluence of $G'$ follows from the lemma.
  \end{enumerate}
\end{proof}

Note that Lemma~\ref{lem:exp-class2}(\ref{lem:exp-class2-a}) requires the confluence of $G$.
As a counterexample, consider a non-confluent ARS $G = \{ a \to b, a \to c, c \to d \}$
and its extension $G' = G \cup \{ b \to c \}$.
Here, letting $v := b$ and $v' := c$, we have $b \in \pnf_G$, but $c \notin \pnf_{G'}$.

\begin{lemma}\label{lem:exp-class2d}\label{cor:exp-class3}
  Let $G_1$ and $G_2$ be ARSs such that $G_1\conserveq G_2$.
  Let $G := G_1 \cup {\lb v \rarrow{\ast}{G_2} v' \rb}$
  for a given sequence $v\rarrow{\ast}{G_2} v'$
  satisfying $v \in \pnf_{G_1}$ or $v' \in \pnf_{G_1}$.
  Then the following properties hold:
  \begin{enumerate}
  \item\label{lem:exp-class2d-1} $G_1 \conserveq G \conserveq G_2$.
  \item\label{lem:exp-class2d-2} $v'\in\pnf_G$,
  \item\label{lem:exp-class2d-3} $G$ inherits confluence from $G_1$.
  \item\label{lem:exp-class2d-4} $w \in \pnf_{G}$ for any $w\in\pnf_{G_1}\setminus [v]_{\barrow{\ast}{G}}$.
  \end{enumerate}
\end{lemma}

\begin{proof}
  We represent the sequence $v \rarrow{\ast}{G_2} v'$ as
  $v = v_0\rarrow{}{G_2} v_1\rarrow{}{G_2}\cdots\rarrow{}{G_2} v_m = v'$
  for some $m\geq 0$.
  We first show (\ref{lem:exp-class2d-1})-(\ref{lem:exp-class2d-3}).

  \begin{description}
  \item[\textbf{Case 1:} $v \in \pnf_{G_1}$.]
    For $0\leq i \leq m$, let $G'_{i} := G_1 \cup \lb v \mathrel{\rarrow{i}{G_2}} v_i \rb$,
    where $G'_0 = G_1$ and $G'_m = G$.
    By induction on $i$, we show that $G_1 \conserveq G'_i \conserveq G_2$,
    $v_i \in \pnf_{G_i}$,
    and $G'_i$ is confluent if $G_1$ is.

    The base case $i=0$ is obvious.
    For the induction step $i>0$, we have
    $G'_{i} = G'_{i-1} \cup \lb v_{i-1} \mathrel{\rarrow{}{G_2}} v_i \rb$.
    The induction hypothesis gives
    $G_1 \conserveq G'_{i-1} \conserveq G_2$,
    $v_{i-1}\in\pnf_{G'_{i-1}}$, and
    $G'_{i-1}$ is confluent if $G_1$ is.

    Since $v_{i-1} \in V_{G'_{i-1}}$, Proposition~\ref{prop:conserv} implies
    $G'_{i-1} \conserveq G'_i \conserveq G_2$, and
    hence $G_1 \conserveq G'_i \conserveq G_2$.
    By (\ref{lem:exp-class2-a}) and (\ref{lem:exp-class2-c}) of
    Lemma~\ref{lem:exp-class2},  we obtain $v_i \in\pnf_{G'_i}$ and that
    $G'_i$ inherits confluence of $G_1$.

  \item[\textbf{Case 2:} $v' \in \pnf_{G_1}$.]
    For $1\leq i \leq m$,
    let $G'_i := G_1 \cup  \lb v_{m-i} \mathrel{\rarrow{m-i}{G_2}} v' \rb$,
    where $G'_0 = G_1$ and $G'_m = G$.
    By induction on $i$, we show that $G_1 \conserveq G'_i \conserveq G_2$,
    $v' \in \pnf_{G'_i}$, and
    $G'_i$ is confluent if $G_1$ is.

    The base case $i=0$ is obvious.
    For the induction step $i>0$, we have
    $G'_{i} = G'_{i-1} \cup \lb v_{m-i} \mathrel{\rarrow{}{G_2}} v_{m-i+1} \rb$.
    The induction hypothesis gives
    $G_1 \conserveq G'_{i-1} \conserveq G_2$,
    $v' \in \pnf_{G'_{i-1}}$, and that
    $G'_{i-1}$ is confluent if $G_1$ is.
    Then, since $v_{m-i+1} \in V_{G'_{i-1}}$, Proposition~\ref{prop:conserv} implies
    $G'_{i-1} \conserveq G'_i \conserveq G_2$, and hence
    $G_1 \conserveq G'_i \conserveq G_2$.
    By applying Lemma~\ref{lem:exp-class2}(\ref{lem:exp-class2-b}) with
    $v':=v_{m-i+1}$ and $w:=v'$, we obtain $v' \in\pnf_{G'_i}$.
    Furthermore, Lemma~\ref{lem:exp-class2}(\ref{lem:exp-class2-c}) ensures that
    $G'_i$ inherits confluence from $G_1$.

  \end{description}
  Finally, we prove (\ref{lem:exp-class2d-4}).
  Assume $w \in \pnf_{G_1} \setminus [v]_{\barrow{\ast}{G}}$.
  Since $w\notin [v]_{\barrow{\ast}{G}}$, we have $w \not\barrow{\ast}{G} v$,
  which implies $w \not\barrow{\ast}{G} v'$ because $v \rarrow{\ast}{G} v'$.
  Suppose $w \rarrow{\ast}{G} w'$. From $w \not\barrow{\ast}{G} v$ and the
  construction of $G$, it follows that $w \rarrow{\ast}{G_1} w'$.
  Since $w \in \pnf_{G_1}$ (i.e., $w \Rrightarrow_{G_1} w$),
  we have $w' \rarrow{\ast}{G_1} w$.
  Therefore, $w' \rarrow{\ast}{G} w$ holds, which implies $w \in \pnf_G$.
\end{proof}

\begin{lemma}\label{lem:abs-edge-extension2}
  Let $G_1$ be a finite ARS and $G_2$ be a confluent ARS
  such that $G_1 \conserveq G_2$.
  Suppose further that $s, v \in \pnf_{G_1}$, $s \barrow{\ast}{G_1} v$, and
  $\classD{s}_{\stcongstarss{G_1}} \neq \classD{v}_{\stcongstarss{G_1}}$.
  Then, there exists a finite ARS $G$ satisfying the following properties:
  \begin{enumerate}
    \item $G_1 \conserveq G \conserveq G_2$,
    \item $|\NF_{G_{\ddiv}}| < |\NF_{(G_1)_{\ddiv}}|$.
  \end{enumerate}
\end{lemma}

\begin{proof}
  By confluence of $G_2$, there exists a join sequence
  $s \rarrow{\ast}{G_2} w \larrow{\ast}{G_2} v$ for some $w \in V_{G_2}$.
  We prove the claim by induction on the shortest total length $n$
  of such a join sequence.

  Since 
  $\classD{s}_{\stcongstarss{G}} \neq \classD{v}_{\stcongstarss{G}}$,
  we must have $s \neq w$ or $w \neq v$.
  By the symmetry of the assumptions on $s$ and $v$,
  we can assume $w \neq v$ without loss of generality.
  Thus, the sequence can be written as
  $s \rarrow{\ast}{G_2} w \larrow{\ast}{G_2} v' \larrow{}{G_2} v$
  for some $v' \in V_{G_2}$, where the remaining join sequence
  between $s$ and $v'$ has length $n-1$.

  Let $G' := G_1 \cup (v \to v')$.
  Since $v \in \pnf_{G_1}$, we obtain
  (1') $G_1 \conserveq G' \conserveq G_2$ by Proposition~\ref{prop:conserv}, and
  (2') $|\NF_{G'_{\ddiv}}| \leq |\NF_{(G_1)_{\ddiv}}|$ by Lemma~\ref{lem:NF-number}.

  We now consider the following two cases based on $v'$:
  \begin{description}
    \item[\textbf{Case where $v' \in V_{G_1}$ and $v' \not\Rrightarrow_{G_1} v$:}]
    Since $v \in \pnf_{G_1}$, Lemma~\ref{lem:NF-number}
    immediately yields $|\NF_{(G_1)_{\ddiv}}| > |\NF_{G'_{\ddiv}}|$.
    Thus, the claim holds by setting $G := G'$.

    \item[\textbf{Otherwise ($v' \notin V_{G_1}$ or $v' \Rrightarrow_{G_1} v$):}]
    To apply the induction hypothesis to $G'$ with respect to
    $s$ and $v'$, we verify the following four properties:
    \begin{description}
      \item[\textbf{Property 1:}] $s \in \pnf_{G'}$.
      Since $\classD{s}_{\stcongstarss{G_1}} \neq \classD{v}_{\stcongstarss{G_1}}$
      and $s \in \pnf_{G_1}$,
      we immediately have
      \begin{align}
      & s \mathrel{\not\rarrow{\ast}{G_1}} v.
        \label{eq:abs-edge-extension2}
      \end{align}
      Thus, applying Lemma~\ref{lem:exp-class}(\ref{sublem:exp-class_1}) with $u:=s$
      gives $s \in \pnf_{G'}$.

      \item[\textbf{Property 2:}] $v' \in \pnf_{G'}$.
      If $v' \notin V_{G_1}$,
      this holds by applying Lemma~\ref{lem:exp-class}(\ref{sublem:exp-class_3})
      with $u := v'$.
      If $v' \Rrightarrow_{G_1} v$, the property follows
      by applying (\ref{sublem:exp-class_1}) and (\ref{sublem:exp-class_2})
      of Lemma~\ref{lem:exp-class} with $u := v'$.

      \item[\textbf{Property 3:}] $s \barrow{\ast}{G'} v'$.
      This holds immediately from the assumption $s \barrow{\ast}{G_1} v$
      and the new edge $v \rarrow{}{G'} v'$.

      \item[\textbf{Property 4:}]
      $\classD{s}_{\stcongstarss{G'}} \neq \classD{v'}_{\stcongstarss{G'}}$.
      To establish the claim,
      we show $s \mathrel{\not\rarrow{\ast}{G'}} v'$ by contradiction.
      Suppose $s \rarrow{\ast}{G'} v'$.
      \begin{itemize}
      \item If $v' \notin V_{G_1}$,
        any sequence $s \rarrow{\ast}{G'} v'$ must utilize
        the new edge $v\to v'$.
        Since $v'$ has no outgoing edges in $G'$, the sequence must be of
        the form $s \rarrow{\ast}{G_1} v \rarrow{}{G'} v'$,
        which contains $s \rarrow{\ast}{G_1} v$.
      \item If $v' \Rrightarrow_{G_1} v$, then $v' \rarrow{\ast}{G_1} v$ by
        Proposition~\ref{prop:pnf}(\ref{subprop:pnf-1}).
        From the contradictory assumption $s\rarrow{\ast}{G'} v'$,
        either $s\rarrow{\ast}{G_1} v'$
        or the reduction sequence utilizes the new edge $v \to v'$.
        In the former case, combined with $v' \rarrow{\ast}{G_1} v$,
        we immediately obtain $s \rarrow{\ast}{G_1} v$.
        In the latter case, the sequence $s\rarrow{\ast}{G_1} v\rarrow{}{G'} v'$
        exists, which also contains $s \rarrow{\ast}{G_1} v$.
      \end{itemize}
      In both cases, we obtain $s \rarrow{\ast}{G_1} v$,
      contradicting \eqref{eq:abs-edge-extension2}.
      Thus, we have $s \mathrel{\not\rarrow{\ast}{G'}} v'$,
      which implies the claim.
    \end{description}
    Therefore, by the induction hypothesis, there exists an ARS $G$ satisfying
    (1'') $G' \conserveq G \conserveq G_2$ and
    (2'') $|\NF_{G_{\ddiv}}| < |\NF_{G'_{\ddiv}}|$.

    Since $G_1 \conserveq G' \conserveq G \conserveq G_2$,
    we obtain Property~(1): $G_1 \conserveq G \conserveq G_2$.
    Finally, combining (2'') with (2') shows Property~(2):
    $|\NF_{G_{\ddiv}}| < |\NF_{G'_{\ddiv}}| \leq |\NF_{(G_1)_{\ddiv}}|$.
  \end{description}
\end{proof}

\begin{lemma}\label{lem:merge}
  Let $G'$ be a finite ARS, and $G_2$ be a confluent ARS.
  If $G' \conserveq G_2$, then there exists a confluent finite ARS $G$
  with $G' \conserveq G \conserveq G_2$.
\end{lemma}
\begin{proof}
  We prove the claim by induction on $|\NF_{G'_{\ddiv}}|$,
  which is finite due to the finiteness of $G'$.

  If $G'$ is confluent, we are done
  by taking $G := G'$.
  Otherwise, there exist $t,u \in V_{G'}$ such that $t \barrow{\ast}{G'} u$
  and $t \mathrel{\not\,\downarrow_{G'}} u$.
  By Proposition~\ref{prop:looping_pnf}, there exist $s,v \in \pnf_{G'}$ such
  that $t \rarrow{\ast}{G'} s$ and  $u \rarrow{\ast}{G'} v$.
  Combining this with $t \mathrel{\not\,\downarrow_{G'}} u$,
  we obtain $s \not\rarrow{\ast}{G'} v$, and hence,
  $\classD{s}_{\stcongstarss{G'}} \neq \classD{v}_{\stcongstarss{G'}}$.

  This fulfills the assumptions $s, v \in \pnf_{G'}$, $s \barrow{\ast}{G'} v$,
  and $\classD{s}_{\stcongstarss{G'}} \neq \classD{v}_{\stcongstarss{G'}}$
  of Lemma~\ref{lem:abs-edge-extension2}.
  Hence,
  there exists a finite ARS $G''$ with
  $G' \conserveq G'' \conserveq G_2$ and $|\NF_{G''_{\ddiv}}| < |\NF_{G'_{\ddiv}}|$.
  Thus, by the induction hypothesis applied to $G''$, there exists a confluent finite ARS $G$
  such that $G'' \conserveq G \conserveq G_2$.
  Since $G' \conserveq G''$, we conclude $G' \conserveq G \conserveq G_2$.
\end{proof}

\section{Spanning graph and its normalization}
\label{sec:edge_properties}

We define spanning graphs and their normalization steps $\vdash_n$.
Then, we investigate their properties.

\subsection{Spanning graph and its properties}

This subsection defines an acyclic and strongly deterministic quotient ARS,
which we call a \emph{spanning graph} $(G, T_1, T_2)$. 

\begin{definition}\label{def:spanning_graph}
  Given ARSs $G$ and $\Vrt$ and finite ARSs $T_1, T_2$,
  a triple $(G, T_1, T_2)$ is a \emph{spanning graph} with respect to $\Vrt$
  if $G \conserveq \Vrt$, $V_{T_1\cup T_2} \subseteq V_G$,
  $\rarrow{}{T_1} \cap \rarrow{}{T_2} = \emptyset$,
  and $(T_1\cup T_2) / G$ is acyclic and strongly deterministic.
\end{definition}

We define a normalization relation and a monotonically decreasing measure
on spanning graphs $(G, T_1, T_2)$ to show that the normalizing steps terminate.
Then, confluence of the ARS $G' \cup T'_1 \cup T'_2$ of a normalized
spanning graph $(G', T'_1, T'_2)$ is obtained.
Fig.~\ref{fig:seq-edge-norm} illustrates the weight construction.
Since $(T_1 \cup T_2) / G$ is finite, acyclic, and strongly deterministic,
the weight is well-defined.

\begin{definition}[Weight]\label{def:fixed-weight}
Let $G,T_1,T_2$ be ARSs with $V_{T_1} \cup V_{T_2} \subseteq V_{G}$.
For a vertex $t\in V_G$, the weight $w_{\mathcal S}(t)$ is
\[
w_{\mathcal S}(t) =
\begin{cases}
  w_{\mathcal S}(u)+1 &
  \text{if } \exists u \in V_G . \ \classB{t}_{\barrow{\ast}{G}} \to_{T_1/G} \classB{u}_{\barrow{\ast}{G}}
  \\
  \max\{ w_{\mathcal S}(u') \mid
  \classB{u}_{\barrow{\ast}{G}} \barrow{\ast}{T_1/G} \classB{u'}_{\barrow{\ast}{G}} \} + 1 &
  \text{if } \exists u \in V_G . \ \classB{t}_{\barrow{\ast}{G}} \to_{T_2/G} \classB{u}_{\barrow{\ast}{G}}
  \\
  0 & \text{otherwise}
\end{cases}
\]
\end{definition}
In general, $w_{\mathcal S}$ is not well-defined, but it is so for a spanning graph.
\begin{proposition}\label{prop:weight}
  If ${\mathcal S} = (G, T_1, T_2)$ is a spanning graph,
  $w_{\mathcal S}$ is well-defined and its minimum value is $0$.
\end{proposition}

\begin{proof}
  Since ${\mathcal S} = (G, T_1, T_2)$ is a spanning graph,
  $(T_1 \cup T_2) / G$ is finite, acyclic, and strongly deterministic;
  hence, it forms a finite reversed rooted forest.
  This implies that for any reversed root $\classB{v}_{\barrow{\ast}{G}}$ in $(T_1 \cup T_2) / G$,
  no further transitions exist, and thus we have
  $w_{\mathcal S}(v) = 0$ by the ``otherwise'' case of Definition~\ref{def:fixed-weight}.

  Since $(T_1 \cup T_2) / G$ is strongly deterministic, the cases in Definition~\ref{def:fixed-weight} do not overlap.
  Therefore, we can inductively define $w_{\mathcal S}$ backward from the reversed roots,
  which ensures that $w_{\mathcal S}(t) = w_{\mathcal S}(t')$ whenever $t \barrow{\ast}{G} t'$.
\end{proof}

\begin{figure}
  \begin{center}
    \begin{minipage}{0.59\textwidth}
      \begin{center}
  \begin{tikzpicture}[line width=0.07em, x=5, y=5, scale=1.2]
    \path pic at (0,0) {simplecircle};  
    \path pic at (10,0) {simplecircle}; 
    \path pic at (10,-10) {simplecircle}; 
    \path pic at (20,0) {simplecircle}; 
    \path pic at (20,-10) {simplecircle}; 
    \path pic at (30,-10) {simplecircle}; 

    \node at (1,2)(ua) {\small $\ $};
    \node[color=red] at (-2,-1) {\small $w=4$};
    \node at (11,2)(ub) {\small $\ $};
    \node[color=red] at (8,-1) {\small $w=4$};
    \node at (12,-11)(uc) {\small $\ $};
    \node at (8,-11)(vc) {\small $\ $};
    \node at (9,-9)(vc2) {\small $\ $};
    \node[color=red] at (9,-12) {\small $w=3$};
    \node at (21,2)(ud) {\small $\ $};
    \node[color=red] at (18,-1) {\small $w=2$};
    \node at (22,-11)(ue) {\small $\ $};
    \node at (18,-11)(ve) {\small $\ $};
    \node at (19,-9)(ve2) {\small $\ $};
    \node[color=red] at (19,-12) {\small $w=1$};
    \node at (28,-11)(vf) {\small $\ $};
    \node[color=red] at (29,-12) {\small $w=0$};

    \path[->,color=blue] (ua) edge node[color=blue,below,xshift=-4,yshift=0]  {\scriptsize $T_1$} (vc);

    \path[->,color=blue] (ub) edge node[color=blue,below,xshift=-7,yshift=0]  {\scriptsize $T_1$} (vc2);

    \path[->,color=darkgreen] (uc) edge node[below,xshift=0,yshift=0]  {\scriptsize $T_2$} (ve);

    \path[->,color=blue] (ud) edge node[color=blue,below,xshift=-7,yshift=0]  {\scriptsize $T_1$} (ve2);

    \path[->,color=blue] (ue) edge node[color=blue,below,xshift=0,yshift=0]  {\scriptsize $T_1$} (vf);

  \end{tikzpicture}
      \end{center}
    \end{minipage}

    \medskip
    \caption{A spanning graph and the weights}\label{fig:seq-edge-norm}
  \end{center}
\end{figure}
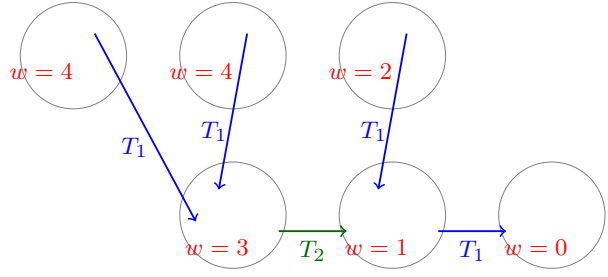

\begin{proposition}\label{prop:rewrite-decrease}
  Let ${\mathcal S} = (G,T_1,T_2)$ be a spanning graph, and suppose $s \rarrow{\ast}{G\cup T_1\cup T_2} t$.
  Then, the following properties hold:
  \begin{enumerate}
    \item $w_{\mathcal S}(s) \geq w_{\mathcal S}(t)$.
    \item If $w_{\mathcal S}(s) = w_{\mathcal S}(t)$, then $s \rarrow{\ast}{G} t$.
    \item If $w_{\mathcal S}(s) > w_{\mathcal S}(t)$,
    then there exists a sequence $u \to_{T_1\cup T_2} u'$
    such that $s \rarrow{\ast}{G} u \to_{T_1\cup T_2} u' \rarrow{\ast}{G\cup T_1\cup T_2} t$
    and $w_{\mathcal S}(s) > w_{\mathcal S}(u')$.
  \end{enumerate}
\end{proposition}

\begin{proof}
  These claims are proved simultaneously by induction on the length
  of the reduction sequence $s \rarrow{\ast}{G\cup T_1\cup T_2} t$.
  The base case where the length is $0$ (i.e., $s = t$) trivially holds for all properties.
  For the induction step, the properties follow directly
  from Definition~\ref{def:fixed-weight} by analyzing the first step of the edge,
  considering whether it belongs to $G$, $T_1$, or $T_2$.
\end{proof}

\subsection{Normalization Step of Spanning Graphs}\label{subsec:normalization}

This subsection is devoted to the normalization of edges in $T_1$ and $T_2$
on a deterministic and acyclic spanning graph $(G, T_1, T_2)$ with respect to $\Vrt$.

\begin{definition}[Normalization of spanning graphs]\label{def:normalization-step}
  Let $\Vrt$, $G$, $T_1$, and $T_2$ be ARSs such that $G \conserveq \Vrt$
  and $V_{T_1 \cup T_2} \subseteq V_G$.
  A spanning graph $\mathcal{S} = (G,T_1,T_2)$ is \emph{normalizable}
  to $\mathcal{S}' = (G',T'_1,T'_2)$ with respect to $\Vrt$, written as $\mathcal{S} \vdash^{(u,u')}_{(v,v')} \mathcal{S}'$,
  if $u \notin \pnf_G$, $(u,u') \in T_i$ (for $i \in \{1,2\}$),
  and there exist $w, x \in \pnf_G$ and $y \in V_{\Vrt}$ such that
  \begin{itemize}
    \item[(i)] $u \rarrow{\ast}{G} w \rarrow{\ast}{\Vrt} v$,
    \item[(ii)] $u' \barrow{\ast}{G} x$ if $i=1$, and $u' \barrow{\ast}{G\cup T_1} x$ if $i=2$,
    \item[(iii)] $x \rarrow{\ast}{\Vrt} y \larrow{\ast}{\Vrt} v'$,
    \item[(iv)] $G' = G \cup \lb w \rarrow{\ast}{\Vrt} v \rb \cup \lb x
    \rarrow{\ast}{\Vrt} y \rb \cup \lb v' \rarrow{\ast}{\Vrt} y \rb$, and
    \item[(v)] $V_{T'_i} = V_{T_i} \cup \{v,v'\}$,
    ${\to_{T'_i}} = ({\to_{T_i}} \setminus \{(u,u')\}) \cup \{(v,v')\}$, and $T'_{3-i} = T_{3-i}$.
  \end{itemize}
\end{definition}

Fig.~\ref{fig:normalization2} illustrates the normalization step described
in Definition~\ref{def:normalization-step}.
Note that since $x \in \pnf_G$, the last $\barrow{\ast}{G}$ connecting
  to $x$ in condition (ii) simplifies
to $\rarrow{\ast}{G}$.
We may write $\mathcal{S} \vdash_n \mathcal{S}'$ instead of
$\mathcal{S} \vdash^{(u,u')}_{(v,v')} \mathcal{S}'$ when the specific choice of pairs $(u,u')$ and $(v,v')$ is irrelevant.
We say that a spanning graph $(G,T_1,T_2)$ is \emph{normalized} if $u \in \pnf_{G}$ for every $(u,u') \in T_1 \cup T_2$.
Trivially, the normalization relation preserves the finiteness of $G \cup T_1 \cup T_2$.

\begin{figure}
  \begin{center}
    \begin{minipage}{0.49\textwidth}
      \begin{center}
  \begin{tikzpicture}[line width=0.07em, x=5, y=5, scale=1.2]
    \path pic at (0,0) {simplecircle};  
    \path pic at (16,0) {simplecircle}; 
    \node at (1.2,2)(u) {\small $u$};
    \node at (-2,-1.8)(uz) {$w$};
    \node at (2,-6)(v) {\small $v$};

    \node at (15,2)(udn) {\small $u'$};
    \node at (18,-2)(x) {\small $x$};
    \node at (14,-8)(y) {\small $y$};
    \node at (8,-6)(vd) {\small $v'$};

    \path[->,color=blue]
          (u) edge node[color=blue,below] {\scriptsize $T_1$} (udn);
    \path[->]
    (u) edge node [xshift=-4,yshift=2] {\scriptsize $\ast$}
    node [xshift=6,yshift=-4] {\scriptsize $G$}  (uz);
    \path[->,color=red]
    (uz) edge node [above,xshift=2,yshift=-2] {\scriptsize $\ast$}
    node [xshift=-7,yshift=-5] {\scriptsize $G'$}  (v);

          %
          %
          \path[->]
          (udn) edge node [xshift=3,yshift=2] {\scriptsize $\ast$}
          node [xshift=-5,yshift=-3] {\scriptsize $G$} (x);

    \path[->,color=red]
    (x) edge node [above,xshift=-2,yshift=0] {\scriptsize $\ast$}
    node[below,xshift=7,yshift=3] {\scriptsize $G'$} (y);
    \path[->,color=red]
    (vd) edge node[above,xshift=1,yshift=-1] {\scriptsize $\ast$}
    node[below,xshift=-2] {\scriptsize $G'$} (y);

    \path[->,color=red]
          (v) edge node[color=red,below] {\scriptsize $T'_1$} (vd);
  \end{tikzpicture}
        \subcaption{The case $(u,u')\in T_1$ ($i=1$)}
      \end{center}
    \end{minipage}
    \bigskip

    \begin{minipage}{0.59\textwidth}
      \begin{center}
  \begin{tikzpicture}[line width=0.07em, x=5, y=5, scale=1.2]
    \path pic at (0,0) {simplecircle};  
    \path pic at (10,0) {simplecircle}; 
    \path pic at (20,0) {simplecircle}; 
    \path pic at (36,0) {simplecircle}; 

    \node at (1.2,2)(u) {\small $u$};
    \node at (-2,-1.8)(uz) {$w$};
    \node at (2,-6)(v) {\small $v$};
    \node at (8,2)(ud) {\small $u'$};
    \node at (12.5,0)(w0) {$\cdot$};
    \node at (17.5,0)(ud1) {\small $\cdot$};
    \node at (22.5,0)(w1) {\small $\cdot$};
    \node at (27.5,0)(ud2) {};
    \node at (29.5,0)(w2) {};
    \node at (28.5,0)(dot) {\scriptsize $\cdots$};

    \node at (35,2)(udn) {\small $\cdot$};
    \node at (38,-2)(x) {\small $x$};
    \node at (34,-8)(y) {\small $y$};
    \node at (28,-6)(vd) {\small $v'$};

    \path[->,color=blue]
          (u) edge node[color=blue,below] {\scriptsize $T_2$} (ud);
    \path[->]
    (u) edge node [xshift=-4,yshift=2] {\scriptsize $\ast$}
    node [xshift=6,yshift=-4] {\scriptsize $G$}  (uz);
    \path[->,color=red]
    (uz) edge node [above,xshift=2,yshift=-2] {\scriptsize $\ast$}
    node [xshift=-7,yshift=-5] {\scriptsize $G'$}  (v);

    \path[<->]
          (w0) edge node[above] {\scriptsize $\mathit{T_1}$} (ud1);
    \path[<->]
    (w0) edge node [above,xshift=1,yshift=-1] {\scriptsize $\ast$}
    node [below] {\scriptsize $G$}  (ud);

    \path[<->]
          (w1) edge node[above] {\scriptsize $\mathit{T_1}$} (ud2);
    \path[<->]
    (w1) edge node [above] {\scriptsize $\ast$}
    node [below,xshift=2] {\scriptsize $G$}  (ud1);

    \path[<->]
          (w2) edge node[below,above] {\scriptsize $\mathit{T_1}$} (udn);
          \path[->]
          (udn) edge node [xshift=3,yshift=2] {\scriptsize $\ast$}
          node [xshift=-5,yshift=-3] {\scriptsize $G$} (x);

    \path[->,color=red]
    (x) edge node [above,xshift=-2,yshift=0] {\scriptsize $\ast$}
    node[below,xshift=7,yshift=3] {\scriptsize $G'$} (y);
    \path[->,color=red]
    (vd) edge node[above,xshift=1,yshift=-1] {\scriptsize $\ast$}
    node[below,xshift=-2] {\scriptsize $G'$} (y);

    \path[->,color=red]
          (v) edge node[color=red,below] {\scriptsize $T'_2$} (vd);
  \end{tikzpicture}
        \subcaption{The case $(u,u')\in T_2$ ($i=2$)}
      \end{center}
    \end{minipage}
    \caption{Normalization of spanning graphs}
    \label{fig:normalization2}
  \end{center}
\end{figure}
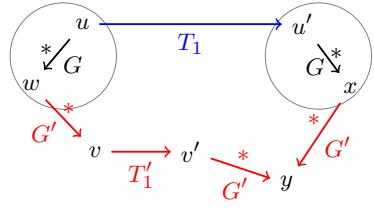
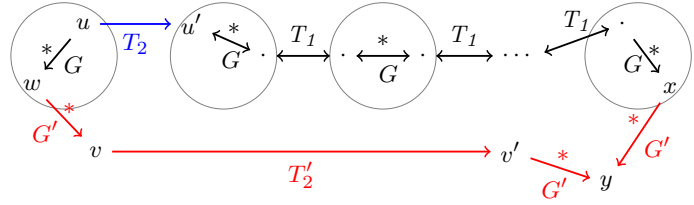

\begin{lemma}\label{prop:norm-basic}\label{lem:reduction-conf}
  Let $(G,T_1,T_2)$ be a spanning graph with $G \conserveq \Vrt$.
  For a normalization $(G,T_1,T_2) \vdash^{(u,u')}_{(v,v')} (G',T'_1,T'_2)$, the following
  properties hold.
  \begin{enumerate}
  \item \label{prop:norm-basic-2}
    $V_{T'_1\cup T'_2} \subseteq V_{G'}$,
  \item\label{prop:norm-basic-0}
  ${\barrow{\ast}{G\cup T_1}}
  \subseteq{\barrow{\ast}{G'\cup T'_1}}$,
  \item\label{prop:norm-basic-3}
  ${\barrow{\ast}{G\cup T_1\cup T_2}}
  \subseteq{\barrow{\ast}{G'\cup T'_1\cup T'_2}}$,
  \item \label{prop:norm-basic-1}
    $G \conserveq G' \conserveq \Vrt$,
  \item \label{lem:reduction-conf-3}
    $v \in \pnf_{G'}$,
  \item \label{lem:reduction-conf-4}
    $G'$ inherits the finiteness and confluence of $G$, and
  \item \label{lem:reduction-conf-2}
    $z \in \pnf_{G'}$ for any
    $z \in \pnf_{G} \setminus
    (\classB{v}_{\barrow{\ast}{G'}} \cup \classB{v'}_{\barrow{\ast}{G'}})$.
  \end{enumerate}
\end{lemma}
\begin{proof}
  (\ref{prop:norm-basic-2}): $V_{T'_1\cup T'_2} \subseteq V_{G'}$ is clear
  from $V_{T_1\cup T_2}\subseteq V_G \subseteq V_{G'}$ and
  the construction of $T'_1$ and $T'_2$.

  (\ref{prop:norm-basic-0}): If $(u,u') \in T_1$, then $(v,v') \in T'_1$
  by condition (v) of Definition~\ref{def:normalization-step}.
  From conditions (i), (ii), and (iii) of Definition~\ref{def:normalization-step},
  we have $u \barrow{\ast}{G'} v$ and $u' \barrow{\ast}{G'} v'$.
  Thus, $u \barrow{\ast}{G'\cup T'_1} u'$ holds.
  If $(u,u') \in T_2$, the claim is trivial since $T_1 = T'_1$
  follows from condition (v) of Definition~\ref{def:normalization-step}.

  (\ref{prop:norm-basic-3}): Since $G\subseteq G'$ and $(T_1\cup T_2)\setminus (T'_1\cup T'_2) = \{(u,u')\}$,
  it suffices to show that $u \barrow{\ast}{G'\cup T'_1\cup T'_2} u'$,
  which follows from $(v,v')\in T'_1\cup T'_2$ and
  (i), (ii), and (iii) of Definition~\ref{def:normalization-step}.

  We will show (\ref{prop:norm-basic-1}), (\ref{lem:reduction-conf-3}),
  (\ref{lem:reduction-conf-4}), and (\ref{lem:reduction-conf-2}),
  From (iv) of Definition~\ref{def:normalization-step},
  $G'$ can be obtained from $G$ by the following steps:
  \[
  \left\{
  \begin{array}{l}
    G_1 := G \cup \lb w \rarrow{\ast}{\Vrt} v \rb, \\
    G_2 := G_1 \cup \lb x \rarrow{\ast}{\Vrt} y \rb, \\
    G'  := G_2 \cup \lb v' \rarrow{\ast}{\Vrt} y \rb. \\
  \end{array}
  \right.
  \]
  where $w, x \in \pnf_G$ and $y\in V_{\Vrt}$ are vertices specified
  in Definition~\ref{def:normalization-step}.
  Moreover,
  \begin{align}
    & w \not\barrow{\ast}{G} x,  \label{eq:norm-basic_w-x_G}
  \end{align}
  otherwise (i) and (ii) imply $u \barrow{\ast}{G} u'$
  which contradicts the acyclicity of $(T_1 \cup T_2)/G$ because $(u,u') \in T_1 \cup T_2$.

  \begin{description}
  \item[Properties of $G_1$:]
  Since $w \in \pnf_G$, we can apply Lemma~\ref{lem:exp-class2d}, which shows that
  \begin{align}
    & \label{eq:norm-basic_3_G1}
        G \conserveq G_1 \conserveq \Vrt \\
    & \label{eq:norm-basic_4_G1}
        v\in\pnf_{G_1} \\
    & \label{eq:norm-basic_5_G1}
        \text{$G_1$ inherits confluence of $G$} \\
    & \label{eq:norm-basic_6_G1}
      \text{$z \in \pnf_{G_1}$ for any $z \in \pnf_{G} \setminus \classB{w}_{\barrow{\ast}{G_1}}$.}
  \end{align}

  \item[Properties of $G_2$:]
  From $x \in \pnf_G$, \eqref{eq:norm-basic_w-x_G}, and \eqref{eq:norm-basic_6_G1},
  we have $x \in \pnf_{G_1}$.
  Thus we can apply Lemma~\ref{lem:exp-class2d}, which shows that
  \begin{align}
    & \label{eq:norm-basic_3_G2}
        G_1 \conserveq G_2 \conserveq \Vrt \\
    & \label{eq:norm-basic_4_G2}
      y \in\pnf_{G_2}  \\
    & \label{eq:norm-basic_5_G2}
          \text{$G_2$ inherits confluence of $G_1$} \\
    & \label{eq:norm-basic_6_G2}
      \text{$z \in \pnf_{G_2}$ for any $z \in \pnf_{G_1} \setminus \classB{x}_{\barrow{\ast}{G_2}}$.}
  \end{align}

  \item[Properties of $G'$:]
  From \eqref{eq:norm-basic_4_G2}, we can apply Lemma~\ref{lem:exp-class2d}, which shows that
  \begin{align}
    & \label{eq:norm-basic_3_Gd}
        G_2 \conserveq G' \conserveq \Vrt \\
    & \label{eq:norm-basic_5_Gd}
        \text{$G'$ inherits confluence of $G_2$} \\
    & \label{eq:norm-basic_6_Gd}
        \text{$z \in \pnf_{G'}$ for any $z \in \pnf_{G_2} \setminus \classB{v'}_{\barrow{\ast}{G'}}$.}
  \end{align}
  \end{description}
  The property (\ref{prop:norm-basic-1}) follows directly
  from \eqref{eq:norm-basic_3_G1}, \eqref{eq:norm-basic_3_G2},
  and \eqref{eq:norm-basic_3_Gd}.
  For (\ref{lem:reduction-conf-3}),
  $v \barrow{\ast}{G_1} w \not\barrow{\ast}{G} x \barrow{\ast}{G'} v'$
  and $G\conserveq G_1 \conserveq G_2 \conserveq G'$ imply
  $v \not\barrow{\ast}{G'} x$ (also $v \not\barrow{\ast}{G_2} x$) and $v \not\barrow{\ast}{G'} v'$.
  Thus, from \eqref{eq:norm-basic_4_G1} and \eqref{eq:norm-basic_6_G2} we have $v \in \pnf_{G_2}$.
  Hence, from \eqref{eq:norm-basic_6_Gd} we obtain  (\ref{lem:reduction-conf-3}): $v \in \pnf_{G'}$.
  The property (\ref{lem:reduction-conf-4}) is clear from \eqref{eq:norm-basic_5_G1},
  \eqref{eq:norm-basic_5_G2}, and \eqref{eq:norm-basic_5_Gd}.

  For (\ref{lem:reduction-conf-2}),
  suppose $z\in \pnf_{G}\setminus (\classB{v}_{\barrow{\ast}{G'}} \cup \classB{v'}_{\barrow{\ast}{G'}})$.
  Since $w \barrow{\ast}{G'} v$ from the construction of $G'$, we have $z \not\barrow{\ast}{G'} w$.
  Analogously, $x \barrow{\ast}{G'} v'$, and thus, $z \not\barrow{\ast}{G_1} w$ follows from $G_1 \conserveq G'$.
  Combining this with \eqref{eq:norm-basic_6_G1}, we obtain $z\in \pnf_{G_1}$.
  Thus, similarly, $z\in \pnf_{G_2}$ follows from \eqref{eq:norm-basic_6_G2} and $z \not\barrow{\ast}{G_2} x$.
  Finally, $z\in \pnf_{G'}$ follows from \eqref{eq:norm-basic_6_Gd} and $z \not\barrow{\ast}{G'} v'$.
\end{proof}

\begin{lemma}\label{lem:reduction-welldef}
  Let $\Vrt$, $G$, $T_1$, and $T_2$ be ARSs with $G\conserveq \Vrt$ and
  $V_{T_1 \cup T_2} \subseteq V_G$.
  A normalization step ${\mathcal S} \vdash^{(u,u')}_{(v,v')} {\mathcal S'}$
  of a spanning graph ${\mathcal S} = (G,T_1,T_2)$ with respect to $\Vrt$ is
  well-defined;
  i.e., ${\mathcal S'} = (G',T'_1,T'_2)$ is a spanning graph
  (with respect to $\Vrt$).
\end{lemma}
\begin{proof}
  Let $(G,T_1,T_2) \vdash^{(u,u')}_{(v,v')} (G',T'_1,T'_2)$.
  From (1)--(3) of Lemma~\ref{prop:norm-basic}, it suffices to show that
  $\to_{T'_1} \cap \to_{T'_2} = \emptyset$, and that
  $(T'_1 \cup T'_2) / G'$ is acyclic and strongly deterministic.

  We consider two cases for $i \in \{1, 2\}$ such that $(u,u') \in T_i$ and $(v,v') \in T'_i$:
  \begin{itemize}
    \item[(i)]
    If $i = 1$, we have $T'_2 = T_2$, $v \to_{T'_1} v'$, and
    $v \larrow{\ast}{G'} u \rarrow{}{T_1} u' \barrow{\ast}{G'} v'$.

    \item[(ii)]
    If $i = 2$, we have $T'_1 = T_1$, $v \to_{T'_2} v'$, and
    $v \larrow{\ast}{G'} u\rarrow{}{T_2} u'\barrow{\ast}{T_1\cup G'} v'$.
  \end{itemize}

  \noindent
  {\bf Mutually disjoint:}
  We prove $\to_{T'_1} \cap \to_{T'_2} = \emptyset$ by contradiction;
  assume that $\to_{T'_1} \cap \to_{T'_2} \neq \emptyset$.
  This overlap must be caused by the introduced edge $(v,v')$.
  In the case where $i=1$, since $v \to_{T'_1} v'$, we must have $v \to_{T_2} v'$.
  Thus, $v' \gets_{T_2} v \barrow{\ast}{G'} u \to_{T_1} u'$, which contradicts
  the strong determinacy of $(T_1 \cup T_2)/G$ (noting that $G \conserveq G'$
  ensures the connectivity relation is preserved on $V_G$).
  The case $i=2$ is similar.

  \smallskip
  \noindent
  {\bf Strongly deterministic:}
  Suppose that $(T'_1 \cup T'_2) / G'$ is not strongly deterministic.
  This failure must be caused by the introduced edge $(v,v')$.

  In the case where $i=1$,
  we have $v' \gets_{T'_1} v \barrow{\ast}{G'} w \to_{T'_1 \cup T_2} w'$ and
  $v \neq w \lor v' \neq w'$ for some $w$ and $w'$,
  which implies $w \to_{(T'_1\setminus \{(v,v')\}) \cup T_2} w'$.
  Hence, from (i),
    we obtain $u' \gets_{T_1} u \barrow{\ast}{G'} w \to_{T_1 \cup T_2} w'$.
    Since $(u,u') \notin T'_1 \cup T_2$ but $(w,w') \in T'_1 \cup T_2$, it follows that $(u,u') \neq (w,w')$,
    which contradicts the strong determinacy of $(T_1 \cup T_2)/G$.

    The case $i=2$ is similar; i.e.,
    we have $v' \gets_{T'_2} v \barrow{\ast}{G'} w \to_{T_1 \cup T'_2} w'$ and
    $v \neq w \lor v' \neq w'$ for some $w$ and $w'$,
    which implies $w \to_{(T_1 \cup T'_2\setminus \{(v,v')\})} w'$.
    Hence, from (ii),
    we obtain $u' \gets_{T_2} u \barrow{\ast}{G'} w \to_{T_1 \cup T_2} w'$.
    Since $(u,u') \notin T_1 \cup T'_2$ but $(w,w') \in T_1 \cup T'_2$, it follows that $(u,u') \neq (w,w')$,
    which contradicts the strong determinacy of $(T_1 \cup T_2)/G$.

  \smallskip
  \noindent
  {\bf Acyclic:}
  Suppose that $(T'_1 \cup T'_2) / G'$ is not acyclic.
  This failure must be caused by the introduced edge $(v,v')$; hence, $(v,v')$ must
  be part of a cycle.

  In the case where $i=1$, by replacing $v \rarrow{}{T'_1} v'$ in the cycle with
  $v \larrow{\ast}{G'} u \rarrow{}{T_1} u' \barrow{\ast}{G'} v'$
  as obtained from (i), we obtain a cycle in $(T_2 \cup T_1) / G'$.
  Since $G \conserveq G'$, this implies the existence of a cycle in $(T_1 \cup T_2) / G$,
  which is a contradiction.

  In the case where $i=2$, we replace $v \rarrow{}{T'_2} v'$ in the cycle with
  $v \larrow{\ast}{G'} u \rarrow{}{T_2} u' \barrow{\ast}{T_1 \cup G'} v'$
  as obtained from (ii), which forms a cycle in $T_2 / (T_1 \cup G')$.
  Since $G \conserveq G'$, this implies a cycle in $T_2 / (T_1 \cup G)$.
  On the other hand, since $T_2/(T_1\cup G) \cong (T_2/G)/(T_1/G)$
  by Proposition~\ref{Prop:NestedQuotient},
  the presence of a cycle in $T_2/(T_1\cup G)$ implies that $(T_2/G)/(T_1/G)$ is cyclic as well.
    Applying the contraposition of Proposition~\ref{Prop:CycleInDeterministic} with
    $G_1 := T_1/G$ and $G_2 := T_2/G$,
    we obtain a cycle in $(T_1/G) \cup (T_2/G)$.
  Since $(T_1/G) \cup (T_2/G) = (T_1 \cup T_2) / G$,
  this contradicts the acyclicity of $(T_1 \cup T_2) / G$
  in the spanning graph $(G, T_1, T_2)$.
\end{proof}

\subsection{Termination of spanning-graph normalization}
\label{subsec:norm-termination}

Given a spanning graph ${\mathcal S} = (G,T_1,T_2)$ with confluent $G$,
we aim to obtain a spanning graph ${\mathcal S'} = (G',T'_1,T'_2)$
by normalizing spanning graphs that satisfies the following properties:
\begin{enumerate}
  \item $G'\ (\supseteq G)$ and $G' \cup T'_1 \cup T'_2$ are confluent, and
  \item $P({\mathcal S}, {\mathcal S'})$ holds for a given property $P$.
\end{enumerate}

We abbreviate the notation
$G\cup(w_0 \to w_1) \cup (w_1 \to w_2) \cup\cdots\cup (w_{k-1}\to w_k)$
as $G\cup(w_0 \to w_1 \to \cdots \to w_k)$.
Similarly, the notation $G\cup(w_0 \gets w_1 \gets \cdots \gets w_k)$ denotes
$G\cup(w_1 \to w_0) \cup (w_2 \to w_1) \cup\cdots\cup (w_k\to w_{k-1})$.

We introduce a measure $m({\mathcal S})$ on spanning graphs
${\mathcal S} = (G,T_1,T_2)$ to prove the termination of the normalization steps.

\begin{lemma}\label{lem:weight-stability}
  For a spanning graph ${\mathcal S} = (G,T_1,T_2)$,
  the weight function $w$ is stable under the normalization step
  ${\mathcal S} \vdash_n {\mathcal S'}$,
  i.e., $w_{\mathcal S}(t)=w_{\mathcal S'}(t)$ for any $t\in V_G$.
\end{lemma}
\begin{proof}
  By Lemma~\ref{prop:norm-basic}(\ref{prop:norm-basic-1}), we have $G \conserveq G'$.
  Suppose that the normalization step is denoted by ${\mathcal S} \vdash^{(u,u')}_{(v,v')} {\mathcal S'}$.
  If $(u,u') \in T_1$, then as in case (i) of the proof of Lemma~\ref{lem:reduction-welldef},
  the step preserves the structure of $(T_1 \cup T_2) / G$.
  If $(u,u') \in T_2$, then similarly to case (ii), the edge $u \to u'$ is reconnected to $v \to v'$ such that
  $u \barrow{\ast}{G'} u'$ and $v \barrow{\ast}{G' \cup T_1} v'$.
  In either case, the weight of each vertex is preserved by the definition of the weight function.
\end{proof}

\begin{definition}[Measure of a spanning graph]\label{def:measure}
  For a spanning graph ${\mathcal S} = (G,T_1,T_2)$, its \emph{measure} $m({\mathcal S})$ is defined as
  \[
  m({\mathcal S}) = \sum_{ \substack{(u,u') \in T_1 \cup T_2 \\ u \notin \pnf_{G}} } w_{\mathcal S}(u).
  \]
\end{definition}
It trivially holds that $m({\mathcal S}) \geq 0$ since the weight function
$w_{\mathcal S}$ takes non-negative values by Proposition~\ref{prop:weight}.

\begin{lemma}\label{lem:reduction}
  Let ${\mathcal S} = (G,T_1,T_2)$ be a spanning graph such that
  $G \conserveq \Vrt$ and $G$ is confluent.
  If ${\mathcal S} \vdash_n {\mathcal S'}$ then
  $m({\mathcal S}) > m({\mathcal S'})$.
\end{lemma}

\begin{proof}
  Let ${\mathcal S} \vdash^{(u,u')}_{(v,v')} {\mathcal S'}$
  with ${\mathcal S'} = (G',T'_1,T'_2)$.
  Lemma~\ref{lem:weight-stability} guarantees the stability of the weight
  function $w$, i.e., $w_{\mathcal S}(t) = w_{\mathcal S'}(t)$ for any $t \in V_G$.
  We use $x$ as introduced in Definition~\ref{def:normalization-step} and
  illustrated in Fig.~\ref{fig:normalization2}.

  Since $G \conserveq G'$ by Lemma~\ref{prop:norm-basic}(\ref{prop:norm-basic-1}),
  pseudo-normal forms are preserved by the normalization except for those in
  two components: the one is $\classB{u}_{\barrow{\ast}{G}}$ and
  the other is $\classB{x}_{\barrow{\ast}{G}}$ (see Fig.~\ref{fig:normalization2}).

  For the former component, $u \notin \pnf_{G}$ holds by the definition of the
  normalization step.
  On the other hand, $v \in \pnf_{G'}$ by Lemma~\ref{lem:reduction-conf}(\ref{lem:reduction-conf-3}).
  Combining these with the replacement of $(u,u') \in T_1 \cup T_2$ by
  $(v,v') \in T'_1 \cup T'_2$,
  the measure $m$ decreases by $w_{\mathcal S}(u)$ according to Definition~\ref{def:measure}.

  For the latter component, $\classB{x}_{\barrow{\ast}{G}}$ is extended to
  $\classB{x}_{\barrow{\ast}{G'}}$.
  There exists at most one outgoing edge $(s,s') \in T_1 \cup T_2$
  leaving $\classB{x}_{\barrow{\ast}{G}}$
  since $(T_1 \cup T_2) / G$ is deterministic.
  Thus, the normalization increases the measure $m$ by at most $w_{\mathcal S}(s)$
  for this component.
  Here, the maximum increase occurs when $s \in \pnf_G$ and $s \notin \pnf_{G'}$.

  From these arguments and the stability of $w$, we obtain
  \[
  m({\mathcal S'}) \leq m({\mathcal S}) - w_{\mathcal S}(u) + w_{\mathcal S}(s).
  \]
  Thus, the claim holds if $w_{\mathcal S}(u) > w_{\mathcal S}(s)$ is satisfied.
  \begin{itemize}
    \item In the case of $(u,u') \in T_1$, we have $u' \barrow{\ast}{G'} s$.
    Since $w_{\mathcal S}(u) = w_{\mathcal S}(u') + 1 = w_{\mathcal S}(s) + 1$,
    we obtain $w_{\mathcal S}(u) > w_{\mathcal S}(s)$.
    \item In the case of $(u,u') \in T_2$, we have $u' \barrow{\ast}{G' \cup T_1} s$.
    By definition, the weight is given by
    $w_{\mathcal S}(u) = \max\{ w_{\mathcal S}(u'') \mid u'' \barrow{\ast}{G' \cup T_1} u' \} + 1$.
    Since $s$ belongs to this set, it follows that
    $w_{\mathcal S}(u) \geq w_{\mathcal S}(s) + 1 > w_{\mathcal S}(s)$.
  \end{itemize}
  Therefore, we conclude $m({\mathcal S}) > m({\mathcal S'})$.
\end{proof}

\begin{theorem}\label{thm:termination-normalization}
  Spanning-graph normalization terminates.
\end{theorem}
\begin{proof}
  Suppose for contradiction that there exists an infinite normalization sequence
  ${\mathcal S}_1 \vdash_n {\mathcal S}_2 \vdash_n \cdots$.
  Clearly, the initial measure $m({\mathcal S}_1)$ is a finite integer.
  Moreover, by Lemma~\ref{lem:reduction}, each normalization step strictly decreases the measure,
  yielding an infinite strictly decreasing sequence of measures:
  $m({\mathcal S}_1) > m({\mathcal S}_2) > \cdots$.
  However, this contradicts the fact that the measure is always non-negative by Definition~\ref{def:measure}.
\end{proof}
For a spanning graph ${\mathcal S} = (G, T_1, T_2)$ with a confluent $G$,
if $m({\mathcal S}) = 0$, then $G \cup T_1 \cup T_2$ is confluent by Lemma~\ref{lem:vdash-norm}.
Thus, if a normalization step is applicable to any spanning graph satisfying $m({\mathcal S}) > 0$,
Theorem~\ref{thm:termination-normalization} guarantees that we can eventually obtain a spanning graph
whose components $G \cup T_1 \cup T_2$ are confluent.
The sufficient conditions to ensure the applicability of the normalization step $\vdash_n$
are provided by the notions of \emph{compatibility} and \emph{edge-commutativity}
introduced in Section~\ref{Sec:CompatCommute}.

\section{Compatibility and Edge Commutativity for Constructing Confluent
  Finite ARSs}
\label{sec6}\label{Sec:CompatCommute}

We first introduce two conditions for ARSs:
\emph{compatibility} (Definition~\ref{def:compatible} in
Section~\ref{subsec:compatibility}), which addresses overlay cases, and
\emph{edge commutativity} (Definition~\ref{def:edge-commutative} in
Section~\ref{subsec:commutativity}), which addresses non-top-overlapping cases.
The former is characterized by the property that a reduction at a redex is
deterministic in quotient graphs.
The latter is characterized by the property that a redex can be recovered after
its subpart is modified.
These conditions enable the normalization step $\vdash_n$
with respect to a confluent ARS $\Vrt$.

In Section~\ref{subsec:confluent-edge-norm},
we establish a condition for constructing a confluent ARS.
In that section,
we intend to define $\Hol$ and $\Vrt$ as a subset of top reductions
and a confluent set of inner reductions, respectively.
Accordingly, $T_2$ and $T_1$ represent finite subsets of collapsing and
non-collapsing edges in $\Hol$, respectively.

\subsection{Compatibility for Spanning Graph Construction}
\label{subsec:compatibility}

\begin{definition} \label{def:compatible}
  Let $\Vrt$, $N$, and $T'$ be ARSs such that $V_{T'} \subseteq V_{\Vrt}$.
  An ARS $T$ is said to be \emph{$N$-compatible over $\Vrt$}
  if $V_T \subseteq V_\Vrt$ and the following condition holds:
  for any $u,u',v,v' \in V_\Vrt$,
  $u' \larrow{}{T} u \barrow{\ast}{\Vrt} v \rarrow{}{T} v'$ implies
  $u' \barrow{\ast}{\Vrt \cup N} v'$.
\end{definition}
We simply say that $T$ is \emph{compatible over $\Vrt$}
when $N$ is the empty ARS.
\begin{figure}
  \begin{center}
    \begin{minipage}{0.59\textwidth}
      \begin{center}
  \begin{tikzpicture}[line width=0.07em, x=5, y=5]

    \node at (0,2.5)(u) {\small $u$};
    \node at (0,-2.5)(v) {$v$};

    \node at (12,2.5)(ud) {\small $u'$};
    \node at (12,-2.5)(vd) {\small $v'$};

    \path[<->] (u) edge node [left,yshift=0] {\scriptsize $\ast$}
    node [right,yshift=-6] {\scriptsize $\Vrt$}  (v);
    \path[->] (u) edge node [xshift=0,yshift=-8] {\scriptsize $T$} (ud);

    \path[->] (v) edge node [xshift=0,yshift=-8] {\scriptsize $T$} (vd);

    \path[<->,color=red]
    (ud) edge
    node[below,xshift=18] {\scriptsize $\Vrt \cup N$}
    node[xshift=-6]  {\scriptsize $\ast$}
    (vd);
  \end{tikzpicture}
      \end{center}
    \end{minipage}
    \caption{$N$-Compatibility of $T$ over $\Vrt$}\label{fig:compatibility}
  \end{center}
\end{figure}
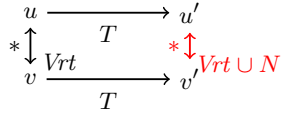

\begin{proposition}\label{prop:preserveCompatibility}
  Let $\Vrt$, $T_1$, $T_2$, $T'_1$, and $T'_2$ be ARSs
  such that $T_1 \supseteq T'_1$, $T_2 \supseteq T'_2$, and
  ${\barrow{\ast}{T_1 \cup \Vrt}} = {\barrow{\ast}{T'_1 \cup \Vrt}}$.
  Then, the following properties hold:
  \begin{enumerate}
    \item If $T_1$ is compatible over $\Vrt$, then so is $T'_1$.
    \item If $T_1 \cup T_2$ is $T_1$-compatible over $\Vrt$,
    then $T'_1 \cup T'_2$ is $T'_1$-compatible over $\Vrt$.
  \end{enumerate}
\end{proposition}
\begin{proof}
  \mbox{}
  \begin{enumerate}
  \item Assume $u' \larrow{}{T'_1} u \barrow{\ast}{\Vrt} v \rarrow{}{T'_1} v'$.
    Since $T_1 \supseteq T'_1$, we trivially have
    $u' \larrow{}{T_1} u \barrow{\ast}{\Vrt} v \rarrow{}{T_1} v'$.
    Because $T_1$ is compatible over $\Vrt$, we immediately obtain
    $u' \barrow{\ast}{\Vrt} v'$.
  \item Assume $u'\larrow{}{T'_1\cup T'_2} u \barrow{\ast}{\Vrt}
    v\rarrow{}{T'_1\cup T'_2} v'$.
    Since $T_1 \supseteq T'_1$ and $T_2 \supseteq T'_2$, it follows that
    $u' \larrow{}{T_1 \cup T_2} u \barrow{\ast}{\Vrt} v \rarrow{}{T_1 \cup T_2}
    v'$.
    Since $T_1 \cup T_2$ is $T_1$-compatible over $\Vrt$,
    we obtain $u' \barrow{\ast}{T_1 \cup \Vrt} v'$.
    Combining this with the assumption
    ${\barrow{\ast}{T_1 \cup \Vrt}} = {\barrow{\ast}{T'_1 \cup \Vrt}}$,
    we conclude that $u' \barrow{\ast}{T'_1 \cup \Vrt} v'$.
  \end{enumerate}
\end{proof}

The following lemma ensures that we can eliminate redundant edges from $T_1$ and
$T_2$ so that the resulting quotient graph is strongly deterministic and acyclic.

\begin{lemma} \label{lem:treeshaped2}
  Let $T_i$ ($i=1,2$) be ARSs such that:
  \begin{enumerate}
    \item[(i)] $T_1$ is compatible over $\Vrt$, and
    \item[(ii)] $T_1 \cup T_2$ is $T_1$-compatible over $\Vrt$.
  \end{enumerate}
  Then, there exist ARSs $T'_i$ ($i=1,2$) such that the following properties
  hold:
  \begin{enumerate}
    \item $T_i \supseteq T'_i$ ($i=1,2$),
    \item ${\barrow{\ast}{T_1\cup \Vrt}} = {\barrow{\ast}{T'_1\cup \Vrt}}$,
    \item ${\barrow{\ast}{T_1\cup T_2\cup \Vrt}} =
      {\barrow{\ast}{T'_1\cup T'_2\cup \Vrt}}$, and
    \item $(T'_1 \cup T'_2)/\Vrt$ is strongly deterministic, acyclic,
      and satisfies $\rarrow{}{T'_1} \cap \rarrow{}{T'_2} = \emptyset$.
  \end{enumerate}
\end{lemma}
\begin{proof}
  We consider the following conversion rules on pairs of ARSs:
  \[
  \begin{array}{l@{\hspace{10pt}}l}
    \infer[
      \text{if $u' \larrow{}{T_1} u \barrow{\ast}{\Vrt} v \rarrow{}{T_1} v'$ and
        $u \neq v \lor u' \neq v'$}]
          {(T_1 \setminus (v \rarrow{}{} v')), T_2}
          {(T_1, T_2)}
          & \raise1.8ex\hbox{(A$_{11}$)} \\

    \infer[
      \text{if $u' \larrow{}{T_1} u \barrow{\ast}{\Vrt} v \rarrow{}{T_2} v'$}]
          {(T_1, T_2 \setminus (v \rarrow{}{} v'))}
          {(T_1, T_2)}
          & \raise1.8ex\hbox{(A$_{12}$)} \\

    \infer[
      \text{if $u' \larrow{}{T_2} u \barrow{\ast}{\Vrt} v \rarrow{}{T_2} v'$
        and $u \neq v \lor u' \neq v'$}]
          {(T_1, T_2 \setminus (v \rarrow{}{} v'))}
          {(T_1, T_2)}
          & \raise1.8ex\hbox{(A$_{22}$)} \\

    \infer[
      \text{if $(u,u') \in T_1$ and $u' \mathrel{(\barrow{\ast}{\Vrt}
          \rarrow{}{T_1} \barrow{\ast}{\Vrt})^\ast} u$}]
          {(T_1 \setminus (u \rarrow{}{} u'), T_2)}
          {(T_1,T_2)}
          & \raise1.8ex\hbox{(B)} \\

   \infer[
     \text{if $(u,u') \in T_2$ and $u' \mathrel{(\barrow{\ast}{\Vrt}
         \rarrow{}{T_1 \cup T_2}\barrow{\ast}{\Vrt})^\ast} u$}]
         {(T_1, T_2 \setminus (u \rarrow{}{} u'))}
         {(T_1,T_2)}
         & \raise1.8ex\hbox{(C)} \\
  \end{array}
  \]
  We denote the resulting pair in the conclusion of each rule by $(T'_1, T'_2)$.

  First, we show that each single conversion step preserves properties (1)--(3).
  Since property (1) holds trivially, we focus on property (2) for each rule:
  \begin{description}
    \item[Rule (A$_{11}$):] Since $T_1$ is compatible over $\Vrt$,
      the peak $u' \larrow{}{T_1} u \barrow{\ast}{\Vrt} v \rarrow{}{T_1} v'$
      implies $u' \barrow{\ast}{\Vrt} v'$. Since the edge $u \rarrow{}{T_1} u'$
      is not removed, it remains in $T'_1$.
      Thus, we have $v \barrow{\ast}{\Vrt} u \barrow{\ast}{T'_1 \cup \Vrt} u'
      \barrow{\ast}{\Vrt} v'$, which gives $v \barrow{\ast}{T'_1 \cup\Vrt} v'$.
    \item[Rules (A$_{12}$), (A$_{22}$), and (C):]
      These cases are trivial for property (2) because $T_1 = T'_1$.
    \item[Rule (B):]
      The removed edge $u \rarrow{}{T_1} u'$ satisfies the condition
      $u' \mathrel{(\barrow{\ast}{\Vrt} \rarrow{}{T_1}\barrow{\ast}{\Vrt})^\ast}
      u$.
      By choosing a shortest sequence among those satisfying this condition,
      it follows by the minimality that no $T_1$-edge in the sequence is removed
      in this step.
      Hence, all these $T_1$-edges belong to $T'_1$, resulting in
      $u' \barrow{\ast}{T'_1 \cup \Vrt} u$ and ensuring that the conversion
      relation is preserved.
  \end{description}

  Next, we verify property (3) for each rule:
  \begin{description}
  \item[Rule (A$_{11}$):] The preservation of property (3) directly follows
    from property (2),
    which has already been established.
  \item[Rules (A$_{12}$), (A$_{22}$):]
    In these rules, an edge $v \rarrow{}{T_2} v'$ is removed under the
    precondition $u' \larrow{}{T_i} u \barrow{\ast}{\Vrt} v \rarrow{}{T_2} v'$
    for $i \in \{1,2\}$.
    Since $T_1 \cup T_2$ is $T_1$-compatible over $\Vrt$,
    this peak implies $u' \barrow{\ast}{T_1 \cup \Vrt} v'$.
    Since neither of these rules modifies $T_1$,
    we have $T_1 = T'_1$, which implies $u' \barrow{\ast}{T'_1 \cup \Vrt} v'$.
    Furthermore, because the edge $u \rarrow{}{T_i} u'$ remains intact,
    it belongs to $T'_1 \cup T'_2$.
    Therefore, we obtain $v \barrow{\ast}{\Vrt} u \barrow{}{T'_1 \cup T'_2} u'
    \barrow{\ast}{T'_1 \cup \Vrt} v'$,
    which establishes $v \barrow{\ast}{T'_1 \cup T'_2 \cup \Vrt} v'$.
  \item[Rule (B):] As shown in the proof of property (2), the removed edge
    $u \rarrow{}{T_1} u'$ satisfies $u' \barrow{\ast}{T'_1 \cup \Vrt} u$,
    which immediately implies $u \barrow{\ast}{T'_1 \cup T'_2 \cup \Vrt} u'$.
  \item[Rule (C):] Similarly to Rule (B), the removed edge
    $u \rarrow{}{T_2} u'$ satisfies $u' \mathrel{(\barrow{\ast}{\Vrt}
    \rarrow{}{T_1 \cup T_2} \barrow{\ast}{\Vrt})^\ast} u$.
    By choosing a shortest sequence and applying a minimality argument analogous
    to that for Rule (B), we ensure that all edges in the sequence remain
    intact and thus belong to $T'_1 \cup T'_2$. Consequently,
    we have $u' \barrow{\ast}{T'_1 \cup T'_2 \cup \Vrt} u$,
    ensuring the preservation of the relation.
  \end{description}

  On the other hand, each single conversion step preserves the preconditions (i)
  and (ii), which follows from properties (1)--(3) by
  Proposition~\ref{prop:preserveCompatibility}.
  Consequently, properties (1)--(3) are inductively preserved throughout the
  entire conversion process.

  Finally, we establish property (4). If $(T_1 \cup T_2)/\Vrt$ is not strongly
  deterministic, one of the rules (A$_{ij}$) is applicable.
  If $(T_1 \cup T_2)/\Vrt$ is not acyclic, then either Rule (B) or Rule (C) is
  applicable.
  If $\rarrow{}{T_1/\Vrt} \cap \rarrow{}{T_2/\Vrt} \neq \emptyset$,
  Rule (A$_{12}$) is applicable.
  Since $T_1$ and $T_2$ are finite sets of edges and each rule monotonically
  decreases the number of edges, this conversion process must terminate.
  Therefore, the irreducible final systems $T'_1$ and $T'_2$ satisfy property
  (4).
\end{proof}

\begin{lemma} \label{lem:make-spanning-graph}
  Let $\Vrt$ be an ARS, and let $G$ be a finite ARS such that
  $G \subseteq \Vrt$.
  Suppose $T_i$ ($i=1,2$) are finite ARSs satisfying
  $V_{T_1 \cup T_2} \subseteq V_G$.
  Assume that:
  \begin{enumerate}
    \item[(i)] $T_1$ is compatible over $\Vrt$, and
    \item[(ii)] $T_1 \cup T_2$ is $T_1$-compatible over $\Vrt$.
  \end{enumerate}
  Then, there exists a spanning graph $(G',T'_1,T'_2)$ such that:
  \begin{enumerate}
    \item $T_i \supseteq T'_i \ (i=1,2)$,
    \item $G \subseteq G' \conserveq \Vrt$,
    \item ${\barrow{\ast}{T_1\cup G}} \subseteq {\barrow{\ast}{T'_1\cup G'}}$,
      and
    \item ${\barrow{\ast}{T_1\cup T_2\cup G}} \subseteq
      {\barrow{\ast}{T'_1\cup T'_2\cup G'}}$.
  \end{enumerate}
  Furthermore, if $\Vrt$ is confluent, then $G'$ can be chosen to be confluent.
\end{lemma}
\begin{proof}
  First, by applying Lemma~\ref{lem:treeshaped2}, we obtain ARSs $T'_1$ and $T'_2$
  such that:
  \begin{align}
    & T_i \supseteq T'_i, \label{eq:lem:make-spanning-graph_subTi} \\
    & {\barrow{\ast}{T_1\cup\Vrt}} = {\barrow{\ast}{T'_1\cup\Vrt}},
    \label{eq:lem:make-spanning-graph_T1d} \\
    & {\barrow{\ast}{T_1\cup T_2\cup\Vrt}} =
      {\barrow{\ast}{T'_1\cup T'_2\cup\Vrt}},
      \label{eq:lem:make-spanning-graph_T12d} \\
    & (T'_1\cup T'_2)/\Vrt \text{ is acyclic, strongly deterministic, and }
      \rarrow{}{T'_1} \cap \rarrow{}{T'_2} = \emptyset.
      \label{eq:lem:make-spanning-graph_st-det}
  \end{align}
  Next, by Proposition~\ref{prop:conserv2}, we obtain an ARS $G'$ satisfying:
  \begin{align}
    & G \subseteq G' \conserveq \Vrt. \label{eq:lem:make-spanning-graph_Gd}
  \end{align}

  We first verify that $(G', T'_1, T'_2)$ forms a spanning graph.
  The relation $G' \conserveq \Vrt$ follows directly from
  \eqref{eq:lem:make-spanning-graph_Gd}.
  The inclusion
  $V_{T'_1\cup T'_2} \subseteq V_{T_1\cup T_2} \subseteq V_{G} \subseteq V_{G'}$
  holds by assumption, and the disjointness condition
  $\rarrow{}{T'_1} \cap \rarrow{}{T'_2} = \emptyset$
  is guaranteed by \eqref{eq:lem:make-spanning-graph_st-det}.
  Since $V_{T'_1 \cup T'_2} \subseteq V_{G'}$ and $G' \conserveq \Vrt$,
  we have an isomorphism on quotient sets:
  $V_{T'_1 \cup T'_2} / G' \cong V_{T'_1 \cup T'_2} / \Vrt$.
  Combining this with \eqref{eq:lem:make-spanning-graph_st-det},
  it follows that $(T'_1\cup T'_2)/G'$ is acyclic and strongly deterministic.
  Therefore, $(G', T'_1, T'_2)$ is indeed a spanning graph.

  It remains to establish properties (1)--(4). Properties (1) and (2) are
  immediately given by \eqref{eq:lem:make-spanning-graph_subTi}
  and \eqref{eq:lem:make-spanning-graph_Gd}, respectively.
  Property (3) follows from \eqref{eq:lem:make-spanning-graph_Gd},
  \eqref{eq:lem:make-spanning-graph_T1d}, and
  Proposition~\ref{prop:ConservUnion}, which together induce the following chain:
  \[
    {\barrow{\ast}{T_1\cup G}} = {\barrow{\ast}{T_1\cup \Vrt}}|_{V_G}
    \subseteq {\barrow{\ast}{T_1\cup \Vrt}}|_{V_{G'}}
    = {\barrow{\ast}{T'_1\cup \Vrt}}|_{V_{G'}}
    = {\barrow{\ast}{T'_1\cup G'}}.
  \]
  Property (4) is analogously established using
  \eqref{eq:lem:make-spanning-graph_Gd},
  \eqref{eq:lem:make-spanning-graph_T12d},
  and Proposition~\ref{prop:ConservUnion}.

  Furthermore, if $\Vrt$ is confluent, we can replace $G'$ with a confluent ARS
  $G''$ by applying Lemma~\ref{lem:merge}, which completes the proof.
\end{proof}

\subsection{Edge commutativity and its preservation by normalization}
\label{subsec:commutativity}

Edge commutativity is the property ensuring that red edges exist for any blue
edge, as illustrated in Fig.~\ref{fig:normalization}.
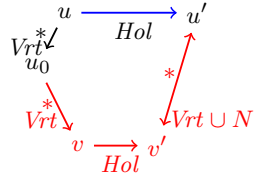
\begin{figure}
  \begin{center}
    \begin{minipage}{0.59\textwidth}
      \begin{center}
        \begin{tikzpicture}[line width=0.07em, x=5, y=5]

          \node at (1,2)(u) {\small $u$};
          \node at (-1,-2)(uz) {$u_0$};
          \node[color=red] at (2,-8)(v) {\small $v$};
          \node at (11,1)(v2) {};
          \node at (11,2)(ud) {\small $u'$};
          \node[color=red] at (8,-8)(vd) {\small $v'$};

          \path[->,color=blue] (u) edge (ud);
          \node[below,xshift=-24] at (ud) (x) {\scriptsize $\Hol$};
          \path[->] (u) edge node [xshift=-4,yshift=2] {\scriptsize $\ast$}
          node [xshift=-10,yshift=-3] {\scriptsize $\Vrt$}  (uz);
          \path[->,color=red] (uz) edge node [xshift=-4,yshift=1]
               {\scriptsize $\ast$} node [xshift=-7,yshift=-5]
               {\scriptsize $\Vrt$}  (v);

          \path[<->,color=red]
          (ud) edge node [xshift=-3,yshift=2] {\scriptsize $\ast$} (vd);
          \node[color=red,xshift=20,yshift=10] at (vd) (x3cd)
               {\scriptsize $\Vrt\cup N$};

          \path[->,color=red] (v) edge (vd);
          \node[color=red,below,xshift=-14] at (vd) (udvd) {\scriptsize $\Hol$};
        \end{tikzpicture}
      \end{center}
    \end{minipage}

    \medskip
    \caption{Edge commutativity of $(u,v)$ over $(\Vrt, \Hol)$}
    \label{fig:normalization}
  \end{center}
\end{figure}

\begin{definition}\label{def:edge-commutative}
  Let $\Vrt$, $\Hol$, and $N$ be ARSs.
  An edge $(u,u') \in {\rarrow{}{\Hol}}$
  is \emph{$N$-edge-commutative over $(\Vrt,\Hol)$}
  if, for any $w\in V_{\Vrt}$ with $u \Rarrow{\ast}{}{}{\Vrt} w$,
  there exists $(v,v') \in {\rarrow{}{\Hol}}$ such that
  \begin{enumerate}
    \item $w \rarrow{\ast}{\Vrt} v$ and
    \item $u' \barrow{\ast}{\Vrt \cup N} v'$.
  \end{enumerate}
  Moreover, if condition (2') $u' \barrow{\ast}{\Vrt} v'$, which is stronger
  than (2),  holds, we simply say that the edge $(u,u')$ is
  \emph{edge-commutative over $(\Vrt,\Hol)$}.

  We say that an ARS $T \subseteq \Hol$ is edge-commutative
  (resp., $N$-edge-commutative) over $(\Vrt,\Hol)$ if every edge
  $(u,u') \in {\rarrow{}{T}}$ is so.
\end{definition}

\begin{proposition}\label{prop:preserveCommutativity}
  Let $T_1$, $T_2$, $T'_1$, $T'_2$, and $\Vrt$ be ARSs
  such that $T_1 \supseteq T'_1$, $T_2 \supseteq T'_2$, and
  ${\barrow{\ast}{T_1 \cup \Vrt}} = {\barrow{\ast}{T'_1 \cup \Vrt}}$.
  Then, the following hold:
  \begin{enumerate}
  \item If $T_1$ is edge-commutative over $(\Vrt,\Hol)$, then so is $T'_1$.
  \item If $T_2$ is $T_1$-edge-commutative over $(\Vrt,\Hol)$, then $T'_2$ is
    $T'_1$-edge-commutative over $(\Vrt,\Hol)$.
  \end{enumerate}
\end{proposition}

\begin{proof}
  Property (1) is clear from $T_1 \supseteq T'_1$.
  Property (2) also follows from $T_2 \supseteq T'_2$
  and ${\barrow{\ast}{T_1 \cup \Vrt}} = {\barrow{\ast}{T'_1 \cup \Vrt}}$.
\end{proof}
We sometimes omit "over $(\Vrt, \Hol)$" if it is clear from the context.
In the subsequent sections, we use the empty reduction graph
$\<\emptyset,\emptyset\>$ for $N$ when normalizing a non-collapsing edge
$(u,u')$.
Conversely, we use a reduction graph consisting of non-collapsing top edges for $N$ when normalizing collapsing edges.

\begin{lemma}\label{lem:stability-of-commutativity}
  Let $G \conserveq \Vrt$ and $T_1 \cup T_2 \subseteq \Hol$.
  Let $(G,T_1,T_2) \vdash^{(u,u')}_{(v,v')} (G',T'_1,T'_2)$ be a normalization with
  respect to $\Vrt$ such that $(v,v')\in\Hol$.
  Then, the normalization preserves the following properties.
  \begin{enumerate}
  \item $T_1$ is edge-commutative over $(\Vrt,\Hol)$.
  \item $T_2$ is $T_1$-edge-commutative over $(\Vrt,\Hol)$.
  \end{enumerate}
\end{lemma}
\begin{proof}
  From (v) of Definition~\ref{def:normalization-step}, we have
  $(u,u') \in T_i \iff (v,v') \in T'_i$ and
  ${\to_{T_1\cup T_2}}\setminus\{(u,u')\} = {\to_{T'_1\cup T'_2}}\setminus\{(v,v')\}$.

\begin{enumerate}
  \item If $(u,u') \in {\to_{T_2}}$, the claim is trivial because $T_1=T'_1$.
  For the case where $(u,u') \in {\to_{T_1}}$,
  we show that $(v,v') \in {\to_{T'_1}}$ is edge-commutative.
  Suppose $v\rarrow{\ast}{\Vrt} w$.
  Since $u \rarrow{\ast}{\Vrt} v$ follows from
  Definition~\ref{def:normalization-step}(i) and $G \subseteq \Vrt$,
  we obtain $u\rarrow{\ast}{\Vrt} w$.
  Combining this with the edge-commutativity of $(u,u') \in {\to_{\Hol}}$,
  there exists $(s,s') \in {\Hol}$ such that $w \rarrow{\ast}{\Vrt} s$ and
  $u' \barrow{\ast}{\Vrt} s'$.
  Here, $u'\barrow{\ast}{\Vrt} v'$ holds from conditions (ii) and (iii) of
  Definition~\ref{def:normalization-step} and $G \subseteq \Vrt$.
  Thus, $v' \barrow{\ast}{\Vrt} s'$, and hence,
  $(v,v') \in {\to_{T'_1}}$ is edge-commutative over $(\Vrt,\Hol)$.

  \item
  If $(u,u') \in {\to_{T_1}}$, then $T'_2 = T_2$ holds.
  From Lemma~\ref{prop:norm-basic}(\ref{prop:norm-basic-0}) and
  $G \subseteq G' \subseteq \Vrt$,
  we obtain ${\barrow{\ast}{\Vrt \cup T_1}} = {\barrow{\ast}{\Vrt \cup T'_1}}$.
  Since $T_2$ is $T_1$-edge-commutative, this equality ensures that $T'_2$ is
  $T'_1$-edge-commutative.

  If $(u,u') \in {\to_{T_2}}$, similarly to the proof of (1),
  we show that $(v,v')$ is $T'_1$-edge-commutative over $(\Vrt,\Hol)$.
  Suppose $v\rarrow{\ast}{\Vrt} w$.
  Since $u \rarrow{\ast}{\Vrt} v$ follows from
  Definition~\ref{def:normalization-step}(i) and $G \subseteq \Vrt$,
  we obtain $u\rarrow{\ast}{\Vrt} w$.
  Combining this with the $T_1$-edge-commutativity of $(u,u') \in {\to_{\Hol}}$,
  there exists $(s,s') \in {\Hol}$ such that $w \rarrow{\ast}{\Vrt} s$ and
  $u' \barrow{\ast}{\Vrt \cup T_1} s'$.
  Here, $u'\barrow{\ast}{\Vrt \cup T'_1} v'$ follows from conditions (ii) and
  (iii) of Definition~\ref{def:normalization-step},
  $G \subseteq \Vrt$, and $T_1 = T'_1$.
  Thus, $v' \barrow{\ast}{\Vrt \cup T_1} s'$, and hence,
  $(v,v') \in {\to_{T'_2}}$ is $T'_1$-edge-commutative over $(\Vrt,\Hol)$.
\end{enumerate}
\end{proof}

\subsection{Making Finite ARSs Confluent}
\label{subsec:confluent-edge-norm}

In this subsection, we discuss how to normalize a spanning graph $(G, T_1, T_2)$
so that $G$ and $G \cup T_1 \cup T_2$ are finite and confluent
(Theorem~\ref{thm:seq-edge-norm}).

\begin{lemma}\label{lem:norm-possible}
  Let $\Hol$ and $\Vrt$ be ARSs.
  For finite ARSs $G$, $T_1$, and $T_2$,
  suppose that ${\mathcal S} = (G,T_1,T_2)$ is a spanning graph with
  $m({\mathcal S})>0$.
  Assume that
  \begin{enumerate}
    \item[(i)] $T_1$ is edge-commutative over $(\Vrt,\Hol)$,
    \item[(ii)] $T_2$ is $T_1$-edge-commutative over $(\Vrt,\Hol)$,
    \item[(iii)] $G$ and $\Vrt$ are confluent, and
    \item[(iv)] $G$ is finite.
  \end{enumerate}
  Then, there exists a spanning graph ${\mathcal S}' = (G',T'_1,T'_2)$
  such that ${\mathcal S} \vdash_n {\mathcal S}'$,
  satisfying $G \subseteq G' \conserveq \Vrt$ and
  ${\barrow{\ast}{G\cup T_1\cup T_2}} \subseteq
  {\barrow{\ast}{G'\cup T'_1\cup T'_2}}$.
  Furthermore, ${\mathcal S}'$ satisfies conditions (i)--(iv).
\end{lemma}

\begin{proof}
  Since $m({\mathcal S})>0$,
  there exist $u \notin\pnf_G$ and $i \in \{1,2\}$ such that $(u,u') \in T_i$.

  Since $G$ is confluent (from (iii)), $G_{\ddiv}$, i.e., $G/{\stcongstarss{G}}$
  is confluent by
  Proposition~\ref{prop:termination-looping}\eqref{prop:termination-looping-1}.
  Since $G$ is finite, $G_{\ddiv}$ is terminating by
  Proposition~\ref{prop:termination-looping}\eqref{prop:termination-looping-3}.
  Hence, there uniquely exists a class $\classD{w}_{\stcongstarss{G}}$
  with $w \in \pnf_G$ such that $u \rarrow{\ast}{G} w$.
  Since $G \subseteq \Vrt$, we have $u \rarrow{\ast}{\Vrt} w$.
  Then, from assumptions (i) and (ii), there exists $(v,v') \in \Hol$ such that
  $w \rarrow{\ast}{\Vrt} v$, where $u' \barrow{\ast}{\Vrt} v'$ if $i = 1$, and
  $u' \barrow{\ast}{T_1 \cup \Vrt} v'$ if $i = 2$.

  \begin{description}
    \item[\textbf{Case 1:} $i=1$.]
    Let $y$ be a common reduct of $u'$ and $v'$ in $\Vrt$,
    i.e., $u' \rarrow{\ast}{\Vrt} y \larrow{\ast}{\Vrt} v'$.
    Then, by setting $x := u'$ in Definition~\ref{def:normalization-step},
    we obtain $(G,T_1,T_2) \vdash^{(u,u')}_{(v,v')} (G',T'_1,T'_2)$.

    \item[\textbf{Case 2:} $i=2$.]
    Let $u' \barrow{\ast}{T_1 \cup \Vrt} v'$ be a sequence of the form
    \[
    u' = u_0 \barrow{\ast}{\Vrt} v_0 \barrow{\ast}{T_1} u_1 \barrow{\ast}{\Vrt}
    v_1 \barrow{\ast}{T_1} \cdots \barrow{\ast}{T_1} u_k \barrow{\ast}{\Vrt}
    v_k \barrow{\ast}{T_1} u_{k+1} \barrow{\ast}{\Vrt} v'.
    \]
    Let $y$ be a common reduct of $u_{k+1}$ and $v'$ in $\Vrt$,
    i.e., $u_{k+1} \rarrow{\ast}{\Vrt} y \larrow{\ast}{\Vrt} v'$.
    Then, by setting $x := u_{k+1}$ in Definition~\ref{def:normalization-step},
    we obtain $(G,T_1,T_2) \vdash^{(u,u')}_{(v,v')} (G',T'_1,T'_2)$.
  \end{description}
  Thus, $(G',T'_1,T'_2)$ is a spanning graph by
  Lemma~\ref{lem:reduction-welldef},
  and also $G\subseteq G' \conserveq \Vrt$ and
  ${\barrow{\ast}{G\cup T_1\cup T_2}} \subseteq
  {\barrow{\ast}{G'\cup T'_1\cup T'_2}}$ hold
  by Lemma~\ref{lem:reduction-conf}(\ref{itm:seq-edge-confluence}) and
  (\ref{itm:seq-edge-inclusion}).

  Finally, the normalization $(G,T_1,T_2) \vdash^{(u,u')}_{(v,v')} (G',T'_1,T'_2)$
  preserves properties (i) and (ii) by
  Lemma~\ref{lem:stability-of-commutativity}, and properties (iii) and (iv) by
  Lemma~\ref{lem:reduction-conf}(\ref{lem:reduction-conf-4}).
\end{proof}

\begin{lemma}\label{lem:vdash-norm}
  Let ${\mathcal S} = (G,T_1,T_2)$ be a spanning graph.
  If $G$ is confluent and $m({\mathcal S}) = 0$ (i.e., $u \in \pnf_G$ for each
  $(u,u') \in T_1 \cup T_2$), then $G \cup T_1 \cup T_2$ is confluent.
\end{lemma}

\begin{proof}
  Let $V_k = \{ t \in V_G \mid w_{\mathcal S}(t) \leq k \}$.
  We will show confluence of $G_k = (G \cup T_1 \cup T_2) |_{V_k}$
  by induction on $k \geq 0$.

  Assuming $s \rarrow{\ast}{G_k} t$ and $s \rarrow{\ast}{G_k} t'$,
  we have $w_{\mathcal S}(s) \geq w_{\mathcal S}(t)$ and
  $w_{\mathcal S}(s) \geq w_{\mathcal S}(t')$ by
  Proposition~\ref{prop:rewrite-decrease}(1).
  We consider the following cases based on the value of $w_{\mathcal S}(s)$.
  \begin{itemize}
    \item If $w_{\mathcal S}(s) < k$, then $s \rarrow{\ast}{G_{k-1}} t$ and
    $s \rarrow{\ast}{G_{k-1}} t'$ hold. Hence, $t$ and $t'$ have a common reduct
    in $G_{k-1}$ by the induction hypothesis.
    \item Otherwise, i.e., $w_{\mathcal S}(s) = k$:
    \begin{itemize}
      \item If $w_{\mathcal S}(t) = w_{\mathcal S}(t') = k$, then
        $s \rarrow{\ast}{G} t$ and $s \rarrow{\ast}{G} t'$ by
        Proposition~\ref{prop:rewrite-decrease}(2).
        Since $G$ is confluent, $t$ and $t'$ have a common reduct in $G$, and
        thus in $G_k$.
      \item If $w_{\mathcal S}(t) = k > w_{\mathcal S}(t')$
        (or vice versa by symmetry), then there exists
        $(u,v) \in T_1 \cup T_2$ such that
        $s \rarrow{\ast}{G} u \rarrow{}{T_1 \cup T_2} v
        \rarrow{\ast}{G_{k-1}} t'$.
        From $(u,v) \in T_1 \cup T_2$ and the premise, we have $u \in \pnf_G$,
        i.e., $u \Rrightarrow_G u$.
        From confluence of $G$ and $s \rarrow{\ast}{G} u$,
        we have $s \Rrightarrow_{G} u$ by
        Proposition~\ref{prop:pnf}~(\ref{subprop:pnf-3}).
        Combining this with $s \rarrow{\ast}{G} t$,
        this implies $t \rarrow{\ast}{G} u$. Hence, $t \rarrow{\ast}{G_k} t'$
        holds, meaning that $t'$ itself serves as a common reduct.
      \item If $w_{\mathcal S}(t) < k$ and $w_{\mathcal S}(t') < k$, then
        there exist $(u,v), (u',v') \in T_1 \cup T_2$ such that
        $s\rarrow{\ast}{G} u \rarrow{}{T_1 \cup T_2} v \rarrow{\ast}{G_{k-1}} t$
        and
        $s\rarrow{\ast}{G} u'\rarrow{}{T_1 \cup T_2} v'\rarrow{\ast}{G_{k-1}} t'$
        by Proposition~\ref{prop:rewrite-decrease}(3).
        On the other hand, $v = v'$ holds by the determinacy of
        $(T_1 \cup T_2) / G$.
        Hence, $t$ and $t'$ have a common reduct in $G_{k-1}$ by the
        induction hypothesis.
    \end{itemize}
  \end{itemize}
\end{proof}

Now, we are ready to present the main theorem of this section.

\begin{theorem}\label{thm:seq-edge-norm}
  Let $G_1$ be a finite ARS and $\Vrt$ be a confluent ARS such that $G_1 \subseteq \Vrt$.
  Suppose that $T_1$ and $T_2$ are finite ARSs with $T_1, T_2 \subseteq \Hol$ and
  $V_{T_1 \cup T_2} \subseteq V_{G_1}$.
  Assume that
  \begin{enumerate}
    \item[(i)] $T_1$ is compatible over $\Vrt$,
    \item[(ii)] $T_1 \cup T_2$ is $T_1$-compatible over $\Vrt$,
    \item[(iii)] $T_1$ is edge-commutative over $(\Vrt,\Hol)$, and
    \item[(iv)] $T_2$ is $T_1$-edge-commutative over $(\Vrt,\Hol)$.
  \end{enumerate}
  Then, there exist a confluent finite ARS $G'_1$ and a finite ARS $T' \subseteq \Hol$ such that
  \begin{enumerate}
    \item\label{itm:seq-edge-dash} $G_1 \subseteq G'_1 \conserveq \Vrt$,
    \item\label{itm:seq-edge-subset-S} $V_{T'} \subseteq V_{G'_1}$,
    \item\label{itm:seq-edge-inclusion} ${\barrow{\ast}{G_1\cup T_1\cup T_2}} \subseteq {\barrow{\ast}{G'_1\cup T'}}$, and
    \item\label{itm:seq-edge-confluence} $G'_1 \cup T'$ is confluent and finite.
  \end{enumerate}
\end{theorem}

\begin{proof}
  By Lemma~\ref{lem:make-spanning-graph}, there exists a spanning graph ${\mathcal S} = (G',T'_1,T'_2)$ satisfying
  \begin{align}
    & T_i \supseteq T'_i \quad (i=1,2)       \label{eq:thm:seq-edge-norm_subTi} \\
    & \text{$G'$ is confluent with } G_1 \subseteq G' \conserveq \Vrt
    \label{eq:thm:seq-edge-norm_conserv} \\
    & {\barrow{\ast}{G_1\cup T_1}} \subseteq {\barrow{\ast}{G'\cup T'_1}}
    \label{eq:thm:seq-edge-norm_equvT1} \\
    & {\barrow{\ast}{G_1\cup T_1\cup T_2}} \subseteq {\barrow{\ast}{G'\cup T'_1 \cup T'_2}}.
    \label{eq:thm:seq-edge-norm_equvT12}
  \end{align}
  Here, by Proposition~\ref{prop:preserveCommutativity}, \eqref{eq:thm:seq-edge-norm_subTi},
  and \eqref{eq:thm:seq-edge-norm_equvT1}, the properties (iii) and (iv) on commutativity are preserved.

  Now, we repeatedly apply Lemma~\ref{lem:norm-possible} to the spanning graph $\mathcal S$ as long as $m({\mathcal S}) > 0$.
  The conditions (i)--(iv) of Lemma~\ref{lem:norm-possible} hold initially.
  Moreover, Lemma~\ref{lem:norm-possible} guarantees the preservation of these conditions,
  which justifies the repeated applications of the lemma.
  Thus, we obtain a normalization sequence
  \[
  {\mathcal S} = {\mathcal S}^0 \vdash_n {\mathcal S}^1 \vdash_n \cdots.
  \]
  Note that from \eqref{eq:thm:seq-edge-norm_conserv}, \eqref{eq:thm:seq-edge-norm_equvT12},
  and Lemma~\ref{lem:norm-possible}, each spanning graph ${\mathcal S}^i = (G^i, T^i_1, T^i_2)$ satisfies
  \begin{align*}
    & \text{$G^i$ is confluent, and} \\
    & {\barrow{\ast}{G_1\cup T_1\cup T_2}} \subseteq {\barrow{\ast}{G^i\cup T^i_1\cup T^i_2}}.
  \end{align*}
  Furthermore, $m({\mathcal S}^i) > m({\mathcal S}^{i+1})$ holds by Lemma~\ref{lem:reduction} for each $i$.
  Since the measure $m(\mathcal{S}^i)$ strictly decreases, the normalization sequence eventually terminates at some step $k$.
  By setting $G'_1 := G^k$ and $T' := T^k_1 \cup T^k_2$,
  we obtain statements (1)--(3) of the theorem, where condition (2) holds because ${\mathcal S}^k$ is a spanning graph.

  Moreover, since the sequence terminates at $k$, we observe $m({\mathcal S}^k) = 0$.
  Therefore, $G'_1 \cup T'$ is confluent by Lemma~\ref{lem:vdash-norm}, and thus statement (4) holds.
\end{proof}


\section{Term rewriting systems and their $R$-graphs}
\label{sec7}\label{sec:R-reduction-graph}


\subsection{Notation and terminology of TRS}

Let $F$ be a finite set of function symbols, and let $X$ be a countably infinite set of variables with $F\cap X=\emptyset$.
By convention, we use $s,t,u,v,w$ for terms,
$x,y$ for variables,
$a,b,c,f,g$ for function symbols,
$p,q$ for positions, and
$\sigma,\theta$ for substitutions.
$\T(F,X)$ (or simply $\T$) denotes the set of terms constructed from $F$ and $X$, and
$\Var(t)$ denotes the set of variables occurring in a term $t$.
A \emph{ground} term is a term in $\T(F,\emptyset)$.
The set of positions in $t$ is $\Pos(t)$, and
the \emph{root} position is $\varepsilon$.

For $p\in\Pos(t)$, the subterm of $t$ at position $p$ is denoted by $t|_p$.
We say a term $s$ is a subterm of $t$ if $s=t|_p$ for some $p\in\Pos(t)$, which is denoted by $s \unlhd t$.
The root symbol of $t$ is $\Root(t)$.
The set of variable positions in $t$ is denoted by
$\Pos_{X}(t)=\{p \mid \Root(t|_p) \in X \}$, and $\Pos_{F}(t)=\Pos(t) \setminus \Pos_{X}(t)$.
The term obtained from $t$ by replacing its subterm
at position $p$ with $s$ is denoted by $t[s]_p$.
The \emph{size} $|t|$ of a term $t$ is $|\Pos(t)|$, and
the \emph{height} $\mathrm{ht}(t)$ is $\max \{ |p| \mid p \in \Pos(t)\}$,
where $|p|$ denotes the length of $p$.

A \emph{rewrite rule} is a pair $(\ell, r)$ of terms such
that $\ell \notin X$ and $\Var(\ell) \supseteq \Var(r)$.
We write this as $\ell \rightarrow r$.
A \emph{term rewriting system} (TRS) is a finite set $R$ of rewrite rules.
The \emph{rewrite relation} of $R$ on $\T(F,X)$ is denoted by $\Rarrow{}{R}{}{}$.
We sometimes write $s\Rarrow{p}{R}{}{}t$ (resp., $s\Rarrow{p}{\ell \rightarrow r}{}{}t$)
to indicate the \emph{rewrite step} at position $p$
(resp., with the rule $\ell \rightarrow r \in R$).
A rewrite step $s\Rarrow{p}{R}{}{}t$ is a \emph{top reduction}
if $p=\varepsilon$.
Otherwise, it is an \emph{inner reduction},
written as $s\Rarrow{\inn}{R}{}{}t$.
We write $\NF_R\; (\subseteq \T(F,X))$ for the set of normal forms, i.e., $\NF_R = \NF_{\<\T(F,X),\Rarrow{}{R}{}{} \>}$.

Given a TRS $R$, the set $D$ of \emph{defined symbols} is $\{\Root(\ell) \mid \ell\to r\in R\}$.
The set $C$ of \emph{constructor symbols} is $F \setminus D$.
For $T \subseteq \T(F,X)$ and $f \in F$,
we use $T\restRoot{f}$ to denote $\{ s\in T \mid \Root(s)=f \}$.
For a subset $F'$ of $F$,
we use $T\restRoot{F'}$ to denote the union $\cup_{f\in F'} T\restRoot{f}$.

A \emph{constructor term} is a term in $\T(C,X)$, and
a \emph{semi-constructor term} is a term in which defined function symbols
appear only in ground subterms.
A term is \emph{shallow} if the length $|p|$ is 0 or 1
for every variable position $p$ in the term.
A \emph{weakly shallow term} is a term in which defined function symbols
appear only at the root position or in ground subterms
(i.e., $p\not=\varepsilon$ and $\Root(s|_p)\in D$ imply that $s|_p$ is ground).
Note that every shallow term is weakly shallow.

A rewrite rule $\ell \rightarrow r$ is
\emph{shallow} (resp., \emph{weakly shallow}) if both $\ell$ and $r$ are
shallow (resp., weakly shallow).
It is \emph{collapsing} if $r$ is a variable, and
\emph{non-collapsing} otherwise.
We denote the set of collapsing (resp., non-collapsing) rewrite rules in $R$ by
$R^c$ (resp., $R^{nc}$).
A TRS is \emph{shallow} (resp., \emph{weakly shallow}) if each rewrite rule is
shallow (resp., weakly shallow).
A TRS is \emph{non-collapsing} if it contains no collapsing rules.
A rewrite step $s \Rarrow{}{R}{}{} t$ is collapsing (resp., non-collapsing),
written as $s \Rarrow{}{R}{c}{} t$ (resp., $s \Rarrow{}{R}{nc}{} t$),
if the rule applied is collapsing (resp., non-collapsing).

\begin{example}\label{ex:E-overlapping}
  The following TRS $R_{\settrs{trs:shallow-noncollapsing}}$~\cite{Huet80} is weakly shallow and non-collapsing.
  \[
  R_{\ref{trs:shallow-noncollapsing}} = \{
  f(x,x) \rightarrow a,\
  f(x,g(x)) \rightarrow b,\
  c \rightarrow g(c)
  \}
  \]
\end{example}


Let $\ell_1\rightarrow r_1$ and $\ell_2\rightarrow r_2$ be variants of rewrite
rules in a TRS $R$ such that $\ell_1$ and $\ell_2$ have no common variables.
Let $p$ be a position in $\ell_2$ such that $\ell_2|_p$ is not a variable.
If there exist substitutions $\theta_1,\theta_2$ such that
$\ell_1\theta_1 = \ell_2|_p\theta_2$
(resp., $\ell_1\theta_1 \Barrow{\inn}{R}{\ast}{} \ell_2|_p\theta_2$),
we call $(\ell_2\theta_2[r_1\theta_1]_p, r_2\theta_2)$ an \emph{overlapping pair}
(resp., an \emph{$E$-overlapping pair}) with a \emph{superposition}
$\ell_2\theta_2$
(resp., superpositions $(\ell_2\theta_2[\ell_1\theta_1]_p, \ell_2\theta_2)$).
A TRS $R$ is \emph{non-overlapping} (resp., \emph{non-$E$-overlapping})
if each overlapping (resp., $E$-overlapping) pair is trivial,
i.e., $p=\varepsilon$ and the two rules are identical (up to renaming of variables).
For a non-trivial overlapping pair (resp., non-$E$-overlapping pair) $(s,t)$
with a superposition $u$ (resp., superpositions $(u_1,u_2)$), every rewrite
sequence $s\sigma \Larrow{}{R}{}{} u\sigma \Rarrow{}{R}{}{} t\sigma$
(resp., $s\sigma\Larrow{}{R}{}{} u_1\sigma\Barrow{p<}{R}{\ast}{} u_2\sigma
\Rarrow{}{R}{}{} t\sigma$)
with a substitution $\sigma$ and a position $p$ is a \emph{peak}
(resp., \emph{$E$-peak}).

It is well-known that left-linear and non-overlapping TRSs are
confluent~\cite{Huet80}.
Thus, this paper assumes that TRSs contain non-constant function symbols;
otherwise, the non-overlapping property automatically guarantees confluence.

  For non-left-linear TRSs, the non-overlapping condition
  is insufficient to guarantee confluence witnessed by
  $R_{\ref{trs:nonCR2}}$ in Example~\ref{ex:nonCR2}
  and $R_{\ref{trs:nonCR}}$ in Example~\ref{ex:nonCR}~\cite{Huet80}.
  The former has a divergent sequence
  $a \larrow{}{R_{\ref{trs:nonCR2}}} d(a,a) \rarrow{}{R_{\ref{trs:nonCR2}}} d(a,f(a))
  \rarrow{}{R_{\ref{trs:nonCR2}}} c$,
  and the latter (which is non-$E$-overlapping) has a sequence
  $a \rarrow{\ast}{R_{\ref{trs:nonCR}}} b$ given by
  $a \rarrow{}{R_{\ref{trs:nonCR}}} f(a) \rarrow{}{R_{\ref{trs:nonCR}}} d(a,f(a))
  \rarrow{}{R_{\ref{trs:nonCR}}} d(f(a),f(a)) \rarrow{}{R_{\ref{trs:nonCR}}} b$,
  and a divergent sequence $b \larrow{\ast}{R_{\ref{trs:nonCR}}} a
  \rarrow{\ast}{R_{\ref{trs:nonCR}}} d(a,f(a)) \rarrow{\ast}{R_{\ref{trs:nonCR}}} d(b,f(b))$.
  Instead, a non-$E$-overlapping TRS has the unique normal form property
  (UN)~\cite{Ogawa89,MO01,MOOO97}.

In the rest of the paper, we assume the following restrictions on TRSs,
which define the class of \emph{root-distinct TRSs}.
\begin{enumerate}
  \item The rule $h(x) \to x$ is the only possible collapsing rule.
  \item $\Root(\ell)\neq \Root(r)$ holds for every non-collapsing rule
    $\ell\to r \in R$.
  \item The TRS contains both non-constant function symbols and constants.
\end{enumerate}
Restriction (1) is reasonable for the following reason.
Any TRS can be converted into a TRS with at most one collapsing rule
$h(x) \to x$ via a confluence-preserving transformation.
This transformation introduces a fresh function symbol $h$ and the rule
$h(x) \to x$,
and then replaces each collapsing rule $\ell \to x$ with a non-collapsing rule
$\ell \to h(x)$.
For example, the rule $f(g(x),y,x) \to x$ is converted into
$f(g(x),y,x) \to h(x)$.
It is obvious that this transformation preserves confluence.

Restriction (2) eases proofs related to the subterm property.
This restriction is also reasonable because any TRS can be transformed
into this class via a confluence-preserving transformation.
Specifically, each rule $\ell \to r$ satisfying $\Root(\ell) = \Root(r)$ is
replaced with two rules,
$\ell \to g(x_1,\ldots,x_n)$ and $g(x_1,\ldots,x_n) \to r$,
where $g$ is a fresh function symbol introduced for each such rule, and
$x_1,\ldots,x_n$ are the variables occurring in $\ell$.

In this section, we use the following TRS as a running example, where $g^3(t)$
denotes $g(g(g(t)))$.
\begin{example}\label{ex:overview}
  \(R_{\settrs{trs:running}}=\{f(x,g(x))\to g^3(x), c\to g(c) \}.\)
\end{example}

This section introduces $R$-graphs (in Section~\ref{subsec:r-graph}),
direct subterm graphs and monotonic extensions (in Section~\ref{subsec:monotonic-extension}),
and subterm-closed graphs and their properties (in Section~\ref{subsec:subterm-closed}).

\subsection{Top and inner edges of $R$-graph}\label{subsec:r-graph}

\begin{definition}
  Let $R$ be a TRS.
  An ARS $G=\< V,\mathord{\rightarrow}\>$ is an \emph{$R$-graph}
  if $V$ is a subset of $\T(F,X)$ and
  $\mathord{\rightarrow}\subseteq\mathord{\Rarrow{}{R}{}{}}$.
  The $R$-graph induced by $V \subseteq \T(F,X)$ is denoted by
  $\Ind_R(V) := \< V, {\rarrow{}{R}}|_{V \times V}\>$.
\end{definition}

\begin{example}\label{ex:G0}
  Consider the TRS $R_{\ref{trs:running}}$ in Example~\ref{ex:overview}.
  The $R_{\ref{trs:running}}$-graph induced by $V = \{c, g(c), g^2(c) \}$ is
  $G = \< \{ c, g(c), g^2(c) \}, \{ (c, g(c)), (g(c), g^2(c)) \}\>$.
  In the sequel, we use a simple representation of graphs such as
  $G = \{ c \rightarrow g(c) \rightarrow g^2(c)\}$.
\end{example}

A finite $R$-graph corresponds to an $R$-graph as defined in~\cite{Sakai10}.
For an $R$-graph $G=\< V,\mathord{\rightarrow}\>$, \emph{inner edges},
\emph{top edges}, and \emph{non-top edges} are defined by
${\Rarrow{\inn}{}{}{}}  = {\rarrow{}{}} \cap \mathord{\Rarrow{\inn}{R}{}{}}$,
${\Rarrow{\tpp}{}{}{}} = \mathord{\rightarrow} \cap {\Rarrow{\tpp}{R}{}{}}$, and
${\Rarrow{\ntpp}{}{}{}} = \mathord{\Rarrow{\inn}{}{}{}} \setminus
\mathord{\Rarrow{\tpp}{}{}{}}$, respectively.
We use $G^{\inn}$, $G^{\tpp}$, and $G^{\ntpp}$ to denote the graphs
$\< V, \mathord{\Rarrow{\inn}{}{}{}} \>$,
$\< V, \mathord{\Rarrow{\tpp}{}{}{}} \>$, and
$\< V, \mathord{\Rarrow{\ntpp}{}{}{}} \>$, respectively.

We note that $\mathord{\Rarrow{\inn}{}{}{}} \cap \mathord{\Rarrow{\tpp}{}{}{}}
\neq \emptyset$ can hold in general.
For example, for $R = \{ h(x)\rightarrow x \}$ and $G: h(h(a)) \to h(a)$,
we have $\mathord{\Rarrow{\inn}{}{}{G}} = \mathord{\Rarrow{\tpp}{}{}{G}}
= \mathord{\rarrow{}{G}}$.
From assumptions (1) and (2), any edge in $G^\tpp \cap G^\inn$ must be
a collapsing edge in the form of $h(h(v)) \rightarrow h(v)$ for a term $v$.

For an $R$-graph $G = \< V,\mathord{\rightarrow}\>$, \emph{collapsing} edges
and \emph{non-collapsing} edges are given by
${\rarrow{}{}} \cap {\Rarrow{}{R}{c}{}}$ and
${\rarrow{}{}} \cap {\Rarrow{}{R}{nc}{}}$, respectively.
We use $G^{c}$ (resp., $G^{nc}$) to denote the subgraph of $G$ consisting of
collapsing edges (resp., non-collapsing edges).

\begin{lemma}\label{lem:inner-edges}
  For an $R$-graph $G$,
  \begin{enumerate}
  \item\label{lem:inner-edges-a}
    $G = G^{\tpp} \cup G^{\ntpp} = G^{\tpp} \cup G^{\inn}$,
  \item\label{lem:inner-edges-d}
    $(G \cup G')^\ntpp = G^\ntpp \cup G'^\ntpp$,
  \item\label{lem:inner-edges-e}
    $(G^\tpp)^\ntpp = \<V_G, \emptyset\>$,
  \item\label{lem:inner-edges-b}
    ${\rarrow{}{G^{\inn}}} \cap {\rarrow{}{G^{\tpp}}} =
    \{ (s,t) \mid s=h(h(u)) \to_G h(u)=t \text{ for some $u$} \}$,
  \item\label{lem:inner-edges-c}
    $G^\ntpp \subseteq G^\inn$. Moreover, $G^\inn = G^\ntpp$ for a non-collapsing
    system.
  \end{enumerate}
\end{lemma}

\begin{proof}\mbox{}
  \begin{itemize}
  \item[(\ref{lem:inner-edges-a})]
    Trivial.
  \item[(\ref{lem:inner-edges-d})]
    Let $e$ be an edge of $(G \cup G')^\ntpp$.
    If $e$ is an edge of $G$, then it is not a top edge.
    Thus, it is an edge of $G^\ntpp$.
    The other case is similar.

    Conversely, let $e$ be an edge of $G^\ntpp$.
    Then $e$ is an edge of $G$, but not a top edge.
    Thus, it is an edge of $(G \cup G')^\ntpp$.
    The other case is similar.
  \item[(\ref{lem:inner-edges-e})]
    The vertex condition is clear.
    Assume that an edge $e$ exists in $(G^\tpp)^\ntpp$.
    Then $e$ is an edge of $G^\tpp$ but not a top edge,
    which is a contradiction.
  \item[(\ref{lem:inner-edges-b})]
    This holds from the root-distinct property as follows.
    Since every non-collapsing rule $\ell\to r$ satisfies
    $\Root(\ell)\neq\Root(r)$,
    any edges in ${\rarrow{}{G^{\inn}}} \cap {\rarrow{}{G^{\tpp}}}$ are obtained
    by applying the unique collapsing rule $h(x) \to x$.
    Thus, such edges are in the form of $h(h(u)) \to_G h(u)$.
  \item[(\ref{lem:inner-edges-c})]
    The former claim is obvious.
    The latter claim is due to the restriction that $\Root(\ell) \neq \Root(r)$
    for every non-collapsing rule $\ell \to r$.
  \end{itemize}
\end{proof}

\subsection{Direct-subterm graph of $G$}\label{subsec:monotonic-extension}

We define the set of direct subterms $\Sub(t)$ of a term $t$ by
$\Sub(x)=\emptyset$
for a variable $x$, and $\Sub(t)=\{t_1,\ldots,t_n\}$ if $t=f(t_1,\ldots,t_n)$.
For an $R$-graph $G = \< V, \mathord{\rightarrow} \>$ and a set of function
symbols $F' \subseteq F$, we denote the source-root-restricted subgraph
by $G\restRoot{F'} = \< V, {{\rightarrow}\restRoot{F'}} \>$, where
${{\rightarrow}\restRoot{F'}} = \mathord{\rightarrow}|_{V\restRoot{F'}\times V}$.

Recall that $T\restRoot{F'}$ represents the union
$\bigcup_{f\in F'} T\restRoot{f}$,
where $T\restRoot{f}$ is the set of $f$-rooted terms in $T$.
Thus, we have ${{\rightarrow}\restRoot{C}}
= \mathord{\rightarrow}|_{V\restRoot{C}\times V\restRoot{C}}$
and $\mathord{\rightarrow} = {{\rightarrow}\restRoot{D}} \cup
{{\rightarrow}\restRoot{C}}$.

\begin{definition}\label{def:subterm}
  Let $G=\< V, \rightarrow\>$ be an $R$-graph.
  The \emph{direct-subterm $R$-graph} $\Sub(G)$ of $G$ is defined as
  $\< \Sub(V), \Sub(\mathord{\rightarrow}) \>$, where
  \[
  \left\{
  \begin{array}{l}
    \Sub(V) = \bigcup_{t \in V} \Sub(t) \\
    \Sub(\mathord{\rightarrow})
    =\{(s_i,t_i) \mid f(s_1,\ldots,s_n) \Rarrow{\inn}{}{}{} f(t_1,\ldots,t_n),\,
    s_i\not=t_i \text{ for } 1 \leq i \leq n\}. \\
  \end{array}
  \right.
  \]
\end{definition}

\begin{definition} \label{def:mono-ext}
  For a subset $F' \subseteq F$ and an $R$-graph $G = \< V, \rightarrow \>$,
  the \emph{$F'$-monotonic extension} $M_{F'}(G)=\< V',\rightarrow'\>$ of $G$ is
  defined by
  \[
  \left\{
  \begin{array}{lcl}
    V' & = & \{f(s_1,\ldots,s_n) \mid f \in F',\; s_1, \ldots, s_n \in V \},\\
    \mathord{\rightarrow'} & = & \{(f(\cdots s \cdots), f(\cdots t \cdots))
    \in V' \times V' \mid s\rightarrow t \}.
  \end{array}
  \right.
  \]
\end{definition}

\begin{example}\label{ex:sub-ext}
  For the $R_{\ref{trs:running}}$-graph
  $G_0 = \{ f(g(c),c) \gets f(c,c) \to f(c,g(c)) \Rarrow{\tpp}{}{}{} g^3(c) \}$,
  its direct-subterm $R_{\ref{trs:running}}$-graph is given by
  $\Sub(G_0) = \{ c, g^2(c), c \to g(c) \}$,
  where $c$ and $g^2(c)$ denote isolated vertices.
  Then, $M_F(\Sub(G_0))$ is the $R_{\ref{trs:running}}$-graph shown in
  Fig.~\ref{fig:ex-extension}.

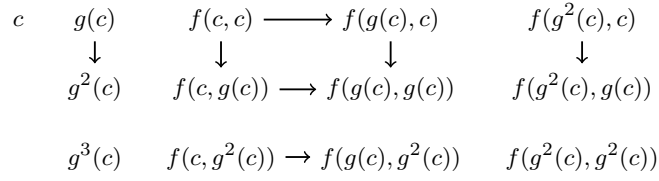
\begin{figure}
  \begin{center}
    \mbox{}\\[-1em]
    \(
    \begin{tikzpicture} [line width=0.07em]
      \matrix (m) [matrix of math nodes,row sep=0.9em,column sep=1em,minimum width=0.5em] {
        c & g(c)   & f(c,c)      & f(g(c),c)      & f(g^2(c),c)   \\
        & g^2(c) & f(c,g(c))   & f(g(c),g(c))   & f(g^2(c),g(c))  \\
        & g^3(c) & f(c,g^2(c)) & f(g(c),g^2(c)) & f(g^2(c),g^2(c)) \\
      };
      \path[->]
      (m-1-3) edge (m-1-4) 
      (m-2-3) edge (m-2-4) 
      (m-3-3) edge (m-3-4) 
      (m-1-2) edge (m-2-2)  (m-1-3) edge (m-2-3) (m-1-4) edge (m-2-4) (m-1-5) edge (m-2-5)
      ;
    \end{tikzpicture}
    \)
    \\[-2em]\mbox{}
  \end{center}
  \caption{ $R_{\ref{trs:running}}$-graph $M_F(\Sub(G_0))$  }
  \label{fig:ex-extension}
\end{figure}

\end{example}

The following proposition is immediate.

\begin{proposition}\label{prop:monotonic-extension}
  For $R$-graphs $G$ and $G'$, and a set of function symbols $F' \subseteq F$:
  \begin{enumerate}
  \item\label{prop:monotonic-extension-a}
    If $G$ is confluent (resp., finite), then $M_{F'}(G)$ is also confluent
    (resp., finite).
  \item\label{prop:monotonic-extension-b}
    $M_{F'}$ is monotonic; i.e., $G \subseteq G'$ implies
    $M_{F'}(G) \subseteq M_{F'}(G')$.
  \item\label{prop:monotonic-extension-c}
    $M_{F'}(G) \subseteq M_F(G)$.
  \item\label{prop:monotonic-extension-d}
    $M_{F'}(G) \cup M_{F'}(G') \subseteq M_{F'}(G \cup G')$.
  \end{enumerate}
\end{proposition}
Note that the reverse inclusion of (\ref{prop:monotonic-extension-d}) does not
always hold if $V_G \neq V_{G'}$.
For example, let $G = \{ a \to b \}$, $G' = \{a'\}$, and $g \in F$ be a binary
function symbol.
Then, the edge $g(a,a') \to g(b,a')$ is in $M_{F}(G \cup G')$, but it is in
neither $M_F(G)$ nor $M_F(G')$.

\begin{lemma}\label{lem:mono-ext}\label{lem:monotonic-extension}
  For an $R$-graph $G$,
  \begin{enumerate}
  \item\label{lem:mono-ext-a}
    $\Sub(G^\ntpp) \subseteq \Sub(G^\inn)=\Sub(G)$.

  \item\label{lem:mono-ext-b}
    The vertices of $M_F(\Sub(G))$ include those of $G$.
    Moreover, $G^{\inn} \subseteq M_F(\Sub(G))$ and
    $G^{\ntpp} \subseteq M_F(\Sub(G^\ntpp))$.

  \item\label{lem:mono-ext-c}
    If $\Sub(G') \subseteq G$, then
    $G'^\ntpp \subseteq G'^\inn \subseteq M_F(G)$, which implies that
    $V_{G'} \subseteq V_{M_F(G)}$.
  \end{enumerate}
\end{lemma}
\begin{proof}\mbox{}
  \begin{enumerate}
  \item[(\ref{lem:mono-ext-a})]
    $G^\ntpp \subseteq G^\inn \subseteq G$, where the former follows
    from Lemma~\ref{lem:inner-edges}(\ref{lem:inner-edges-c}).
    Thus, $\Sub(G^\ntpp) \subseteq \Sub(G^\inn) \subseteq \Sub(G)$.
    Moreover, $\Sub(G^\inn) \supseteq \Sub(G)$ follows from the definition;
    assuming $(s,t) \in \Sub(G)$, there exists an edge
    $f(\cdots s \cdots) \to_{G^\inn} f(\cdots t\cdots)$,
    which implies that $(s,t) \in \Sub(G^\inn)$.

  \item[(\ref{lem:mono-ext-b})]
    \emph{Proof of $V_G \subseteq V_{M_F(\Sub(G))}$:}
    Let $t \in V_G$, which is of the form $f(t_1,\ldots,t_n)$ for some $n\geq 0$.
    Then, $t_i \in V_{\Sub(G)}$ for each $i$ ($1\leq i\leq n$).
    Thus, $t\in V_{M_F(\Sub(G))}$.

    \smallskip\noindent
    \emph{Proof of $G^{\inn} \subseteq M_F(\Sub(G))$:}
    Let $(s,t) \in G^\inn$. We can write them as $s=f(s_1,\ldots,s_n)$
    and $t=f(t_1,\ldots,t_n)$ for some $n \geq 1$.
    Since $(s, t) \in \mathord{\to_{G^\inn}}$, there exists an index $i$
    ($1\leq i\leq n$)
    such that $s_i \to_{\Sub(G)} t_i$ and $s_j=t_j$ for $j\neq i$.
    Hence, $(s,t) \in \mathord{\to_{M_F(\Sub(G))}}$.

    \smallskip\noindent
    \emph{Proof of $G^{\ntpp} \subseteq M_F(\Sub(G^\ntpp))$:}
    Let $(s,t) \in G^\ntpp$. We can write them as $s=f(s_1,\ldots,s_n)$
    and $t=f(t_1,\ldots,t_n)$ for some $n \geq 1$.
    Since $(s, t) \in \mathord{\to_{G^\ntpp}}$, there exists an index $i$
    ($1\leq i\leq n$)
    such that $s_i \to_{\Sub(G^\ntpp)} t_i$ and $s_j=t_j$ for $j\neq i$.
    Hence, $(s,t) \in \mathord{\to_{M_F(\Sub(G^\ntpp))}}$.

    \item[(\ref{lem:mono-ext-c})]
    We have $M_F(\Sub(G')) \subseteq M_F(G)$ by the assumption and
    Proposition~\ref{prop:monotonic-extension}
    (\ref{prop:monotonic-extension-b}).
    Combining this with (\ref{lem:mono-ext-b}) of this lemma, we obtain
    $G'^\inn \subseteq M_F(G)$,
    which implies $V_{G'} = V_{{G'}^\inn} \subseteq V_{M_F(G)}$.
  \end{enumerate}
\end{proof}

\begin{proposition}\label{prop:subterm}
  Let $F'$ be a subset of $F$.
  \begin{enumerate}
  \item\label{prop:subterm-0}
    $\Sub(G \cup G') \subseteq \Sub(G) \cup \Sub(G')$.
  \item\label{prop:subterm-1}
    If $G \subseteq G'$, then $\Sub(G) \subseteq \Sub(G')$.
  \item\label{prop:subterm-2}
    $\Sub(M_{F'}(G)) \subseteq G$.
    Moreover, $\Sub(M_{F'}(G)) = G$ if $F'$ contains a non-constant function
    symbol.
  \item\label{prop:subterm-3}
    $G^{\inn}\restRoot{F'} \subseteq M_{F'}(\Sub(G)) \cup G|_\emptyset$.
  \item\label{prop:subterm-4}
    $M_F(G)\restRoot{F'} = M_{F'}(G) \cup M_F(G)|_\emptyset$.
  \item\label{prop:subterm-5}
    $M_C(G)^\tpp = M_C(G)|_\emptyset$ and $M_C(G)^\inn = M_C(G)^\ntpp = M_C(G)$.
  \end{enumerate}
\end{proposition}

\begin{proof}\mbox{}
  \begin{itemize}
  \item[\eqref{prop:subterm-0}] Trivial.
  \item[\eqref{prop:subterm-1}] Trivial.
  \item[\eqref{prop:subterm-2}] This trivially follows from the definitions.
  \item[\eqref{prop:subterm-3}]
    The vertex inclusion is clear since $V_{G^\inn\restRoot{F'}} = V_G$.
    Consider the edge inclusion.
    Let $s \rarrow{}{G^\inn\restRoot{F'}} t$.
    Then, the root symbols of $s$ and $t$ are identical and belong to $F'$, and
    we can write $s=f(\cdots s' \cdots) \rarrow{}{G} f(\cdots t' \cdots)=t$
    for some $f \in F'$.
    Here, $(s', t') \in \Sub(\mathord{\rarrow{}{G}})$, and hence
    $s \rarrow{}{M_{F'}(\Sub(G))} t$.
  \item[\eqref{prop:subterm-4}] Trivial.
  \item[\eqref{prop:subterm-5}] This is immediate from the fact that $M_C(G)$
    contains no top edges.
  \end{itemize}
\end{proof}

Note that the graph $G|_\emptyset$ is necessary in Proposition~\ref{prop:subterm}
(\ref{prop:subterm-3})
only if (a) $G$ consists of constant vertices, or (b) $F' \neq F$.
We clearly have $V_{G^\inn\restRoot{F'}} = V_G$.
On the other hand, in case (a), we have $V_{\Sub(G)} = \emptyset$, and hence
$V_{M_{F'}(\Sub(G))} = \emptyset$.
In case (b), $V_{M_{F'}(\Sub(G))}$ contains no vertices whose root symbol is in
$F \setminus F'$.

\subsection{Compatibility and edge-commutativity with collapsing edges}
\label{subsec:subterm-closed}

This section shows how to guarantee the compatibility and the edge-commutativity
of collapsing edges.
Here, we assume that the unique collapsing rule in $R$ is $h(x) \rarrow{}{} x$,
where the function symbol $h$ does not occur in the left-hand side of any other
rule.
Lemma~\ref{lem:NestedCollapsing} observes that an edge that is both a top edge
and an inner edge must be in the form of $h(h(u)) \rarrow{}{} h(u)$,
and that the terms $h(u)$ and $u$ are connected via inner edges and top
non-collapsing edges.

\begin{definition} \label{def:subterm-closed}
  Let $T$ be an $R$-graph consisting of top edges only.
  \begin{enumerate}
  \item\label{def:subterm-closed-a}
    An $R$-graph $G$ is \emph{$T$-subterm-closed} if $V_G \supseteq V_T$ and
    $\Sub(G^\ntpp) \subseteq G \cup T$.
  \item\label{def:subterm-closed-b}
    $G$ is \emph{subterm-closed} if it is so with the empty graph
    $\< \emptyset, \emptyset \>$, i.e., $\Sub(G^\ntpp) \subseteq G$.
  \end{enumerate}
\end{definition}
\begin{example}
  For $T = \{ c \to g(c) \}$, the $R_{\ref{trs:running}}$-graph $G_0$ in
  Example~\ref{ex:sub-ext} is $T$-subterm-closed.
\end{example}

Note that if $G$ is $T$-subterm-closed, then $V_{\Sub(G)} \subseteq V_G$ holds,
which follows from $V_G \supseteq V_T$ and $\Sub(G^\ntpp) \subseteq G \cup T$.

We remark that if we adopt $\Sub(G^\inn) \subseteq G \cup T$ as the definition
of the subterm-closed property,
then the subterm property in Theorem~\ref{thm:NOV-Shallow} no longer holds,
as shown in Example~\ref{ex:counter}.

\begin{proposition} \label{prop:T-subterm-closed}\label{prop:T-subterm-closed_a}
  Let $G$ be $T$-subterm-closed. Then, $G$ is $T'$-subterm-closed for any
  $T' \supseteq T$ with $V_{T'} \subseteq V_G$.
\end{proposition}
\begin{proof}
  This is obvious from the definition. 
\end{proof}

\begin{lemma}\label{lem:subterm-closed}
  Let $G$ be a subterm-closed $R$-graph.
  Then, $M_F(G)$ is $G^\tpp$-subterm-closed.
\end{lemma}

\begin{proof}
  The vertex condition $V_{G^\tpp} \subseteq V_{M_F(G)}$ hold
  by Lemma~\ref{lem:monotonic-extension}(\ref{lem:mono-ext-c}) since $G$ is
  subterm-closed. Moreover, we have
  \[\begin{array}{ll}
  \multicolumn{2}{l}{\Sub(M_F(G)^\ntpp)} \\
  \subseteq \Sub(M_F(G))
  & \text{from Lemma~\ref{lem:mono-ext}}(\ref{lem:mono-ext-a}) \\
  \subseteq G
  & \text{by Proposition~\ref{prop:subterm}(\ref{prop:subterm-2})} \\
  = G^\tpp \cup G^\ntpp
  & \text{by Lemma~\ref{lem:inner-edges}(\ref{lem:inner-edges-a})} \\
  \subseteq G^\tpp \cup M_F(\Sub(G^\ntpp))
  & \text{by Lemma~\ref{lem:mono-ext}(\ref{lem:mono-ext-b})} \\
  \subseteq G^\tpp \cup M_F(G)
  & \text{\begin{minipage}[t]{0.6\textwidth}
  by applying
  Proposition~\ref{prop:monotonic-extension}(\ref{prop:monotonic-extension-b})
  to $\Sub(G^\ntpp)\subseteq G$, which holds since $G$ is subterm-closed.
  \end{minipage}}
  \end{array}\]
  Thus, $M_F(G)$ is $G^\tpp$-subterm-closed.
\end{proof}

The following lemma plays a key role in handling collapsing edges.
It reduces the number of collapsing edges to obtain compatibility
while preserving overall connectivity and edge-commutativity.
%

\begin{lemma} \label{lem:NestedCollapsing}
  Let $T$ be an $R$-graph consisting only of top edges (i.e., $T = T^\tpp$),
  and let $\Vrt$ be a confluent and $T$-subterm-closed $R$-graph consisting
  only of inner edges (i.e., $\Vrt = \Vrt^\inn$).
  Then, there exists an $R$-graph $T'^c \subseteq T^c$ such that:
  \begin{enumerate}
    \renewcommand{\labelenumi}{(\alph{enumi})}%
    \renewcommand{\theenumi}{\alph{enumi}}%
  \item\label{lem:NestedCollapsing_a}
    ${\barrow{\ast}{\Vrt\cup T}} = {\barrow{\ast}{\Vrt\cup T'^c\cup T^{nc}}}$,
  \item\label{lem:NestedCollapsing_b} $T'^c$ is $T^{nc}$-compatible over $\Vrt$, and
  \item\label{lem:NestedCollapsing_c} $T'^c$ is $T^{nc}$-edge-commutative over
    $(\Vrt,\Ind_R(V_{\Vrt})^\tpp)$.
  \end{enumerate}
\end{lemma}

\begin{proof}
  Let ${\rarrow{}{T''}} := \{(h(s),s) \in {\rarrow{}{T^c}} \mid
  \exists t \text{ such that } h(s) \rarrow{\ast}{\Vrt} h(h(t))
  \rarrow{}{\Vrt} h(t) \}$ and $T'^c = T^c \setminus T''$.
  We show the following three claims {\bf (*)}, {\bf (*2)}, and {\bf (*3)}:

  \begin{description}
  \item[(*):] If $h(s) \rarrow{\ast}{\Vrt} h(w)$ does not contain an edge of
    the form $h(h(u)) \rarrow{}{\Vrt} h(u)$,
    then $s \rarrow{\ast}{\Vrt \cup T^{nc}} w$.
  \end{description}
  We first show {\bf (*)}.
  Suppose $h(s) = h(s_0) \rarrow{}{\Vrt} h(s_1) \rarrow{}{\Vrt} \cdots
  \rarrow{}{\Vrt} h(s_n) = h(w)$ is a sequence containing no edges of the form
  $h(h(u)) \rarrow{}{\Vrt} h(u)$.
  Since each edge in $\Vrt^\inn \cap \Vrt^\tpp$ is of the form
  $h(h(u)) \rarrow{}{\Vrt} h(u)$ by
  Lemma~\ref{lem:inner-edges}(\ref{lem:inner-edges-b}),
  each step $h(s_{i-1}) \rarrow{}{\Vrt} h(s_i)$
  ($0 < i \leq n$) must actually belong to $\rarrow{}{\Vrt^\ntpp}$.
  From the $T$-subterm-closedness of $\Vrt$,
  we have $s_{i-1} \rarrow{}{\Vrt\cup T} s_i$,
  which means it is not a collapsing step.
  (Otherwise, the initial step $h(s_{i-1}) \rarrow{}{\Vrt} h(s_i)$
  would be of the form $h(h(u)) \rarrow{}{\Vrt} h(u)$).
  Thus, $s \rarrow{\ast}{\Vrt \cup T^{nc}} w$ holds.

  \begin{description}
    \item[(*2):]
    \emph{If $h(s) \rarrow{}{T''} s$, then $h(s) \downarrow_{\Vrt \cup T^{nc}} s$
      (specifically, they join at $h(t)$).}
  \end{description}
  Since $h(s) \rarrow{}{T''} s$, there exists some $t$ such that
  $h(s) \rarrow{\ast}{\Vrt} h(h(t)) \rarrow{}{\Vrt} h(t)$.
  Without loss of generality, we can assume that the sequence
  $h(s) \rarrow{\ast}{\Vrt} h(h(t))$ contains no edges of the form
  $h(h(w)) \rarrow{}{\Vrt} h(w)$ (otherwise, we could choose $w$ for $t$
  to obtain a shorter sequence).
    This means that $h(s) \rarrow{\ast}{\Vrt^\ntpp} h(h(t))$.
    Hence, by {\bf (*)}, we obtain $s \rarrow{\ast}{\Vrt \cup T^{nc}} h(t)$.
  Combined with $h(s) \rarrow{\ast}{\Vrt} h(t)$,
  we see that $h(t)$ is a common reduct of $s$ and $h(s)$,
  meaning $h(s) \downarrow_{\Vrt \cup T^{nc}} s$.

  \begin{description}
  \item[(*3):] \emph{If $h(s) \rarrow{}{T'^c} s$, then
  $h(s) \rarrow{\ast}{\Vrt} h(t)$ implies $s \rarrow{\ast}{\Vrt\cup T^{nc}} t$.
  }
  \end{description}
  By the definition of $\rarrow{}{T'^c}$, an edge of the form
  $h(h(u)) \rarrow{}{\Vrt} h(u)$ is not reachable from $h(s)$ via
  $\rarrow{\ast}{\Vrt}$.
  Thus, any sequence $h(s) \rarrow{\ast}{\Vrt} h(t)$ must actually be
  $h(s) \rarrow{\ast}{\Vrt^\ntpp} h(t)$.
  Hence, applying {\bf (*)}, we obtain $s \rarrow{\ast}{\Vrt \cup T^{nc}} t$.

  \medskip
  Using claims {\bf (*)}, {\bf (*2)}, and {\bf (*3)},
  we prove the main statements:
  \begin{enumerate}
  \item[(a)] The ($\supseteq$)-direction is immediate from $T = T^c \cup T^{nc}$
    and $T^c \supseteq T'^c$.
    The ($\subseteq$)-direction follows from {\bf (*2)};
    i.e., any edge $h(s) \rarrow{}{T^c} s \notin T'^c$ satisfies
    $h(s) \downarrow_{\Vrt \cup T^{nc}} s$.
  \item[(b)] For $s \larrow{}{T'^c} u \barrow{\ast}{\Vrt} u'\rarrow{}{T'^c} s'$,
    we must show that $s \barrow{\ast}{\Vrt\cup T^{nc}} s'$.
    Since $h(x) \to x$ is the unique collapsing rule in $R$ and
    $T'^c \subseteq T^c$, we must have $u = h(s)$ and $u' = h(s')$.
    From confluence of $\Vrt$, there exists some $t \in V_\Vrt$ such that
    $h(s) \rarrow{\ast}{\Vrt} h(t) \larrow{\ast}{\Vrt} h(s')$.
    Hence, {\bf (*3)} implies
    $s \rarrow{\ast}{\Vrt \cup T^{nc}} t \larrow{\ast}{\Vrt \cup T^{nc}} s'$,
    satisfying the condition.
  \item[(c)] Since $\Vrt$ is $T$-subterm-closed, $V_T \subseteq V_{\Vrt}$
    and $V_{\Sub(\Vrt^\ntpp)} \subseteq V_{\Vrt \cup T} = V_\Vrt$.
    On the other hand, since $T'^c \subseteq T$ consists of top edges,
    we have $T'^c \subseteq \Ind_R(V_{\Vrt})^\tpp$.

    We now show that each $(h(u),u) \in {\rarrow{}{T'^c}}$ satisfies
    the conditions of Definition~\ref{def:edge-commutative}.
    Let $w \in V_\Vrt$ such that $h(u) \rarrow{\ast}{\Vrt} w$.
    Since $\Vrt$ consists of inner edges, $w = h(v)$ for some $v$.
    From $v \in V_{\Sub(\Vrt)} = V_{\Sub(\Vrt^\ntpp)} \subseteq V_\Vrt$, we have
    $(h(v),v) \in \Ind_R(V_{\Vrt})^\tpp$.
    Finally, $u \larrow{}{T'^c} h(u) \rarrow{\ast}{\Vrt} h(v)$ implies
    $u \rarrow{\ast}{\Vrt\cup T^{nc}} v$ by {\bf (*3)}.
  \end{enumerate}
\end{proof}

\begin{lemma}\label{lem:mixed-compat}
  Let $T$ be an $R$-graph consisting only of top edges, and let $\Vrt$ be
    an $R$-graph consisting only of inner edges, with $V_T \subseteq V_\Vrt$.
    Let $T'^c \subseteq T^c$.
    If $T^{nc}$ is compatible over $\Vrt$ and $T'^c$ is $T^{nc}$-compatible
    over $\Vrt$, then $T^{nc} \cup T'^c$ is $T^{nc}$-compatible over $\Vrt$.
\end{lemma}
\begin{proof}
  Since $R$ is root-distinct, $h(x) \to x$ is the unique collapsing rule
    and $h$ occurs at the root of no other left-hand side; hence every
    non-collapsing rule $\ell \to r$ satisfies $\Root(\ell) \neq h$.

  Let $u,u',v,v' \in V_\Vrt$ with
    $u' \larrow{}{T^{nc} \cup T'^c} u \barrow{\ast}{\Vrt} v
    \rarrow{}{T^{nc} \cup T'^c} v'$;
    we show $u' \barrow{\ast}{\Vrt \cup T^{nc}} v'$.
    If both top steps lie in $T^{nc}$, then $u' \barrow{\ast}{\Vrt} v'$ since
    $T^{nc}$ is compatible over $\Vrt$.
    If both lie in $T'^c$, then $u' \barrow{\ast}{\Vrt \cup T^{nc}} v'$ since
    $T'^c$ is $T^{nc}$-compatible over $\Vrt$.

  The remaining mixed cases cannot occur.
    A top step by $T^{nc}$ applies a non-collapsing rule at the root, so its
    source $s$ satisfies $\Root(s) \neq h$, whereas a top step by
    $T'^c \subseteq T^c$ applies $h(x) \to x$ at the root, so its source $s$
    satisfies $\Root(s) = h$.
    Since $\Vrt$ consists of inner edges, every step of $\barrow{\ast}{\Vrt}$
    rewrites a proper subterm and hence preserves the root symbol, so
    $\Root(u) = \Root(v)$.
    A mixed peak would make $\Root(u)$ and $\Root(v)$ equal to $h$ and to a
    symbol other than $h$ in some order, a contradiction.
\end{proof}

\section{Non-$E$-overlapping and shallow TRS}
\label{sec8}\label{sec:shallow}

Theorem~\ref{thm:main} in the next section can be viewed as an essential
generalization of
Theorems~\ref{thm:shallow-noncollapsing} and \ref{thm:NOV-Shallow} in this section,
as their proofs share a virtually identical underlying structure.
Nevertheless, the arguments are worth presenting here to clearly demonstrate
how to apply our technique to show confluence based on
Theorem~\ref{thm:seq-edge-norm}.
Our goal is to construct a finite confluent ARS that contains the diverging
sequence $u \barrow{\ast}{R} v$.
Its bottom-up construction extends the finite confluent ARS $G$ over subterms
by adding top reductions $T$ from $u \barrow{\ast}{R} v$ using
Theorem~\ref{thm:seq-edge-norm}.
Here, $T_1$ and $T_2$ denote the subsets of non-collapsing and collapsing edges of $T$,
respectively.
The key lies in the choice of $\Vrt$ and $\Hol$ such that $T_1$ satisfies
compatibility and edge-commutativity with respect to $(\Vrt,\Hol)$,
where $\Hol$ is typically set as $\Ind_R(V_{\Vrt})^\tpp$.

\subsection{Designing $\Vrt$ for shallow TRS}
\label{subsec:constructor-extension-shallow}

Given a diverging sequence $u \barrow{\ast}{} v$,
our aim is to construct a finite confluent ARS that contains
a confluent sequence $u \downarrow v$ in a bottom-up manner; i.e.,
for a confluent $R$-graph $G$ obtained by the induction hypothesis,
we will extend $G$ by one layer using Theorem~\ref{thm:seq-edge-norm}
under edge compatibility and edge commutativity.
In a shallow TRS,
the depth of all variables in the rewrite rules is at most $1$,
and thus a single-layer extension up to depth $1$ is sufficient.
Hence, we design the $R$-graph $M_F(G)$ as $\Vrt$.
The design of $\Hol$, which consists of candidate root reductions to be added,
follows that of $\Vrt$ as $\Hol:= \ Ind_R(V_{\Vrt})^\tpp$.

Some edges in $\Vrt$ are used as fillers (red edges in Fig.~\ref{fig:abs-diag})
to reach an augmented top edge (the dashed red edge in the figure).
Since root reductions may be non-linear (i.e., some subterms
at non-linear variable positions may be rewritten), they need to be
synchronized by augmented edges in $\Vrt$.
Thus, the non-root reduction graph $\Vrt$ must be confluent.

The following Lemma~\ref{lem:non-E-overlap-shallow-compatible}
and~\ref{lem:shallow-nonCol-normalizable} show the conditions for
the compatibility and the edge commutativity of non-collapsing top edges,
respectively.

\begin{lemma}\label{lem:non-E-overlap-shallow-compatible}
  Let $R$ be a shallow and non-$E$-overlapping TRS.
  Suppose $G$ is an $R$-graph, and $T$ is an $R$-graph consisting of non-collapsing
  top edges such that $V_T \subseteq V_{M_F(G)}$.
  Then $T$ is compatible over $M_F(G)$.
\end{lemma}

\begin{proof}
  Let $\Vrt := M_F(G)$. Consider the sequence
  $s' \Larrow{\tpp}{}{}{T} s \barrow{\ast}{\Vrt} t \Rarrow{\tpp}{}{}{T} t'$.
  Since $R$ is non-$E$-overlapping, the rules used in the two top edges
  $\Larrow{\tpp}{}{}{T}$ and $\Rarrow{\tpp}{}{}{T}$ must be identical.
  Thus, let $\ell \to r$ denote the common rule, and let $\sigma, \sigma'$ be
  substitutions such that
  $s = \ell\sigma$, $s' = r\sigma$, $t = \ell\sigma'$, and $t' = r\sigma'$.

  If the signature contains no non-constant function symbols, then there are no
  inner reductions.
  Hence $s=t$, which implies $s'=t'$, so $T$ is compatible over $\Vrt$.

  Assume that the signature contains non-constant function symbols.
  Since $s = \ell\sigma \barrow{\ast}{\Vrt} t = \ell\sigma'$, for each variable
  $x \in \Var(\ell)$, we have $x\sigma \barrow{\ast}{\Sub(\Vrt)} x\sigma'$.
  Because $\Sub(\Vrt) = G$ by
  Proposition~\ref{prop:subterm}(\ref{prop:subterm-2}) and $r$ is a shallow
  term, the rewrite property of shallow TRSs ensures that
  $r\sigma \barrow{\ast}{\Vrt} r\sigma'$.
  (Specifically, for each subterm $r|_i$ that is a ground normal form of $G$,
  $r|_i\sigma = r|_i\sigma'$, and for each variable $x \in \Var(r)$,
  $x\sigma \barrow{\ast}{G} x\sigma'$ holds.)
  Thus, $s' = r\sigma \barrow{\ast}{\Vrt} r\sigma' = t'$, which proves that $T$
  is compatible over $\Vrt$.
\end{proof}

\begin{lemma}\label{lem:shallow-nonCol-normalizable}
  Let $R$ be a shallow TRS.
  Let $\Vrt$ be the $R$-graph $M_F(G)$ constructed from a confluent $R$-graph
  $G$.
  For a non-collapsing rule $\ell\to r$ and a substitution $\sigma$ such that
  $\ell\sigma,r\sigma \in V_{\Vrt}$,
  the pair $(\ell\sigma,r\sigma)$ is edge-commutative over
  $\< \Vrt, \Ind_R(V_{\Vrt})^\tpp \>$.
\end{lemma}

\begin{proof}
  Since $G$ is confluent, $\Vrt$ is also confluent by
  Proposition~\ref{prop:monotonic-extension}(\ref{prop:monotonic-extension-a}).
  Suppose $\ell\sigma \rarrow{\ast}{\Vrt} u_0$ for some $u_0\in V_{\Vrt}$.
  Since $\ell$ is a non-variable shallow term and $\Vrt$ consists of inner
  reductions,
  we have $x\sigma \rarrow{\ast}{G} u_0|_p$ for each variable $x \in \Var(\ell)$
  and each position $p \in \Pos_X(\ell)$ such that $\ell|_p=x$.
  Since $G$ is confluent, for each $x \in \Var(\ell)$,
  all such subterms $u_0|_p$ (which are reducts of $x\sigma$)
  can be joined to a common reduct in $G$, which we denote by $u_x$.
  Let $\sigma'$ be a substitution defined by $x\sigma' = u_x$
  for each $x \in \Var(\ell)$.
  Then, we obtain $u_0 \rarrow{\ast}{\Vrt} \ell\sigma'$.

  Since $r$ is also a non-variable shallow term (as $\ell \to r$ is
  non-collapsing), the sequence $x\sigma \rarrow{\ast}{G} x\sigma'$ for each
  $x \in \Var(\ell)$ implies that $r\sigma \rarrow{\ast}{\Vrt} r\sigma'$.
  Since $\ell\sigma', r\sigma' \in V_{\Vrt}$, we have an edge
  $(\ell\sigma',r\sigma') \in {\Rarrow{\tpp}{R}{}{}}|_{V_{\Vrt} \times V_{\Vrt}}
  = \Ind_R(V_{\Vrt})^\tpp$.
  Therefore, the pair $(\ell\sigma,r\sigma)$ is edge-commutative over
  $(\Vrt,\Ind_R(V_{\Vrt})^\tpp)$.
\end{proof}

\subsection{Confluence without collapsing rules}
\label{subsec:simple-proof}

In this subsection, we provide an alternative proof of Theorem 5
in~\cite{Sakai10}
by applying Theorem~\ref{thm:seq-edge-norm}, which establishes confluence
of non-$E$-overlapping, shallow, and non-collapsing TRSs.
In subsequent sections, we extend this method to
non-$E$-overlapping and shallow TRSs
(Section~\ref{subsec:ShallowConfluenceWithCollapsing})
and non-$E$-overlapping and weakly shallow TRSs
(Section~\ref{subsec:Confluence}), respectively.
This incremental approach allows us to effectively isolate the difficulties
arising from collapsing rules and weakly shallow rules.

Let $\Vrt:= M_F(G)$ and $\Hol:= \Ind_R(V_{\Vrt})^\tpp$.
Note that the subterm-closed property is not
  required for non-collapsing TRSs.

\begin{lemma}\label{lem:NoncollapsingShallow}
  Let $R$ be a shallow, non-collapsing, and non-$E$-overlapping TRS.
  For a finite $R$-graph $G_0$ and a finite and confluent
  $R$-graph $G$ such that $\Sub(G_0) \subseteq G$,
  there exists a finite and confluent
  $R$-graph $G'$ such that $G_0 \subseteq G'$.
\end{lemma}

\begin{proof}
  Let $T := G_0^\tpp$.
  We have $V_T = V_{G_0^\tpp} = V_{G_0} \subseteq V_{M_F(G)} = V_{\Vrt}$,
  where the inclusion is obtained from $\Sub(G_0) \subseteq G$
  by Lemma~\ref{lem:mono-ext}(\ref{lem:mono-ext-c}).
  Combining this with ${\Rarrow{}{}{}{G_0^\tpp}} = {\Rarrow{\tpp}{}{}{G_0}}
  \subseteq {\Rarrow{\tpp}{R}{}{}}$,
  we obtain
  $T = G_0^\tpp \subseteq
  \< V_{\Vrt}, {\Rarrow{\tpp}{R}{}{}}|_{V_{\Vrt}\times V_{\Vrt}} \>
  = \Ind_R(V_{\Vrt})^\tpp = \Hol$.
  Hence, $T$ is a finite subgraph of $\Hol$.

  \begin{description}
  \item[\underline{$\Vrt$ is confluent:}]
    $\Vrt$ is confluent from confluence of $G$ by
    Proposition~\ref{prop:monotonic-extension}
    (\ref{prop:monotonic-extension-a}).

  \item[\underline{Compatibility and Edge-commutativity of
      non-collapsing edges:}]
    $T$ is compatible over $\Vrt$ by
    Lemma~\ref{lem:non-E-overlap-shallow-compatible},
    and edge-commutative over $\< \Vrt, \Ind_R(V_{\Vrt})^\tpp \>$
    by Lemma~\ref{lem:shallow-nonCol-normalizable}.
  \end{description}

  Starting from $G_1 := G^\inn_0$, we finitely extend it to $G'_1$ and then to $G'$.
  \begin{description}
  \item[\underline{Construction of $G'_1$ by Theorem~\ref{thm:seq-edge-norm}:}]
    Let $\Vrt:= M_F(G)$, $G_1 := G^\inn_0$, $T_1 := T = G_0^\tpp$,
    $T_2 := \<\emptyset,\emptyset\>$, and $\Hol:= \Ind_R(V_{\Vrt})^\tpp$.
    Then, the premises of Theorem~\ref{thm:seq-edge-norm} hold as follows:
    \begin{itemize}
    \item $G_1 = G^\inn_0 \subseteq M_F(\Sub(G_0)) \subseteq M_F(G) = \Vrt$,
      where the former inclusion follows from
      Lemma~\ref{lem:monotonic-extension}(2).
    \item $T_1 = T$ is a finite subgraph of $\Hol$.
    \item $V_T = V_{T_1\cup T_2} = V_{G_0^\tpp} = V_{G_0^\inn} = V_{G_1}$.
    \item Conditions (i)--(iv) hold by the discussion above and
      $T_2 = \<\emptyset,\emptyset\>$.
    \end{itemize}
    By Theorem~\ref{thm:seq-edge-norm}, there exist a finite confluent ARS
    $G'_1$ and a finite ARS $T' \subseteq \Hol$ such that
    \begin{itemize}
      \item[(1)] $G_1 \subseteq G'_1 \conserveq \Vrt$,
      \item[(2)] ${\barrow{\ast}{G_1 \cup T}} \subseteq
        {\barrow{\ast}{G'_1 \cup T'}}$, and
      \item[(3)] $G'_1 \cup T'$ is confluent and finite.
    \end{itemize}

  \item[\underline{Construction of $G'$ and its properties:}]
  Let $G' := G'_1 \cup T' \cup T_1 \cup T_2$.
  Since $G'_1 \cup T'$, $T_1$, and $T_2$ are all finite, $G'$ is finite.
  By property (2), ${\rarrow{}{T}} \subseteq {\barrow{\ast}{G'_1 \cup T'}}$.
  Combining this with property (3), i.e., $G'_1 \cup T'$ is confluent,
  we obtain that $G'$ is confluent by Lemma~\ref{lem:inc-confluence}.
  From Lemma~\ref{lem:inner-edges}(\ref{lem:inner-edges-a}),
  we have $G_0 = G_0^\inn \cup G_0^\tpp = G_1 \cup T_1 \subseteq G'$.
  \end{description}
\end{proof}

Note that, unlike Lemma~\ref{lem:CollapsingShallow} and
Lemma~\ref{lem:main}, Lemma~\ref{lem:NoncollapsingShallow} does not assume the
subterm-closed property for $G$ or $G'$.
When a TRS is shallow, each variable occurrence in the left- and right-hand
sides of a rule is at depth $1$.
This ensures that confluent sequences in subterms are preserved at the same
depth. The subterm-closed property becomes necessary when:
\begin{itemize}
\item merging multiple confluent sequences at different depths, or
\item applying collapsing rules that alter the depth of a confluent sequence.
\end{itemize}
The former arises in Proposition~\ref{prop:sub-MF2}
to guarantee confluence of $G^\infty$ in Lemma~\ref{cor:B-monotonic} in
Section~\ref{subsec:constructor-extension},
which merges confluent sequences at different depths while maintaining
confluence.
The latter occurs in conditions (b) and (c) in the proof of
Lemma~\ref{lem:CollapsingShallow} (originating from
Lemma~\ref{lem:NestedCollapsing}) when handling collapsing rules such as
$h(x) \to x$.

\begin{example}\label{ex:non-collapsing}
  Consider the $R_{\ref{trs:simple-div}}$-graph $G_0$ shown in
  Fig.~\ref{fig:graph-divergence},
  here $R_{\ref{trs:simple-div}} = \{ f(x,x) \to g(x), a \to b, b \to a \}$.
  Here, $\Sub(G_0) = \{ a \rightleftarrows b \}$ is already confluent,
  which we take as $G$.

  To apply Theorem~\ref{thm:seq-edge-norm}, we set $\Vrt:= M_F(G)$
  (shown in Fig.~\ref{fig:simple-div-MFG}), inducing the top edges
  $\Hol=\{ f(a,a) \to g(a), g(b,b) \to g(b) \}$.
  We define $G_1:=G_0^\inn$ and $T_1:=G_0^\tpp$, where $G_1 \conserveq \Vrt$.
  \begin{figure}
    \begin{center}
      \mbox{}\\[-1em]
      \(
      \begin{tikzpicture} [line width=0.07em]
        \matrix[matrix of math nodes,row sep=1em,column sep=1.2em,minimum width=0.5em] (m) {
          f(a,b) & f(a,a) & g(a) & a\\
          f(b,a) & f(b,b) & g(b) & b\\
        };
        \path[->] ($(m-1-1.east)!0.2!(m-1-1.north east)$) edge ($(m-1-2.west)!0.2!(m-1-2.north west)$);
        \path[<-] ($(m-1-1.east)!0.2!(m-1-1.south east)$) edge ($(m-1-2.west)!0.2!(m-1-2.south west)$);
        \path[->] ($(m-1-2.south west)!0.4!(m-1-2.south)$) edge (m-2-1);
        \path[<-] (m-1-2) edge ($(m-2-1.north east)!0.4!(m-2-1.north)$);
        \path[->] ($(m-1-1.south east)!0.4!(m-1-1.south)$) edge (m-2-2);
        \path[<-] (m-1-1) edge ($(m-2-2.north west)!0.4!(m-2-2.north)$);
        \path[->] ($(m-2-1.east)!0.2!(m-2-1.north east)$) edge ($(m-2-2.west)!0.2!(m-2-2.north west)$);
        \path[<-] ($(m-2-1.east)!0.2!(m-2-1.south east)$) edge ($(m-2-2.west)!0.2!(m-2-2.south west)$);
        \path[->] ($(m-1-3.south)!0.1!(m-1-3.south east)$) edge ($(m-2-3.north)!0.1!(m-2-3.north east)$);  
        \path[<-] ($(m-1-3.south)!0.1!(m-1-3.south west)$) edge ($(m-2-3.north)!0.1!(m-2-3.north west)$);  

      \end{tikzpicture}
      \)
      \\[-2em]\mbox{}
    \end{center}
    \caption{$M_F(G)$ for Example~\ref{ex:non-collapsing}}
    \label{fig:simple-div-MFG}
  \end{figure}
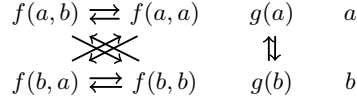
  In the process of applying Theorem~\ref{thm:seq-edge-norm}, we extend $G_1$ to
  be confluent while preserving $G_1 \conserveq \Vrt$
  using Lemma~\ref{lem:merge}.
  This construction involves non-determinism due to
  Lemma~\ref{lem:abs-edge-extension2},
  offering several choices for the join sequence between $f(b,b)$ and $f(a,a)$
  We show two such choices:
  \begin{enumerate}
  \item[(a)] We add $f(b,b) \to f(a,b) \to f(a,a)$.
    The resulting $G_1$ consists of all vertices and solid edges in
    Fig.~\ref{fig:simple-div}(\subref{fig:simple-div-a}).
    In this case, $G_1 \cup T_1$ is already confluent.
  \item[(b)] We add $f(b,b) \to f(b,a) \gets f(a,a)$.
    The resulting $G'_1 \cup T_1 \cup T'_1$ is displayed in
    Fig.~\ref{fig:simple-div}(\subref{fig:simple-div-b}),
    where $G'_1$ consists of all black vertices and solid edges,
    and $T'_1 = \{ f(b,b) \to g(b) \}$.
  \end{enumerate}
  In both cases, the construction yields a confluent
  $R_{\ref{trs:simple-div}}$-graph containing $G_0$.

  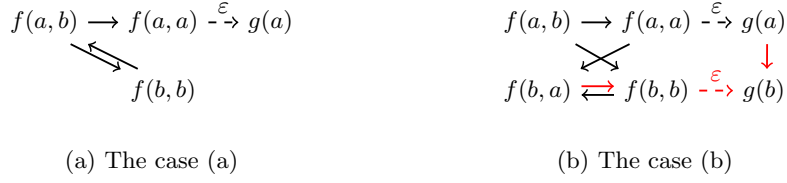
\begin{figure}
    \begin{minipage}{0.49\linewidth}
      \begin{center}
        \mbox{}\\[-1em]
        \(
        \begin{tikzpicture} [line width=0.07em]
          \matrix[matrix of math nodes,row sep=1em,column sep=1.2em,minimum width=0.5em] (m) {
            f(a,b) & f(a,a) & g(a)\\
            & f(b,b) & \\
          };
          \path[->] (m-1-1) edge (m-1-2);
          \path[->] ($(m-1-1.south east)!0.4!(m-1-1.south)$) edge (m-2-2);
          \path[<-] (m-1-1) edge ($(m-2-2.north west)!0.4!(m-2-2.north)$);

          \path[->,dashed] (m-1-2) edge node [above] {\scriptsize $\varepsilon$} (m-1-3);
        \end{tikzpicture}
        \)\\[-2em]\mbox{}
      \end{center}
      \subcaption{The case (a)}\label{fig:simple-div-a}
    \end{minipage}
    \begin{minipage}{0.49\linewidth}
      \begin{center}
        \mbox{}\\[-1em]
        \(
        \begin{tikzpicture} [line width=0.07em]
          \matrix[matrix of math nodes,row sep=1em,column sep=1.2em,minimum width=0.5em] (m) {
            f(a,b) & f(a,a) & g(a)\\
            f(b,a) & f(b,b) & g(b)\\
          };
          \path[->] (m-1-1) edge (m-1-2);
          \path[->] ($(m-1-2.south west)!0.4!(m-1-2.south)$) edge (m-2-1);
          \path[->] (m-1-1) edge (m-2-2);
          \path[->,color=red] ($(m-2-1.east)!0.2!(m-2-1.north east)$) edge ($(m-2-2.west)!0.2!(m-2-2.north west)$);
          \path[<-] ($(m-2-1.east)!0.2!(m-2-1.south east)$) edge ($(m-2-2.west)!0.2!(m-2-2.south west)$);
          \path[->,color=red] ($(m-1-3.south)!0.1!(m-1-3.south east)$) edge ($(m-2-3.north)!0.1!(m-2-3.north east)$);  

          \path[->,dashed] (m-1-2) edge node [above] {\scriptsize $\varepsilon$} (m-1-3);
          \path[->,dashed,color=red]  (m-2-2) edge node [above] {\scriptsize $\varepsilon$} (m-2-3)  ;
        \end{tikzpicture}
        \)\\[-2em]\mbox{}
      \end{center}
      \subcaption{The case (b)}\label{fig:simple-div-b}
    \end{minipage}
    \caption{Cases for Example~\ref{ex:non-collapsing}}\label{fig:simple-div}
  \end{figure}
\end{example}

  Using the preceding lemma along with the induction scheme
  presented in Section~\ref{subsec:framework}, we obtain the following theorem.
\begin{theorem}\label{thm:shallow-noncollapsing}
  For a non-$E$-overlapping, shallow, and non-collapsing TRS $R$ and
  a finite $R$-graph $G_0 = \<V_{G_0}, \rarrow{}{G_0}\>$,
  there exists a finite and confluent $R$-graph $G_0'$ such that
  $G_0 \subseteq G_0'$.
\end{theorem}

\begin{proof}
  Let $G_0 = \<V_0, \rarrow{}{0}\>$ be a finite $R$-graph.
  We proceed by induction on the sum of the sizes of terms in $V_0$,
  i.e., $\sum_{v \in V_0} |v|$.

  If $G_0$ contains no vertices, then $G_0$ is already confluent, and we set
  $G_0' := G_0$.
  Otherwise, by the induction hypothesis, there exists a finite and confluent
  $R$-graph $G$ such that $\Sub(G_0) \subseteq G$.
  Applying Lemma~\ref{lem:NoncollapsingShallow}, we obtain a finite and
  confluent $R$-graph $G_0'$ such that $G_0 \subseteq G_0'$.
\end{proof}


\begin{corollary}\label{cor:shallow-noncollapsing}
  Non-$E$-overlapping, shallow, and non-collapsing TRSs are confluent.
\end{corollary}
\begin{proof}
  Immediate from Theorem~\ref{thm:shallow-noncollapsing}, since any divergent sequence
  in the TRS can be embedded into a finite $R$-graph $G_0$,
  which is then extended to a confluent $R$-graph $G'_0$.
\end{proof}

\subsection{Confluence with collapsing rules}
\label{subsec:ShallowConfluenceWithCollapsing}

Let $\Vrt:= M_F(G)$ and $\Hol:= \Ind_R(V_{\Vrt})^\tpp$.

\begin{lemma}\label{lem:CollapsingShallow}
  Let $R$ be a shallow and non-$E$-overlapping TRS.
  For a finite $R$-graph $G_0$ and a finite, confluent, and subterm-closed
  $R$-graph $G$ satisfying $\Sub(G_0) \subseteq G$,
  there exists a finite, confluent, and subterm-closed
  $R$-graph $G'$ such that $G_0 \subseteq G'$.
\end{lemma}

\begin{proof}
  Let $T := G_0^\tpp \cup G^\tpp$ and $\Vrt:= M_F(G)$.
  Note that $\Vrt$ is finite.

  \begin{description}
  \item[\underline{$\Vrt$ is $T$-subterm-closed:}]
    Since $G$ is subterm-closed, \emph{$\Vrt$ is $G^\tpp$-subterm-closed} by
    Lemma~\ref{lem:subterm-closed}.
    Since $V_{G_0} \subseteq V_\Vrt$ follows from $\Sub(G_0) \subseteq G$ and
    Lemma~\ref{lem:mono-ext}(\ref{lem:mono-ext-c}),
    Proposition~\ref{prop:T-subterm-closed} implies that $\Vrt$ is
    $T$-subterm-closed, i.e.,
    \begin{align}
      & \label{eq:Vrt_T_subterm-closed}
      \Sub(\Vrt^\ntpp) \subseteq \Vrt \cup T
    \end{align}

  \item[\underline{Compatibility and Edge-commutativity of non-collapsing edges:}]
    $T^{nc}$ is compatible over $\Vrt$ by
    Lemma~\ref{lem:non-E-overlap-shallow-compatible},
    and edge-commutative over $(\Vrt, \Ind_R(V_{\Vrt})^\tpp)$ by
    Lemma~\ref{lem:shallow-nonCol-normalizable}.

  \item[\underline{Compatibility and Edge-commutativity of collapsing edges:}]
    By Lemma~\ref{lem:NestedCollapsing}, we have $T'^c (\subseteq T^c)$ with
    \begin{description}
    \item[(a)]
      ${\barrow{\ast}{\Vrt\cup T}}={\barrow{\ast}{\Vrt\cup T'^c \cup T_1}}$,
    \item[(b)] $T'^c$ is $T^{nc}$-compatible over $\Vrt$, and
    \item[(c)] $T'^c$ is $T^{nc}$-edge-commutative over $(\Vrt,\Hol)$.
    \end{description}
  \end{description}

  Let $G_1 = M_F(G)$. We will extend $G_1$ to $G'_1$ and then to $G'$.
  \begin{description}
  \item[\underline{Construction of $G'_1$ by Theorem~\ref{thm:seq-edge-norm}:}]
    Let $G_1 := \Vrt$, $T_1 := T^{nc}$, $T_2 := T'^c$,
    and $\Hol:= \Ind_R(V_{\Vrt})^\tpp$ in Theorem~\ref{thm:seq-edge-norm}.
    Then, the premises hold as follows:
    \begin{itemize}
    \item Since $G$ is finite and confluent, so is $G_1 = M_F(G)$ by
      Proposition~\ref{prop:monotonic-extension}(1).
    \item Trivially $G_1 \subseteq \Vrt$.
    \item Since $\Vrt$ is $T$-subterm-closed,
      we have $V_{T_1\cup T_2} \subseteq V_{T} \subseteq V_{\Vrt} = V_{G_1}$.
    \item Since $T_1 \cup T_2 \subseteq T$,
      ${\Rarrow{\tpp}{}{}{T}} = {\Rarrow{\tpp}{}{}{T^\tpp}} \subseteq
      {\Rarrow{\tpp}{R}{}{}}$, and $V_{T} \subseteq V_{\Vrt}$,
      we obtain $T_1 \cup T_2 \subseteq \Ind_R(V_{\Vrt})^\tpp = \Hol$.
    \item Conditions (i), (iii), and (iv) hold by the discussion above.
      Condition (ii) holds by Lemma~\ref{lem:mixed-compat}: $T$ consists of
        top edges, $V_T \subseteq V_{\Vrt}$ (shown above), $\Vrt = M_F(G)$
        consists of inner edges (Definition~\ref{def:mono-ext}), $T^{nc}$ is
        compatible over $\Vrt$
        (shown above via Lemma~\ref{lem:non-E-overlap-shallow-compatible}), and
        $T'^c$ is $T^{nc}$-compatible over $\Vrt$ by
        Lemma~\ref{lem:NestedCollapsing}(\ref{lem:NestedCollapsing_b}).
    \end{itemize}
    Theorem~\ref{thm:seq-edge-norm} concludes that there exist a finite
    confluent ARS $G'_1$ and a finite ARS $T' \subseteq \Hol$ such that
    \begin{itemize}
    \item[(1)] $G_1 \subseteq G'_1 \conserveq \Vrt$,
    \item[(2)] $V_{T'} \subseteq V_{G'_1}$,
    \item[(3)] ${\barrow{\ast}{G_1 \cup T_1 \cup T_2}} \subseteq
      {\barrow{\ast}{G'_1 \cup T'}}$, and
    \item[(4)] $G'_1 \cup T'$ is confluent and finite.
    \end{itemize}
    Since $G_1 = \Vrt$, the property (1) implies $G'_1= \Vrt$.

  \item[\underline{Construction of $G'$:}]
  Let $G' := G'_1 \cup T' \cup T$, then the following properties hold.
    \begin{description}
    \item[{\bf $G'$ is finite and confluent:}]
      $G'$ is finite since $G'_1 \cup T'$ and $T$ are. Combining
      ${\to_T} \subseteq {\to_{\Vrt \cup T_1\cup T_2}} \subseteq {\to_{G'_1 \cup T'}}$
      with the fact that $G'_1 \cup T'$ is confluent,
      we obtain that $G'$ is confluent by Lemma~\ref{lem:inc-confluence}.
    \item[{\bf $G'$ is subterm-closed:}]
      $V_{T \cup T'} \subseteq V_{\Vrt} = V_{G'_1}$.
      Moreover, $T$ and $T'$ are top edges, i.e., $T\cup T' = (T\cup T')^\tpp$.
      Thus, we obtain by (\ref{lem:inner-edges-d}) and (\ref{lem:inner-edges-e})
      of Lemma~\ref{lem:inner-edges} that
      $\Sub(G'^\ntpp) = \Sub((G'_1 \cup T\cup T')^\ntpp)
      = \Sub((\Vrt \cup (T\cup T')^\tpp)^\ntpp)
      = \Sub(\Vrt^\ntpp \cup ((T\cup T')^\tpp)^\ntpp)
      = \Sub(\Vrt^\ntpp \cup \<V_{T\cup T'},\emptyset\>) = \Sub(\Vrt^\ntpp)$.
      From \eqref{eq:Vrt_T_subterm-closed}, we obtain
      $\Sub(G'^\ntpp)=\Sub(\Vrt^\ntpp) \subseteq \Vrt \cup T \subseteq G'$.
    \item[{\bf $G_0 \subseteq G'$:}]
      Since $G_0 = G_0^\tpp \cup G_0^\inn$, it suffices to show that both
      $G_0^\tpp$ and $G_0^\inn$ are included in $G'$.
      The former inclusion holds by the construction of $T$ and $G'$, i.e.,
      $G_0^\tpp \subseteq T \subseteq G'$.
      For the latter, we have $G_0^\inn \subseteq M_F(\Sub(G_0))$ by
      Lemma~\ref{lem:mono-ext}(2).
      Since $\Sub(G_0) \subseteq G$, Proposition~\ref{prop:monotonic-extension}
      (\ref{prop:monotonic-extension-b}) implies
      $M_F(\Sub(G_0)) \subseteq M_F(G) = \Vrt$.
      From these facts, we obtain $G_0^\inn \subseteq \Vrt \subseteq G'$.
    \end{description}
  \end{description}
\end{proof}

The following example illustrates the transformation described in
Lemma~\ref{lem:CollapsingShallow}.
This also clarifies why we adopted the somewhat unnatural definition of
the subterm-closed property.
\begin{example}\label{ex:counter}
  Let $F= \{ a, f(), h() \}$ and
  $R_{\settrs{trs:collapsing}}= \{ f(x) \to h(x),\ h(x) \to x \}$.
  Consider applying the lemma to $G_0$ shown in
  Fig.~\ref{fig:collapse}(\subref{subfig:collapse_G0}).
  From
  \( \Sub(G_0)= \{ \ a,\ \  f(a) \to h(a) \ \}, \)
  a recursive application constructs $G=\Sub(G_0)$.
  The graph $M_F(G)$ is depicted in
  Fig.~\ref{fig:collapse}(\subref{subfig:collapse_MFG}).
  Here, we have $G^\tpp=\{ \ f(a)\to h(a) \ \}$ and
  $G_0^\tpp = \{ \ f(a)\gets h(f(a)) \ \}$.
  Given $T:= G^\tpp \cup G_0^\tpp$, $\Vrt:= M_F(G)$, and
  $\Hol :=  \Ind_R(V_{\Vrt})^\tpp $ as shown in
  Fig.~\ref{fig:collapse}(\subref{subfig:collapse_Hol}),
  $T_1$ and $T_2$ are determined to be $T_1=G^\tpp$ and $T_2=G_0^\tpp$.
  Theorem~\ref{thm:seq-edge-norm} then produces
  $G'_1 := \Vrt$ and $T':= \{ f(a) \to h(a),\ h(h(a)) \to h(a) \}$.
  The resulting graph $G'_0 := \Vrt \cup T' \cup T$ is shown in
  Fig.~\ref{fig:collapse}(\subref{subfig:collapse_Gd0}).

  If the subterm-closed property of $G$ were defined simply as
  $\Sub(G) \subseteq G$, then $G_0'$ would not satisfy this property,
  because $h(a) \to a$ is in $\Sub(G'_0)$ but not in $G'_0$
  (see Fig.~\ref{fig:collapse}(\subref{subfig:collapse_Gd0})).
\end{example}
\begin{figure}
  \begin{minipage}{0.4\linewidth}
    \begin{center}
      \mbox{}\\[-1em]
      \(
      \begin{tikzpicture} [line width=0.07em]
    \matrix (m) [matrix of math nodes,row sep=0.8em,column sep=1em,minimum width=0.5em] {
      & f(a) & h(f(a)) & \mbox{\hspace{16pt}}\\
      &      & h(h(a)) & \\
    };
    \path[->]
  (m-1-3) edge (m-1-2)
  (m-1-3) edge (m-2-3)
        ;
      \end{tikzpicture}
      \)\\[-3em]\mbox{}
    \end{center}
    \subcaption{$G_0$}\label{subfig:collapse_G0}
  \end{minipage}
  \begin{minipage}{0.4\linewidth}
    \begin{center}
      \mbox{}\\[-1em]
      \(
      \begin{tikzpicture} [line width=0.07em]
  \matrix (m) [matrix of math nodes,row sep=0.8em,column sep=1em,minimum width=0.5em] {
    a & f(a) & h(f(a)) & f(f(a))\\
    & h(a) & h(h(a)) & f(h(a))\\
  };
  \path[->]
  (m-1-3) edge (m-2-3)
  (m-1-4) edge (m-2-4)
  ;
      \end{tikzpicture}
      \)\\[-3em]\mbox{}
    \end{center}
    \subcaption{$M_F(G)$}\label{subfig:collapse_MFG}
  \end{minipage}

\bigskip
\begin{minipage}{0.4\linewidth}
  \begin{center}
    \mbox{}\\[-1em]
    \(
    \begin{tikzpicture} [line width=0.07em]
      \matrix (m) [matrix of math nodes,row sep=0.8em,column sep=1em,minimum width=0.5em] {
  a & f(a) & h(f(a)) & f(f(a))\\
  & h(a) & h(h(a)) & f(h(a))\\
      };
      \path[->]
      (m-1-2) edge (m-2-2)
      (m-2-2) edge (m-1-1)
      (m-1-3) edge (m-1-2) (m-1-4) edge (m-1-3)
      (m-2-3) edge (m-2-2) (m-2-4) edge (m-2-3)
      ;
    \end{tikzpicture}
    \)\\[-3em]\mbox{}
  \end{center}
  \subcaption{$\Hol$}\label{subfig:collapse_Hol}
\end{minipage}
\begin{minipage}{0.4\linewidth}
  \begin{center}
    \mbox{}\\[-1em]
    \(
    \begin{tikzpicture} [line width=0.07em]
      \matrix (m) [matrix of math nodes,row sep=0.8em,column sep=1em,minimum width=0.5em] {
  a & f(a) & h(f(a)) & f(f(a))\\
  & h(a) & h(h(a)) & f(h(a))\\
      };
      \path[->]
      (m-1-3) edge (m-1-2) edge (m-2-3)
      (m-1-2) edge (m-2-2)
      (m-2-3) edge (m-2-2)
      (m-1-4) edge (m-2-4)
      ;
    \end{tikzpicture}
    \)\\[-3em]\mbox{}
  \end{center}
  \subcaption{$G_0'$}\label{subfig:collapse_Gd0}
\end{minipage}
\caption{$R_{\ref{trs:collapsing}}$-graphs in Example~\ref{ex:counter}}\label{fig:collapse}
\end{figure}
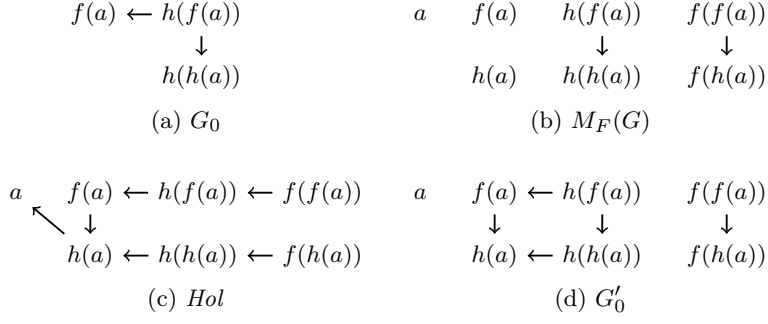

\begin{theorem}\label{thm:NOV-Shallow}
  For a non-$E$-overlapping and shallow TRS $R$ and a finite $R$-graph
  $G_0 = \<V_{G_0}, \rarrow{}{G_0}\>$,
  there exists a finite, confluent, and subterm-closed $R$-graph $G_0'$
  such that $G_0 \subseteq G_0'$.
\end{theorem}
\begin{proof}
  The theorem can be proved in the exact same manner as
  Theorem~\ref{thm:shallow-noncollapsing},
  by replacing the use of Lemma~\ref{lem:NoncollapsingShallow} with
  Lemma~\ref{lem:CollapsingShallow}
  and incorporating the subterm-closed property for $G_0$ and $G$.
\end{proof}

\begin{corollary}\label{CorollaryNOV-Shallow}
  Non-$E$-overlapping and shallow TRSs are confluent.
\end{corollary}

\section{Non-$E$-overlapping and weakly shallow TRSs}
\label{sec9}\label{sec:WS_system}

In this section, we relax the shallowness condition in
Theorem~\ref{thm:NOV-Shallow} to weak shallowness.
To this end, we first construct an appropriate structure for $\Vrt$.
Specifically, we introduce a novel structure called the constructor expansion
in Section~\ref{subsec:constructor-extension}, and then prove confluence
in Section~\ref{subsec:Confluence}.

Note that while this section specifically focuses on weakly shallow TRSs,
the weak shallowness assumption has already been utilized since Section~\ref{subsec:normalization_weaklyshallow}.

\subsection{Designing $\Vrt$ for weakly shallow TRSs}
\label{subsec:constructor-extension}

\begin{definition} \label{def:constructor-expansion}
  For an $R$-graph $G$ and a non-negative integer $k$, we use the following notations:
  \begin{enumerate}
    \item $\overline{M_{C}}(G)$ to denote $G \cup M_C(G)$, and
    \item $\overline{M_C}^k(G)$ (and $(M_C)^k(G)$) to denote the $k$-fold application
    of $\overline{M_C}$ (and $M_C$, respectively).
  \end{enumerate}
\end{definition}

\begin{example}\label{ex:constructor-expansion}
  Consider the $R_{\ref{trs:running}}$-graph $G_0$ from Example~\ref{ex:overview}:
  \[ G_0 = f(g(c),c) \gets f(c,c) \to f(c,g(c)) \Rarrow{\tpp}{}{}{} g^3(c). \]
  We choose $G = \Sub(G_1) = \{ c\to g(c),\, g^2(c) \}$ as given in Example~\ref{ex:sub-ext}.
  Then, we have:
  \[
    \begin{array}{l@{\ $=$\ }l}
    M_C(G)                & \{\ \phantom{c \to {}} g(c) \to g^2(c), \phantom{\to\hspace{0.1em}} g^3(c) \ \} \\
    \overline{M_{C}}(G)   & \{\ c \to g(c) \to g^2(c), \hspace{1.1em} g^3(c) \ \} \\
    \overline{M_{C}}^2(G) & \{\ c \to g(c) \to g^2(c) \to g^3(c),\ g^4(c) \ \}
    \end{array}
  \]
\end{example}

\begin{proposition} \label{prop:constructor-expansion}
  For an $R$-graph $G$ and an integer $k \geq 0$,
  \begin{enumerate}
    \item $\overline{M_C}^k(G) \subseteq \overline{M_C}^{k+1}(G)$, and
    \item
    $\displaystyle
    \overline{M_C}^k(G) =
    \begin{cases}
      G & \text{if $k=0$,} \\
      G \cup M_C(\overline{M_C}^{k-1}(G)) & \text{if $k>0$.}
    \end{cases}
    $
  \end{enumerate}
\end{proposition}

\begin{proof}
  \mbox{}
  \begin{enumerate}
    \item For $k=0$, we have
    $\overline{M_C}^1(G) = G \cup M_C(G) \supseteq G = \overline{M_C}^0(G)$.

    For $k>0$, we have
    $\overline{M_{C}}^{k+1}(G) = \overline{M_C}(\overline{M_C}^{k}(G))
    = \overline{M_{C}}^{k}(G) \cup M_C(\overline{M_{C}}^k(G))
    \supseteq \overline{M_{C}}^{k}(G)$.

    \item We proceed by induction on $k$.
    The cases $k=0$ and $k=1$ are trivial.
    For $k>1$, we have
    $\overline{M_{C}}^{k}(G) = \overline{M_{C}}^{k-1}(G) \cup M_C(\overline{M_{C}}^{k-1}(G))$.
    By the induction hypothesis,
    $\overline{M_{C}}^{k-1}(G) \cup M_C(\overline{M_{C}}^{k-1}(G))
    = G \cup M_C(\overline{M_{C}}^{k-2}(G)) \cup M_C(\overline{M_{C}}^{k-1}(G))$.
    Here, $M_C(\overline{M_{C}}^{k-2}(G)) \subseteq M_C(\overline{M_{C}}^{k-1}(G))$
    follows from Part~(1) and Proposition~\ref{prop:monotonic-extension}(\ref{prop:monotonic-extension-b}).
    Thus, we obtain
    $\overline{M_{C}}^{k}(G) = G \cup M_C(\overline{M_{C}}^{k-1}(G))$.
  \end{enumerate}
\end{proof}

In light of Proposition~\ref{prop:constructor-expansion}(1),
we introduce the notion of $\overline{M_{C}}^\infty(G)$ as follows:
\begin{definition} \label{def:constructor-expansion2}
  We define $\displaystyle \overline{M_{C}}^\infty(G) := \bigcup_{i=0}^\infty \overline{M_{C}}^i(G)$.
\end{definition}
\begin{example}\label{ex:inf}
  In continuation of Example~\ref{ex:constructor-expansion}, we have:
  \[
  \overline{M_{C}}^\infty(G) = \{ c \to g(c) \to g^2(c) \to g^3(c) \to \cdots \}
  \]
\end{example}

We now present some properties of subterm-closed $R$-graphs.

\begin{proposition}\label{prop:sub-MF2}
  For a subterm-closed $R$-graph $G$,
  we have $G\restRoot{C} \subseteq M_C(G) \cup G|_\emptyset$. 
\end{proposition}

\begin{proof}
  Since $G\restRoot{C} = G^\ntpp\restRoot{C} = G^\inn\restRoot{C}$,
  we have $G\restRoot{C} \subseteq M_C(\Sub(G^\ntpp)) \cup G|_\emptyset$
  by Proposition~\ref{prop:subterm}(\ref{prop:subterm-3}).
  Combining this with $M_C(\Sub(G^\ntpp)) \subseteq M_C(G)$
  (which follows from Proposition~\ref{prop:monotonic-extension}(\ref{prop:monotonic-extension-b}) and the fact that $G$ is subterm-closed), 
  we obtain
  $G\restRoot{C} \subseteq M_C(G) \cup G|_\emptyset$.
\end{proof}

\begin{lemma}\label{lem:B-monotonic}
  For a subterm-closed $R$-graph $G$,
  \begin{enumerate}
    \item\label{lem:B-monotonic_2} $\overline{M_C}(G)$ is subterm-closed.
    \item\label{lem:B-monotonic_3} $\overline{M_C}(G)$ is confluent if $G$ is confluent.
  \end{enumerate}
\end{lemma}
\begin{proof}
  Let $G = \<V,\rightarrow\>$.
  \begin{enumerate}
    \item
    We have $(G\cup M_C(G))^\ntpp \subseteq G^\ntpp \cup M_C(G)^\ntpp$ by
    Lemma~\ref{lem:inner-edges}(\ref{lem:inner-edges-d}), and hence
    $\Sub((G\cup M_C(G))^\ntpp) \subseteq \Sub(G^\ntpp \cup M_C(G)^\ntpp)$ by
    Proposition~\ref{prop:subterm}(\ref{prop:subterm-1}).
    Combining this with
    $\Sub(G^\ntpp \cup M_C(G)^\ntpp) \subseteq \Sub(G^\ntpp) \cup \Sub(M_C(G)^\ntpp)$
    obtained by Proposition~\ref{prop:subterm}(\ref{prop:subterm-0}), we have
    \begin{equation}\label{eq:B-mono1}
      \Sub((G \cup M_C(G))^\ntpp) \subseteq \Sub(G^\ntpp) \cup \Sub(M_C(G)^\ntpp).
    \end{equation}
    Here, by Proposition~\ref{prop:subterm}\eqref{prop:subterm-5},
    Proposition~\ref{prop:monotonic-extension}(\ref{prop:monotonic-extension-b}),
    and Proposition~\ref{prop:subterm}(\ref{prop:subterm-2}) we obtain
    \begin{equation}\label{eq:B-mono2}
      \Sub(M_C(G)^\ntpp) = \Sub(M_C(G)) = G.
    \end{equation}
    Combining \eqref{eq:B-mono1} and \eqref{eq:B-mono2} with
    $\Sub(G^\ntpp) \subseteq G$ (which follows from the fact that $G$ is subterm-closed),
    we obtain $\Sub((G \cup M_C(G))^\ntpp) \subseteq
    \Sub(G^\ntpp) \cup \Sub(M_C(G)^\ntpp) \subseteq G \subseteq G \cup M_C(G)$.
    We are done since $\overline{M_C}(G) = G \cup M_C(G)$.

    \item
    We have $G\restRoot{C} \subseteq M_C(G) \cup G|_\emptyset$ by
    Proposition~\ref{prop:sub-MF2},
    which means that ${\to_G}\restRoot{C} \subseteq {\to_{M_C(G)}}$.
    On the other hand, obviously ${\to_G}\restRoot{D} =
    {\to_G}|_{V\restRoot{D}\times V_G} \subseteq (V_G\setminus V_{M_C(G)}) \times V_G$.
    Hence, ${\to_G} = {\to_G}\restRoot{D} \cup {\to_G}\restRoot{C} \subseteq
    ((V_G\setminus V_{M_C(G)}) \times V_G) \cup {\barrow{\ast}{M_C(G)}}$.
    Thus, $\overline{M_C}(G) = G \cup M_C(G) = G \gtrdot M_C(G)$.

    Since $G$ and $M_C(G)$ are confluent, $\overline{M_C}(G)$ is confluent
    by instantiating $G_1 := G$ and $G_2 := M_C(G)$	in Lemma~\ref{lem:hc-confluence-union}.
  \end{enumerate}
\end{proof}

\begin{lemma}\label{lem:MF-inf-union}
  For an $R$-graph $G$,
  \[
    M_F\bigl(\overline{M_C}^\infty(G)\bigr)
    = \bigcup_{i=0}^\infty M_F\bigl(\overline{M_C}^i(G)\bigr).
  \]
\end{lemma}
\begin{proof}
  ($\supseteq$)
  Since $\overline{M_C}^i(G) \subseteq \overline{M_C}^\infty(G)$ for every $i$,
  Proposition~\ref{prop:monotonic-extension}(\ref{prop:monotonic-extension-b})
  gives $M_F(\overline{M_C}^i(G)) \subseteq M_F(\overline{M_C}^\infty(G))$.

  ($\subseteq$)
  Let $X_i := \overline{M_C}^i(G)$, then
  $\overline{M_C}^\infty(G) = \bigcup_{i=0}^\infty X_i$ with
  $X_0 \subseteq X_1 \subseteq \cdots$
  by Definition~\ref{def:constructor-expansion2} and
  Proposition~\ref{prop:constructor-expansion}(1).
  By Definition~\ref{def:mono-ext}, a vertex of $M_F(\bigcup_i X_i)$ is a term
  $f(s_1,\ldots,s_n)$ with $s_1,\ldots,s_n \in \bigcup_i V_{X_i}$, and an edge of
  $M_F(\bigcup_i X_i)$ is induced by one edge of $\bigcup_i X_i$ together with
  finitely many further vertices from $\bigcup_i V_{X_i}$.
  Each of these finitely many vertices, and the inducing edge in the edge case,
  lies in some $X_i$; taking $k$ as the largest of these finitely many indices,
  all of them lie in $X_k$ since $X_0 \subseteq X_1 \subseteq \cdots$.
  Hence the vertex (resp.\ edge) belongs to $M_F(X_k)$ by
  Definition~\ref{def:mono-ext}.
\end{proof}

\begin{lemma}\label{cor:B-monotonic}
  For a subterm-closed $R$-graph $G$ and $k \in \{0,1,\ldots\} \cup \{\infty\}$,
  we have:
  \begin{enumerate}
    \item\label{cor:B-monotonic_2} $\overline{M_C}^k(G)$ is subterm-closed.
    \item\label{cor:B-monotonic_2'}
    $M_F(\overline{M_{C}}^k(G))$ and $M_F(\overline{M_{C}}^\infty(G))$ are
    $G^\varepsilon$-subterm-closed.
    \item\label{cor:B-monotonic_3} $\overline{M_{C}}^k(G)$,
    $M_C(\overline{M_{C}}^k(G))$, and $M_F(\overline{M_{C}}^k(G))$ are confluent
    if $G$ is confluent.
    \item\label{cor:B-monotonic_4} $V_{\Sub(\overline{M_{C}}^\infty(G))} \subseteq V_{\overline{M_{C}}^\infty(G)}$.
  \end{enumerate}
\end{lemma}

\begin{proof}\mbox{}
  \begin{enumerate}
    \item By induction on $k$, it is immediate from Lemma~\ref{lem:B-monotonic}(\ref{lem:B-monotonic_2}).

    \item By Part~(\ref{cor:B-monotonic_2}) and Lemma~\ref{lem:subterm-closed},
      $M_F(\overline{M_C}^k(G))$ is $(\overline{M_C}^k(G))^\varepsilon$-subterm-closed,
      which means that $V_{M_F(\overline{M_C}^k(G))} \supseteq V_{(\overline{M_C}^k(G))^\varepsilon}$
      and $\Sub(M_F(\overline{M_C}^k(G))^\ntpp) \subseteq
      M_F(\overline{M_C}^k(G)) \cup (\overline{M_C}^k(G))^\varepsilon$.
      Combining these with $(\overline{M_C}^k(G))^\varepsilon =
      (G \cup M_C(\overline{M_C}^{k-1}(G)))^\varepsilon =
      G^\varepsilon \cup M_C(\overline{M_C}^{k-1}(G))|_\emptyset$
      obtained by Proposition~\ref{prop:constructor-expansion}(2) and Proposition~\ref{prop:subterm}\eqref{prop:subterm-5}, 
      we obtain $V_{(\overline{M_C}^k(G))^\varepsilon} \supseteq V_{G^\tpp}$ and
      $\Sub(M_F(\overline{M_C}^k(G))^\ntpp) \subseteq
      G^\varepsilon \cup M_C(\overline{M_C}^{k-1}(G))$.
      Combining this with $M_C(\overline{M_C}^{k-1}(G)) \subseteq M_F(\overline{M_C}^k(G))$ obtained by
      Proposition~\ref{prop:constructor-expansion}(1) and by  \eqref{prop:monotonic-extension-b} and \eqref{prop:monotonic-extension-c}
      of Proposition~\ref{prop:monotonic-extension},
      we have $\Sub(M_F(\overline{M_C}^k(G))^\ntpp) \subseteq G^\varepsilon \cup M_F(\overline{M_C}^k(G))$,
      which means that $M_F(\overline{M_C}^k(G))$ is $G^\varepsilon$-subterm-closed.
    %
    Moreover, $M_F(\overline{M_{C}}^\infty(G))$ is $G^\varepsilon$-subterm-closed
    since
    $M_F(\overline{M_{C}}^\infty(G))=\bigcup^\infty_{i=0} M_F(\overline{M_{C}}^i(G))$
    by Lemma~\ref{lem:MF-inf-union}.

    \item
    By induction on $k$, confluence of $\overline{M_{C}}^k(G)$ is immediate from
    Lemma~\ref{lem:B-monotonic}(\ref{lem:B-monotonic_3}).
    Moreover, $M_C(\overline{M_{C}}^k(G))$ and $M_F(\overline{M_{C}}^k(G))$ are confluent
    by Proposition~\ref{prop:monotonic-extension}(\ref{prop:monotonic-extension-a}).

    \item[\eqref{cor:B-monotonic_4}]
    We have
    \begin{align*}
      V_{\Sub(\overline{M_C}^\infty(G))}
      &= V_{\Sub(G \cup M_C(\overline{M_C}^\infty(G)))} \\
      &\subseteq V_{\Sub(G)} \cup V_{\Sub(M_C(\overline{M_C}^\infty(G)))} \quad \text{(by Proposition~\ref{prop:subterm}\eqref{prop:subterm-0})} \\
      &\subseteq V_{\Sub(G)} \cup V_{\overline{M_C}^\infty(G)} \quad \text{(by Proposition~\ref{prop:subterm}\eqref{prop:subterm-2})}.
    \end{align*}
    Since $G$ is subterm-closed, $V_{\Sub(G^\ntpp)} \subseteq V_G$.
    From the fact that $V_{\Sub(G^\ntpp)} = \Sub(V_{G^\ntpp}) = \Sub(V_G) = V_{\Sub(G)}$, this implies $V_{\Sub(G)} \subseteq V_G$.
    On the other hand, $G = \overline{M_C}^0(G) \subseteq \overline{M_C}^\infty(G)$, we have $V_G \subseteq V_{\overline{M_C}^\infty(G)}$.
    Consequently, $V_{\Sub(\overline{M_C}^\infty(G))} \subseteq V_G \cup V_{\overline{M_C}^\infty(G)} = V_{\overline{M_C}^\infty(G)}$.
  \end{enumerate}
\end{proof}

We denote $M_F(\overline{M_C}^\infty(G))$ (resp., $M_F(\overline{M_C}^k(G))$)
by $G^\infty$ (resp., $G^k$).
In the sequel, we use $G^\infty$ as $\Vrt$.

\subsection{Compatibility and edge commutativity of non-collapsing rules in
  non-$E$-overlapping and weakly shallow TRSs} 
\label{subsec:normalization_weaklyshallow}

\begin{lemma}\label{lem:subterm-closed2}
  Let $G=\< V_G,\rarrow{}{G}\>$ be a subterm-closed $R$-graph.
  Assume that
  \begin{enumerate}
    \item $s[t]_p \rarrow{}{G} s[t']_p$, and
    \item for any $p' < p$, $(s[t]_p)|_{p'} \rarrow{}{G} (s[t']_p)|_{p'}$ implies
    $(s[t]_p)|_{p'} \Rarrow{\ntpp}{}{}{G} (s[t']_p)|_{p'}$.
  \end{enumerate}
  Then, $t \rarrow{}{G} t'$.
\end{lemma}

\begin{proof}
  We proceed by induction on the length $|p|$ of $p$.
  If $p=\varepsilon$, the claim is trivial.
  Suppose $p=iq$ and $s = f(s_1,\ldots,s_n)$.
  By Assumption (2) with $p'=\varepsilon$,
  we have $s[t]_p \Rarrow{\ntpp}{}{}{G} s[t']_p$.
  Since $G$ is subterm-closed, this implies $s_i[t]_q \rarrow{}{G} s_i[t']_q$.
  Hence, $t \rarrow{}{G} t'$ holds by the induction hypothesis.
\end{proof}

\begin{lemma} \label{lem:rules}
  Let $R$ be a TRS consisting of rules with weakly shallow left-hand sides,
  and let $\ell \rightarrow r \in R$.
  For a subterm-closed $R$-graph $G$,
  \begin{enumerate}
    \item\label{lem:rules_1}
    If $\ell \sigma \rarrow{\ast}{G^\infty} \ell\theta$,
    then $x\sigma \rarrow{\ast}{\overline{M_C}^\infty(G)} x\theta$
    for each variable $x\in \Var(\ell)$.

    \item\label{lem:rules_3}
    Assume that $G$ is confluent.
    If $\ell\sigma \rarrow{\ast}{G^\infty} u$,
    then there exists a substitution $\theta$ such that
    $u \rarrow{\ast}{G^\infty} \ell\theta$ and
    $x\sigma \rarrow{\ast}{\overline{M_C}^\infty(G)} x\theta$
    for each variable $x\in \Var(\ell)$.
  \end{enumerate}
\end{lemma}

\begin{proof}\mbox{}
  \begin{enumerate}
    \item[\eqref{lem:rules_1}]
    Let $\ell=f(\ell_1,\ldots,\ell_n)$. Then each $\ell_i$ is a semi-constructor 
    term since $\ell$ is weakly shallow.
    We also have $\ell_i\sigma\rarrow{\ast}{\overline{M_C}^\infty(G)} \ell_i\theta$
    for each $i$, because $G^\infty = M_F(\overline{M_C}^\infty(G))$ consists of
    inner edges.
    The function symbols appearing on a path to a variable position 
    $p \in \Pos_{X}(\ell_i)$ are constructor symbols, i.e., $\Root(\ell_i|_q)$ is 
    a constructor symbol for all $q < p$.
    Thus, we have $x\sigma \rarrow{\ast}{\overline{M_C}^\infty(G)} x\theta$
    for every $x\in\Var(\ell)$ by Lemma~\ref{lem:subterm-closed2}.

    \item[\eqref{lem:rules_3}]
    Note that both $\overline{M_C}^\infty(G)$ and $G^\infty$ are confluent by 
    Lemma~\ref{cor:B-monotonic}(\ref{cor:B-monotonic_3}).

    Since $R$ is weakly shallow,
    the sequence $\ell\sigma \rarrow{\ast}{G^\infty} u$ contains no steps 
    at any position $q \in \Pos_{F}(\ell)$ with $q < p$ for some $p \in \Pos_X(\ell)$.
    By a similar argument as in Part~\eqref{lem:rules_1}, we have 
    $\ell|_p\sigma \rarrow{\ast}{\overline{M_C}^\infty(G)} u|_p$
    for each $p \in \Pos_X(\ell)$.
    Since $\overline{M_C}^\infty(G)$ is confluent, there exists a term $v_x$ for
    each $x\in\Var(\ell)$ such that 
    \[
    \ell\sigma|_p \rarrow{\ast}{\overline{M_C}^\infty(G)} u|_p
    \rarrow{\ast}{\overline{M_C}^\infty(G)} v_x
    \quad\text{for any } p\in\Pos_{\{x\}}(\ell). 
    \]
    By taking a substitution $\theta$ with $x\theta = v_x$ for each $x \in \Var(\ell)$, we have
    $x\sigma = \ell\sigma|_p \rarrow{\ast}{\overline{M_C}^\infty(G)} v_x = x\theta$
    for each $p \in \Pos_{X}(\ell)$ with $\ell|_p = x$. 
    Since $\ell$ is weakly shallow, we obtain $u \rarrow{\ast}{G^\infty} \ell\theta$.
  \end{enumerate}
\end{proof}

\begin{remark}
  Lemma~\ref{lem:rules} assumes that $R$ is weakly shallow.
  This restriction on left-hand sides can be relaxed to
  non-$E$-overlapping, and a similar idea can be applied.
  In contrast, Lemma~\ref{lem:rules_rhs} essentially requires the
  weakly shallow assumption on the right-hand side of
  each rewrite rule.
\end{remark}

\begin{lemma} \label{lem:rules_rhs}
  Let $s$ be a weakly shallow term with $s \notin X$.
  For a subterm-closed $R$-graph $G$,
  if $s \sigma \in V_{G^\infty}$ and
  $x\sigma \rarrow{\ast}{\overline{M_C}^\infty(G)} x\theta$ for
  each variable $x \in \Var(s)$, then $s\sigma \rarrow{\ast}{G^\infty} s\theta$.
\end{lemma}

\begin{proof}
  If $s$ is a constant, the claim holds immediately since $s\sigma = s\theta$.
  Assume that $s$ is not a constant.
  Since $s$ is a weakly shallow term with $s \notin X$,
  any proper subterm $t$ of $s$ is a semi-constructor term.
  We show by induction on the structure of $t$ that if $t\sigma \in V_{\overline{M_C}^\infty(G)}$, then
  $t\sigma \rarrow{\ast}{\overline{M_C}^\infty(G)} t\theta$.

  \begin{description}
    \item[\textbf{Case 1}: $t$ is a variable $x$.]
    The claim immediately follows from the assumption that
    $x\sigma \rarrow{\ast}{\overline{M_C}^\infty(G)} x\theta$ for all $x \in \Var(s)$.

    \item[\textbf{Case 2}: $t$ is a constant.]
    The claim is trivial since $t = t\sigma = t\theta \in V_{\overline{M_C}^\infty(G)}$.

    \item[\textbf{Case 3}: $t$ is a functional term.] 
    Let $t=f(t_1,\ldots,t_n)$ for some $f \in F$.
    \begin{description}
      \item[\textbf{Case 3-1}: $f \notin C$.]
      $t$ is ground because $t$ is a semi-constructor term.
      In this case, the claim holds immediately since $t = t\sigma = t\theta \in V_{\overline{M_C}^\infty(G)}$.

      \item[\textbf{Case 3-2}: $f \in C$.]
      Since $t_j$ is a subterm of the semi-constructor term $t$, it is also a semi-constructor term.
      By Lemma~\ref{cor:B-monotonic}\eqref{cor:B-monotonic_4},
      $t\sigma \in V_{\overline{M_C}^\infty(G)}$ implies $t_j\sigma \in V_{\overline{M_C}^\infty(G)}$ for each $j$.
      The induction hypothesis then ensures that $t_j\sigma \rarrow{\ast}{\overline{M_C}^\infty(G)} t_j\theta$ for each $j$.
      Since $M_C(\overline{M_C}^\infty(G)) \subseteq \overline{M_C}^\infty(G)$,
      we obtain $t\sigma \rarrow{\ast}{\overline{M_C}^\infty(G)} t\theta$.
    \end{description}
  \end{description}

  Now, we return to the main statement of the lemma.
  Since $s \notin X$, we can write $s = f(t_1,\ldots,t_n)$, where each $t_j$ is a semi-constructor term.
  Since $s\sigma \in V_{G^\infty}$ and $G^\infty = M_F(\overline{M_C}^\infty(G))$,
  Proposition~\ref{prop:subterm}\eqref{prop:subterm-2} implies that
  $t_j\sigma \in V_{\overline{M_C}^\infty(G)}$ for each $j$.
  By the claim shown above, we obtain $t_j\sigma \rarrow{\ast}{\overline{M_C}^\infty(G)} t_j\theta$ for each $j$.
  Since $M_F(\overline{M_C}^\infty(G)) = G^\infty$, we conclude that $s\sigma \rarrow{\ast}{G^\infty} s\theta$.
\end{proof}

\begin{lemma} \label{lem:shallow-nonCol-compatible2}
  Let $R$ be a weakly shallow and non-$E$-overlapping TRS. 
  For a confluent and subterm-closed $R$-graph $G$,
  let $G^\infty = M_F(\overline{M_C}^\infty(G))$.

  Then, for an $R$-graph $T$ with $V_{T} \subseteq V_{G^\infty}$ and
  $T^\varepsilon = T$, we have
  \begin{enumerate}
    \item\label{lem:shallow-nonCol-compatible2_sub}
    $T^{nc}$ is compatible over $G^\infty$.
    \item\label{lem:shallow-nonCol-normalizable2_sub}
    $T^{nc}$ is edge-commutative
    over $(G^\infty, \Ind_{R}(V_{G^\infty})^\tpp)$.
  \end{enumerate}
\end{lemma}

\begin{proof}\mbox{}
  \begin{enumerate}
    \item[\eqref{lem:shallow-nonCol-compatible2_sub}]
    Let $u' \larrow{}{T^{nc}} u \barrow{\ast}{G^\infty} v \rarrow{}{T^{nc}} v'$. 
    From the non-$E$-overlapping property of $R$, the fact that $T^{nc} \subseteq T = T^\varepsilon$, 
    and that $G^\infty = (G^\infty)^{\inn}$ consists only of inner edges,
    the reductions $u \rarrow{}{T} u'$ and $v \rarrow{}{T} v'$ must be steps of 
    the same non-collapsing rule $\ell \rightarrow r \in R$.
    Thus, we can write $u = \ell\sigma$, $u' = r\sigma$, $v = \ell\sigma'$, and $v' = r\sigma'$ for some substitutions $\sigma$ and $\sigma'$.

    Lemma~\ref{cor:B-monotonic}(\ref{cor:B-monotonic_3}) and the confluence of $G$
    imply the confluence of $G^\infty$, and thus
    there exists a term $w$ such that
    $\ell\sigma \rarrow{\ast}{G^\infty} w \larrow{\ast}{G^\infty} \ell\sigma'$.
    By Lemma~\ref{lem:rules}\eqref{lem:rules_3}, there exists a substitution
    $\theta$ such that $w \rarrow{\ast}{G^\infty} \ell\theta$ and
    $x \sigma \rarrow{\ast}{\overline{M_C}^\infty(G)} x\theta$ for each $x \in \Var(\ell)$.
    Since $\ell\sigma'\rarrow{\ast}{G^\infty} w \rarrow{\ast}{G^\infty} \ell\theta$,
    Lemma~\ref{lem:rules}\eqref{lem:rules_1} implies that
    $x\sigma' \rarrow{\ast}{\overline{M_C}^\infty(G)} x\theta$ for each $x \in \Var(\ell)$.

    Since $\ell \rightarrow r$ is a non-collapsing rule, the right-hand side $r$ is not a variable ($r \notin X$).
    Since $(\ell\sigma, r\sigma) \in T^{nc} \subseteq T$ and $V_T \subseteq V_{G^\infty}$,
    we have $r\sigma \in V_{G^\infty}$ and similarly $r\sigma' \in V_{G^\infty}$.
    Because $r$ is weakly shallow and we have both $x\sigma \rarrow{\ast}{\overline{M_C}^\infty(G)} x\theta$
    and $x\sigma' \rarrow{\ast}{\overline{M_C}^\infty(G)} x\theta$, 
    Lemma~\ref{lem:rules_rhs} ensures that
    $r\sigma \rarrow{\ast}{G^\infty} r\theta$ and $r\sigma' \rarrow{\ast}{G^\infty} r\theta$. 
    Therefore, we obtain $u' = r\sigma \rarrow{\ast}{G^\infty} r\theta \larrow{\ast}{G^\infty} r\sigma' = v'$,
    which proves that $T^{nc}$ is compatible over $G^\infty$.

    \item[\eqref{lem:shallow-nonCol-normalizable2_sub}]
    Consider an edge in $T^{nc}$, which can be written as
    $(\ell\sigma, r\sigma) \in T^{nc}$ for some non-collapsing rule $\ell\to r$
    and substitution $\sigma$.
    Let $u_0 \in V_{G^\infty}$ such that $\ell\sigma \rarrow{\ast}{G^\infty} u_0$. 
    By Lemma~\ref{lem:rules}\eqref{lem:rules_3}, there exists a substitution
    $\theta$ such that $u_0 \rarrow{\ast}{G^\infty} \ell\theta \in V_{G^\infty}$ and 
    $x\sigma \rarrow{\ast}{\overline{M_C}^\infty(G)} x\theta$ for each variable
    $x \in \Var(\ell)$.
    Since $(\ell\sigma,r\sigma) \in T^{nc} \subseteq T$ and $V_T \subseteq V_{G^\infty}$, we have $r\sigma \in V_{G^\infty}$.
    By Lemma~\ref{lem:rules_rhs}, it follows that $r\sigma \rarrow{\ast}{G^\infty} r\theta$.
    Since $(\ell\theta,r\theta) \in \left.{\Rarrow{\tpp}{R}{}{}}\right|_{V_{G^\infty}\times V_{G^\infty}}$, 
    the pair $(\ell\sigma,r\sigma)$ is edge-commutative over 
    $(G^\infty, \<V_{G^\infty},{\Rarrow{\tpp}{R}{}{}}|_{V_{G^\infty}\times V_{G^\infty}}\>)
    = (G^\infty, \Ind_R(V_{G^\infty})^\tpp)$.
  \end{enumerate}
\end{proof}

\subsection{Confluence of non-$E$-overlapping and weakly shallow TRS}
\label{subsec:Confluence}

Let $\Vrt:= G^\infty = M_F(\overline{M_C}^\infty(G))$
and $\Hol:= \Ind_R(V_{\Vrt})^\tpp$.

\begin{lemma}\label{lem:main}
  Let $R$ be a weakly shallow and non-$E$-overlapping TRS.
  For a finite $R$-graph $G_0$ and a finite, confluent, and
  subterm-closed $R$-graph $G$ such that $\Sub(G_0) \subseteq G$,
  there exists a finite, confluent, and subterm-closed $R$-graph $G'$
  such that $G_0 \subseteq G'$.
\end{lemma}

\begin{proof}
  Let $T := G_0^\tpp \cup G^\tpp$.
  Then, by the definition of $T$,
  \begin{align}
    T^\tpp & = T,          \label{eq:T-tpp-T} \\
    G^\tpp & \subseteq T.  \label{eq:G-tpp-T}
  \end{align}

  Let $G'' := G_0^\ntpp \cup G^\ntpp$. Then, we have:
  \begin{align}
    G_0    & = G_0^\tpp \cup G_0^\ntpp && \text{(by Lemma~\ref{lem:inner-edges}
      \eqref{lem:inner-edges-a})} \notag \\
    & \subseteq T \cup G'', \label{eq:G0-TGd} \\
    G      & = G^\tpp \cup G^\ntpp     && \text{(by Lemma~\ref{lem:inner-edges}
      \eqref{lem:inner-edges-a})} \notag \\
    & \subseteq T \cup G''. \label{eq:G-TGd}
  \end{align}

  Since $G$ is subterm-closed ($\Sub(G^\ntpp) \subseteq G$),
    we have $G^\ntpp \subseteq M_F(\Sub(G^\ntpp)) \subseteq M_F(G)$
    by Lemma~\ref{lem:mono-ext}\eqref{lem:mono-ext-b} and the monotonicity of
    $M_F$ (Proposition~\ref{prop:monotonic-extension}%
    \eqref{prop:monotonic-extension-b}).
  Together with $G_0^\ntpp \subseteq M_F(G)$
    (from $\Sub(G_0) \subseteq G$ and
    Lemma~\ref{lem:mono-ext}\eqref{lem:mono-ext-c}),
    this gives $G'' = G_0^\ntpp \cup G^\ntpp \subseteq M_F(G)$.
  Furthermore, since
  $G \subseteq \overline{M_C}^\infty(G)$ holds by
  Proposition~\ref{prop:constructor-expansion}, the monotonicity of $M_F$
  (Proposition~\ref{prop:monotonic-extension}\eqref{prop:monotonic-extension-b})
  implies $M_F(G) \subseteq M_F(\overline{M_C}^\infty(G)) = G^\infty$.
  Hence,
  \begin{align}
    G'' \subseteq M_F(G) & \subseteq G^\infty, \label{eq:G1-MFG} \\
    V_T = V_{G''} & \subseteq V_{G^\infty}, \label{eq:VT-VInf}
  \end{align}
  where the latter uses $V_T = V_{G''}$
    (from $V_{G_0^\tpp} = V_{G_0^\ntpp}$ and $V_{G^\tpp} = V_{G^\ntpp}$).
  Also, by Lemma~\ref{cor:B-monotonic}\eqref{cor:B-monotonic_3} and the
    confluence of $G$,
  \begin{align}
    & \text{$G^\infty$ is confluent.} \label{eq:Ginf-conf}
  \end{align}

  \noindent
  Next, we establish the properties needed to apply
    Lemma~\ref{lem:NestedCollapsing} and Theorem~\ref{thm:seq-edge-norm}:
  \begin{description}
  \item[{\bf Compatibility and Edge-commutativity of non-collapsing edges:}]
    The premises of Lemma~\ref{lem:shallow-nonCol-compatible2} are satisfied
    by the premises of this lemma, together with
    \eqref{eq:T-tpp-T} and \eqref{eq:VT-VInf}. Hence,
    \begin{align}
      & \text{$T^{nc}$ is compatible over $G^\infty$, and} \label{eq:nc-compat}\\
      & \text{$T^{nc}$ is edge-commutative over
        $(G^\infty,\Ind_R(V_{G^\infty})^\tpp)$.} \label{eq:nc-commut}
    \end{align}

  \item[{\bf Compatibility and Edge-commutativity of collapsing edges:}]
    We apply Lemma~\ref{lem:NestedCollapsing} with $T_1 := T^{nc}$,
      $T_2 := T^{c}$, and $\Vrt := G^\infty$, whose premises hold as follows.
      The graph $T$ consists only of top edges by \eqref{eq:T-tpp-T}.
      From Lemma~\ref{cor:B-monotonic}\eqref{cor:B-monotonic_2'}, $G^\infty$ is
      $G^\varepsilon$-subterm-closed, and hence $T$-subterm-closed by
      Proposition~\ref{prop:T-subterm-closed}, \eqref{eq:G-tpp-T}, and
      \eqref{eq:VT-VInf}.  It is confluent by \eqref{eq:Ginf-conf}, and
      consists of inner edges by Definition~\ref{def:mono-ext}.
    Hence,
    by Lemma~\ref{lem:NestedCollapsing},
    there exists $T'^c$ such that
    \begin{align}
      & T'^c \subseteq T^c , \label{eq:Td-Tc} \\
      & {\barrow{\ast}{G^\infty \cup T}} =
      {\barrow{\ast}{G^\infty \cup T^{nc} \cup T'^c}}, \label{eq:GT-equality} \\
      & \text{$T'^c$ is $T^{nc}$-compatible over $G^\infty$, and}
      \label{eq:c-compat} \\
      & \text{$T'^c$ is $T^{nc}$-edge-commutative over
        $(G^\infty,\Ind_R(V_{G^\infty})^\tpp)$.} \label{eq:c-commut}
    \end{align}
    Since $T = T^{nc} \cup T^{c}$, \eqref{eq:Td-Tc} immediately implies
    \begin{align}
      & T^{nc} \cup T'^c \subseteq T. \label{eq:TncTd-T}
    \end{align}
  \end{description}

  \medskip
  The rest of the proof extends $G''$ to the required $G'$ in four steps:
    a finite approximation $G_1 \subseteq G^\infty$ of the infinite $G^\infty$
    (Step~1); its edge-normalization $G'_1$ via
    Theorem~\ref{thm:seq-edge-norm}, making $G'_1 \cup T'$ confluent (Step~2);
    a finite confluent $G_2 = G'_1 \cup T' \cup T$ with $G_0 \cup G \subseteq G_2$
    (Step~3); and finally $G'$, obtained by adjoining a constructor layer to
    $G_2$ to recover subterm-closedness (Step~4).
  \begin{description}
  \item[\underline{Step 1. Construction of $G_1$:}]
    We construct a finite extension $G_1$ of $G''$ satisfying
    \eqref{eq:G1T-incl} and \eqref{eq:G1_construction}.
    Taking $G_1 := G^\infty$ would satisfy \eqref{eq:G1T-incl}
      via \eqref{eq:GT-equality}, but $G^\infty$ is not finite in general.

    Let $A := G'' \cup T$ and $B := G^\infty \cup T^{nc} \cup T'^c$.
    Since $G'' \subseteq G^\infty$ by \eqref{eq:G1-MFG}, we have
      ${\barrow{\ast}{A}} = {\barrow{\ast}{G'' \cup T}}
      \subseteq {\barrow{\ast}{G^\infty \cup T}} = {\barrow{\ast}{B}}$,
      where the last equality is \eqref{eq:GT-equality}.
    Since $A$ is finite (as $G_0$ and $G$ are),
    Proposition~\ref{Prop:FiniteApproxARS} shows that
    there exists a finite ARS $C \subseteq B$ with
    ${\barrow{\ast}{A}} \subseteq {\barrow{\ast}{C}}$, i.e.,
    \begin{align}
      & C \subseteq G^\infty \cup T^{nc} \cup T'^c, \label{eq:C_graph}\\
      & {\barrow{\ast}{G'' \cup T}} \subseteq {\barrow{\ast}{C}}.
      \label{eq:C_arrow}
    \end{align}
    Let $G_1 := (C \setminus (T^{nc} \cup T'^c)) \cup G''$.
    Then, it follows that $G_1 \cup T^{nc} \cup T'^c \supseteq C$.
    Combining this with \eqref{eq:C_arrow}, we obtain
    ${\barrow{\ast}{G''\cup T}}\subseteq{\barrow{\ast}{G_1\cup T^{nc}\cup T'^c}}$.
    Thus, we have
    \begin{align}
      & {\barrow{\ast}{T}} \subseteq {\barrow{\ast}{G_1 \cup T^{nc} \cup T'^c}}.
      \label{eq:G1T-incl}
    \end{align}
    By the construction of $G_1$, we also have
    \begin{align}
      & G'' \subseteq G_1. \label{eq:Gdd-G1}
    \end{align}
    From \eqref{eq:C_graph} and the definition of $G_1$,
    we have $G_1 \subseteq G^\infty \cup G''$.
    Since $G'' \subseteq G^\infty$ holds by \eqref{eq:G1-MFG},
    we obtain $G_1 \subseteq G^\infty$. In summary,
    \begin{align}
     & G'' \subseteq G_1 \subseteq G^\infty. \label{eq:G1_construction}
    \end{align}

    \item[\underline{Step 2. Construction of $G'_1$ by Theorem~\ref{thm:seq-edge-norm}:}]
      Let $\Vrt:= G^\infty$, $T_1:=T^{nc}$, $T_2:=T'^c$, and
      $\Hol:= \Ind_R(V_{\Vrt})^\tpp$ in Theorem~\ref{thm:seq-edge-norm}.
      Then, its premises are satisfied as follows:
      \begin{description}
      \item[\textbf{$G_1$ is a finite ARS:}]
        since $C$ is finite (by Proposition~\ref{Prop:FiniteApproxARS}) and
        $G''$ is finite (from the finiteness of $G_0$ and $G$), their union
        $G_1 = (C \setminus (T^{nc} \cup T'^c)) \cup G''$ is finite.
      \item[\textbf{$\Vrt = G^\infty$ is confluent:}]
        by \eqref{eq:Ginf-conf}.
      \item[\textbf{$G_1 \subseteq G^\infty$:}]
        by \eqref{eq:G1_construction}.
      \item[\textbf{$T_1$ and $T_2$ are finite:}]
        since $T^{nc}, T'^c \subseteq T = G_0^\tpp \cup G^\tpp$ and
        $G_0$, $G$ are finite.
      \item[\textbf{$T_1 \cup T_2 \subseteq \Ind_R(V_{G^\infty})^\tpp$:}]
        since $T_1 \cup T_2 \subseteq T$ by \eqref{eq:TncTd-T},
        ${\Rarrow{\tpp}{}{}{T}} = {\Rarrow{\tpp}{}{}{T^\tpp}} \subseteq
        {\Rarrow{\tpp}{R}{}{}}$ by \eqref{eq:T-tpp-T} and
        $T := G_0^\tpp \cup G^\tpp$, and $V_T \subseteq V_{G^\infty}$ by
        \eqref{eq:VT-VInf}.
      \item[\textbf{$V_{T_1 \cup T_2} \subseteq V_{G_1}$:}]
        we have
        $V_{T_1\cup T_2} = V_{T^{nc} \cup T'^c} \subseteq V_T = V_{G''}\subseteq V_{G_1}$
        by \eqref{eq:TncTd-T}, \eqref{eq:VT-VInf}, and \eqref{eq:Gdd-G1}.
      \item[\textbf{Conditions (i), (iii), and (iv) hold:}]
        by \eqref{eq:nc-compat}, \eqref{eq:nc-commut}, and \eqref{eq:c-commut},
        respectively.
      \item[\textbf{Condition (ii) ($T_1 \cup T_2$ is $T_1$-compatible
          over $\Vrt$) holds:}]
        by Lemma~\ref{lem:mixed-compat}, whose premises hold by
          \eqref{eq:T-tpp-T} ($T$ consists of top edges), \eqref{eq:VT-VInf}
          ($V_T \subseteq V_{G^\infty}$), the inner-edge property of $G^\infty$
          (Definition~\ref{def:mono-ext}),
          \eqref{eq:nc-compat} ($T^{nc}$ is compatible over $G^\infty$), and
          \eqref{eq:c-compat} ($T'^c$ is $T^{nc}$-compatible over $G^\infty$).
      \end{description}
      \smallskip
      By Theorem~\ref{thm:seq-edge-norm}, there exist $G'_1$ and $T'$ such that
      \begin{align}
        & T' \subseteq \Ind_R(V_{G^\infty})^\tpp, \label{eq:Td-Ind} \\
        & G_1 \subseteq G'_1 \conserveq G^\infty, \label{eq:G1d-conserv} \\
        & {\barrow{\ast}{G_1 \cup T^{nc} \cup T'^c}} \subseteq
        {\barrow{\ast}{G'_1\cup T'}}, \label{eq:G1d-inc} \\
        & \text{$G'_1 \cup T'$ is confluent and finite.} \label{eq:conf-finite}
      \end{align}

\item[\underline{Step 3. Construction of confluent $G_2$:}]
  Let $G_2 := G'_1 \cup T' \cup T$.
  Then $G_2$ is finite by \eqref{eq:conf-finite} and the finiteness of $T$.
  We also have ${\to_T} \subseteq {\barrow{\ast}{T}}
  \subseteq {\barrow{\ast}{G_1 \cup T^{nc} \cup T'^{c}}}
  \subseteq {\barrow{\ast}{G'_1 \cup T'}}$
  from \eqref{eq:G1T-incl} and \eqref{eq:G1d-inc}.
  By this and \eqref{eq:conf-finite},
    Lemma~\ref{lem:inc-confluence} gives that
    $G_2 = T \cup (G'_1 \cup T')$ is confluent.
  Furthermore, we obtain
  \begin{align}
    G_0 \cup G & \subseteq T \cup G''   && \text{(by \eqref{eq:G0-TGd} and
      \eqref{eq:G-TGd})} \notag \\
    & \subseteq T \cup G_1   && \text{(by \eqref{eq:Gdd-G1})} \notag \\
    & \subseteq T \cup G'_1  && \text{(by \eqref{eq:G1d-conserv})} \notag \\
    & \subseteq G_2          && \text{(by the construction of $G_2$)}.
    \label{eq:G-G2}
  \end{align}

  \item[\underline{Step 4. Construction of $G'$ and its properties:}]
    We construct $G'$ by $G' := G_2 \cup M_C(\overline{M_C}^{k}(G))$,
      where $k$ is a natural number satisfying
    \begin{equation}
      G_2^{\ntpp} \subseteq M_F(\overline{M_C}^{k}(G)). \label{eq:G2-Gk}
    \end{equation}
    Here, the existence of such $k$ is guaranteed as follows.
    Lemma~\ref{lem:inner-edges}\eqref{lem:inner-edges-d}
    gives $G_2^\ntpp = {G'_1}^\ntpp \cup T'^\ntpp \cup T^\ntpp$.
    We have $T' \subseteq \Ind_R(V_{G^\infty})^\tpp$ by \eqref{eq:Td-Ind}, and
      $T = T^\tpp$ with $V_T \subseteq V_{G^\infty}$ by \eqref{eq:T-tpp-T} and
      \eqref{eq:VT-VInf}; hence both $T$ and $T'$ consist of top edges with
      vertices in $V_{G^\infty}$, so $T'^\ntpp$ and $T^\ntpp$ have no edges by
      Lemma~\ref{lem:inner-edges}\eqref{lem:inner-edges-e} and contribute only
      vertices of $V_{G^\infty}$.  As ${G'_1}^\ntpp \subseteq G'_1 \subseteq
      G^\infty$ by \eqref{eq:G1d-conserv}, we obtain $G_2^\ntpp \subseteq G^\infty$.
    Since $G_2^\ntpp$ is finite and
      $G_2^\ntpp \subseteq G^\infty = M_F(\overline{M_C}^\infty(G))
      = \bigcup_{i=0}^\infty M_F(\overline{M_C}^i(G))$
      by Lemma~\ref{lem:MF-inf-union},
      with the last union increasing in $i$ by
      Proposition~\ref{prop:constructor-expansion}(1) and
      Proposition~\ref{prop:monotonic-extension}(\ref{prop:monotonic-extension-b}),
      there exists $k \in {\mathbb N}$ such that \eqref{eq:G2-Gk}.

    Now, we show that $G'$ is finite, subterm-closed, and confluent with
    $G_0 \subseteq G'$.
    \begin{description}
    \item[\textbf{$G_0 \subseteq G'$:}]
      This follows from \eqref{eq:G-G2} and $G_2 \subseteq G'$
      (by the construction of $G'$).

    \item[\textbf{$G'$ is a finite $R$-graph:}]
      Since $G_2$ and $M_C(\overline{M_C}^{k}(G))$ are finite, $G'$ is finite.

    \item[\textbf{$G'$ is subterm-closed:}]
      Since $(M_C(\overline{M_C}^k(G)))^\ntpp = M_C(\overline{M_C}^k(G))$ by
      Proposition~\ref{prop:subterm}\eqref{prop:subterm-5}
      and $M_C(\overline{M_C}^k(G)) \subseteq M_F(\overline{M_C}^{k}(G))$ by
      Proposition~\ref{prop:monotonic-extension}
      \eqref{prop:monotonic-extension-c},
      we have $G'^\ntpp = G_2^\ntpp \cup (M_C(\overline{M_C}^k(G)))^\ntpp
      \subseteq M_F(\overline{M_C}^{k}(G))$ by
      Lemma~\ref{lem:inner-edges}\eqref{lem:inner-edges-d} and \eqref{eq:G2-Gk}.
      Therefore, we obtain
      \begin{equation}
        \Sub(G'^\ntpp) \subseteq \Sub(M_F(\overline{M_C}^{k}(G)))
        \subseteq \overline{M_C}^k(G) \label{eq:subterm}
      \end{equation}
      by Proposition~\ref{prop:subterm}\eqref{prop:subterm-1} and
      \eqref{prop:subterm-2}.
      Since $\overline{M_C}^k(G) \subseteq \overline{M_C}^{k+1}(G) =
      G \cup M_C(\overline{M_C}^{k}(G))$ by
      Proposition~\ref{prop:constructor-expansion}(1) and (2), and
      $G \subseteq G_2$ by \eqref{eq:G-G2}, we have
      \[
      \Sub(G'^\ntpp) \subseteq G \cup M_C(\overline{M_C}^{k}(G)) \subseteq
      G_2\cup M_C(\overline{M_C}^k(G)) = G'
      \]
      from \eqref{eq:subterm}. This means that $G'$ is subterm-closed.

    \item[\textbf{$G'$ is confluent:}]
      Clearly,
      $G_2\restRoot{C} = G_2^\ntpp\restRoot{C}$,
      because no top reduction occurs from a constructor-rooted term.
      With \eqref{eq:G2-Gk}, we obtain
      $G_2\restRoot{C} \subseteq M_F(\overline{M_C}^{k}(G))\restRoot{C}$, and thus,
      since $M_F(\overline{M_C}^{k}(G))\restRoot{C} =
      M_C(\overline{M_C}^{k}(G)) \cup M_F(\overline{M_C}^{k}(G))|_\emptyset$ by
      Proposition~\ref{prop:subterm}\eqref{prop:subterm-4},
      \begin{equation}
        {\to_{G_2\restRoot{C}}} \subseteq {\to_{M_C(\overline{M_C}^{k}(G))}}.
        \label{eq:toG2C-MCkG}
      \end{equation}
      On the other hand, since
      $V_{G_2}\restRoot{D} \cap V_{M_C(\overline{M_C}^{k}(G))} = \emptyset$, we have
      \begin{equation}
        {\to_{G_2\restRoot{D}}} =
        \left.{{\to}_{G_2}}\right|_{V_{G_2}\restRoot{D}\times V_{G_2}}
        \subseteq (V_{G_2}\setminus V_{M_C(\overline{M_C}^{k}(G))}) \times V_{G_2}.
        \label{eq:toG2D}
      \end{equation}
      From $G_2 = G_2\restRoot{C} \cup G_2\restRoot{D}$, \eqref{eq:toG2C-MCkG},
      and \eqref{eq:toG2D}, we obtain
      \begin{equation}
        {\to_{G_2}} \subseteq
        (V_{G_2}\setminus V_{M_C(\overline{M_C}^{k}(G))}) \times V_{G_2}
        \cup {\barrow{\ast}{M_C(\overline{M_C}^{k}(G))}}.
        \label{eq:G2_hierarchical}
      \end{equation}
      Moreover, $G_2$ is confluent (Step~3), and
      $M_C(\overline{M_C}^{k}(G))$ is confluent by Lemma~\ref{cor:B-monotonic}
      \eqref{cor:B-monotonic_3} and the confluence of $G$.
      Hence, \eqref{eq:G2_hierarchical} shows that
      $G_2 \gtrdot M_C(\overline{M_C}^{k}(G))$ is a hierarchical combination.
      Thus, $G' = G_2 \cup M_C(\overline{M_C}^{k}(G))$ is confluent by
      Lemma~\ref{lem:hc-confluence-union}.
    \end{description}
  \end{description}

  This completes Step~4, and with it the construction: $G'$ is a finite,
  confluent, and subterm-closed $R$-graph with $G_0 \subseteq G'$, as required.
\end{proof}

\begin{example}\label{ex:normalization}
  We explore the previous lemma for $R_{\ref{trs:running}}$ in
  Example~\ref{ex:overview} and $R_{\ref{trs:running}}$-graph
  \[ G_0 := f(g(c),c) \gets f(c,c) \to f(c,g(c)) \Rarrow{\tpp}{}{}{} g^3(c). \]
  Let
  \[\begin{array}{l}
  G:=\Sub(G_0)=\{ c\to g(c), g^2(c)\}, \text{ and }\\
  T:=G^\tpp_0 \cup G^\tpp =\{f(c,g(c))\to g^3(c), c \to g(c)\}.
  \end{array}\]
  Then, $\Vrt:= G^\infty$ is the $R_{\ref{trs:running}}$-graph in
  Fig.~\ref{fig:running-VRT}. When applying Lemma~\ref{lem:NestedCollapsing},
  $T^{nc}=T$ and $T^{c}=\emptyset$, and thus, it has no effects,
  i.e., it results in $T'=\emptyset$.

  We apply Theorem~\ref{thm:seq-edge-norm} with
  $G_1 := G^\inn_0$, $T_1 := T^{nc} = T$, $T_2:=\emptyset$, $\Vrt = G^\infty$, and
  \[
  \Hol := \Ind_{R_{\ref{trs:running}}}(V_{G^\infty})^\tpp =
  \{ f(g^i(c),g^{i+1}(c)) \to g^{i+3}(c) \mid i \ge 0 \}.
  \]
  Fig.~\ref{fig:graph-construction-r3} demonstrates the resulted graph
  $G' := G'_1 \cup T' \cup T$ of the lemma.
  Here, the given $R_{\ref{trs:running}}$-graph $G_0$ is displayed by black solid
  edges, $G'_1$ consists of all displayed vertices and all non-top edges and
  $T'=\{ f(g(c),g^2(c)) \to g^4(c) \}$.
\end{example}

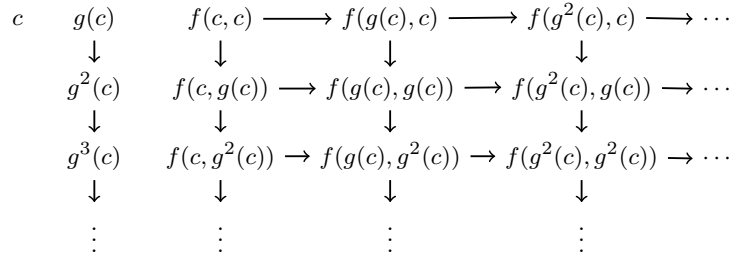
\begin{figure}[b]
  \begin{center}
    \(
    \begin{tikzpicture} [line width=0.07em]
      \matrix (m) [matrix of math nodes,row sep=0.9em,column sep=1em,minimum width=0.5em] {
        c & g(c)   & f(c,c)      & f(g(c),c)      & f(g^2(c),c)      & \cdots\\
        & g^2(c) & f(c,g(c))   & f(g(c),g(c))   & f(g^2(c),g(c))   & \cdots\\
        & g^3(c) & f(c,g^2(c)) & f(g(c),g^2(c)) & f(g^2(c),g^2(c)) & \cdots\\
        & \vdots &  \vdots     &    \vdots      &    \vdots        & \\
      };
      \path[->]
      (m-1-3) edge (m-1-4) (m-1-4) edge (m-1-5) (m-1-5) edge (m-1-6)
      (m-2-3) edge (m-2-4) (m-2-4) edge (m-2-5) (m-2-5) edge (m-2-6)
      (m-3-3) edge (m-3-4) (m-3-4) edge (m-3-5) (m-3-5) edge (m-3-6)
      (m-1-2) edge (m-2-2) (m-1-3) edge (m-2-3) (m-1-4) edge (m-2-4) (m-1-5) edge (m-2-5)
      (m-2-2) edge (m-3-2) (m-2-3) edge (m-3-3) (m-2-4) edge (m-3-4) (m-2-5) edge (m-3-5)
      (m-3-2) edge (m-4-2) (m-3-3) edge (m-4-3) (m-3-4) edge (m-4-4) (m-3-5) edge (m-4-5)
      ;
    \end{tikzpicture}
    \)
    \\[-2em]\mbox{}
  \end{center}
  \caption{$G^\infty$ for $G: \{c\to g(c), g^2(c) \}$}\label{fig:running-VRT}
\end{figure}
\begin{figure}
  \begin{center}
    \(
    \begin{tikzpicture} [line width=0.07em]
      \matrix (m) [matrix of math nodes,row sep=0.9em,column sep=1em,minimum width=0.5em] {
        c &  g(c)  & f(c,c)      & f(g(c),c)      \\
        &        & f(c,g(c))   & f(g(c),g(c))     \\
        & g^3(c) &             & f(g(c),g^2(c))   \\
        & {\color{red}g^4(c)} & & & \\
      };
      \path[->]
      (m-1-1) edge node [above] {\scriptsize $\tpp$} (m-1-2)
      ;
      \path[->]
      (m-1-3) edge (m-2-3)
      (m-1-3) edge (m-1-4)
      ;
      \path[->] (m-2-3) edge node [above] {\scriptsize $\tpp$} (m-3-2);
      \path[->, color=red]
      (m-1-4) edge (m-2-4)
      (m-2-3) edge (m-2-4)
      (m-2-4) edge (m-3-4)
      (m-3-2) edge (m-4-2)
      ;
      \path[->, color=red]  (m-3-4) edge node [above,xshift=-1.5em,yshift=-0.4em] {\scriptsize $\tpp$} (m-4-2) ;

    \end{tikzpicture}
    \)
    \\[-2em]\mbox{}
  \end{center}
  \caption{Graph construction for $R_{\ref{trs:running}}$}\label{fig:graph-construction-r3}
\end{figure}
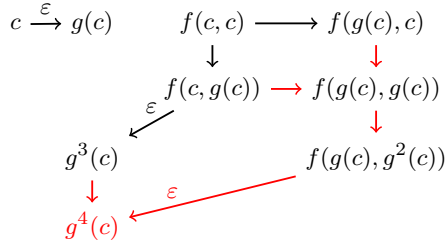

\begin{example}\label{ex:make-subterm}
In Example~\ref{ex:normalization},
$k=2$ is enough to satisfy
$G'^{\ntpp} \subseteq M_F(\bigcup_{i=0}^k \overline{M_C}^i(G))$.
Since $\overline{M_C}^2(G) = \{ c \to g(c) \to g^2(c) \to g^3(c), \ g^4(c) \}$,
the resulted subterm-closed graph $G'_0 := G' \cup M_C(\overline{M_C}^2(G))$ is
shown in Fig.~\ref{fig:graph-construction-r4}, where green vertices and edges
are introduced by $M_C(\overline{M_C}^2(G))$.
\begin{figure}
  \begin{center}
    \(
    \begin{tikzpicture} [line width=0.07em]
      \matrix (m) [matrix of math nodes,row sep=0.9em,column sep=1em,minimum width=0.5em] {
        c &  g(c)              & f(c,c)      & f(g(c),c)         \\
        & {\color{darkgreen}g^2(c)} & f(c,g(c))   & f(g(c),g(c)) \\
        & g^3(c)               &             & f(g(c),g^2(c))    \\
        & {\color{red}g^4(c)} & & & \\
        & {\color{darkgreen}g^5(c)} \\
      };
      \path[->]
      (m-1-1) edge node [above] {\scriptsize $\tpp$} (m-1-2)
      ;
      \path[->]
      (m-1-3) edge (m-2-3)
      (m-1-3) edge (m-1-4)
      ;
      \path[->] (m-2-3) edge node [above] {\scriptsize $\tpp$} (m-3-2);
      \path[->, color=red]
      (m-1-4) edge (m-2-4)
      (m-2-3) edge (m-2-4)
      (m-2-4) edge (m-3-4)
      (m-3-2) edge (m-4-2)
      ;
      \path[->, color=red]
      (m-3-4) edge node [above,xshift=-1.5em,yshift=-0.4em] {\scriptsize $\tpp$} (m-4-2)
      ;
      \path[->, color=darkgreen]
      (m-1-2) edge (m-2-2)
      (m-2-2) edge (m-3-2)
      ;

    \end{tikzpicture}
    \)
    \\[-2em]\mbox{}
  \end{center}
  \caption{Graph construction for $R_{\ref{trs:running}}$}\label{fig:graph-construction-r4}
\end{figure}
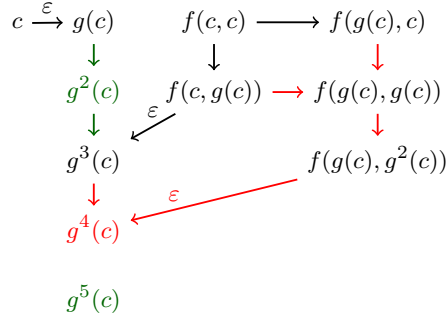
\end{example}

\begin{theorem}\label{thm:main}
  For a non-$E$-overlapping and weakly shallow TRS $R$ and
  a finite $R$-graph $G_0 = \<V_{G_0}, \rarrow{}{G_0}\>$,
  there is a finite, confluent, and subterm-closed $R$-graph $G_0'$
  with $G_0 \subseteq G_0'$.
\end{theorem}

\begin{proof}
  The proof follows that of Theorem~\ref{thm:shallow-noncollapsing} by replacing
  the use of Lemma~\ref{lem:NoncollapsingShallow} with
  Lemma~\ref{lem:main} and
  adding the subterm-closed assumptions on $G_0$ and $G$.
\end{proof}

\begin{corollary} \label{cor:main}
Non-$E$-overlapping and weakly shallow TRSs are confluent.
\end{corollary}

\begin{corollary}\label{cor:decidable}
  Non-$\omega$-overlapping and weakly shallow TRSs are confluent.
\end{corollary}

\begin{proof}
  A weakly shallow TRS is weak monadic, that is, right weakly-shallow.
  Hence a non-$\omega$-overlapping and weakly shallow TRS is
  non-$E$-overlapping by~\cite[Theorem~4.18]{Mitsuhashi12},
  and Corollary~\ref{cor:main} applies.
\end{proof}

Unlike the non-$E$-overlapping property, which is undecidable for weakly
shallow TRSs~\cite{Mitsuhashi12}, the non-$\omega$-overlapping property is
decidable.  Hence Corollary~\ref{cor:decidable} gives a decidable sufficient
condition for confluence.

\section{Conclusion}\label{sec10}
This paper proposed a new method for proving confluence of abstract
reduction systems by introducing the notions of edge compatibility and edge
commutativity.
By applying this technique, we have shown that \emph{non-$E$-overlapping and
weakly shallow TRSs are confluent}.
Since the non-$\omega$-overlapping property is decidable, this also yields
a decidable sufficient condition for confluence
(Corollary~\ref{cor:decidable}).
A future direction of this work is to relax the \emph{weakly shallow}
restriction to an \emph{almost weakly shallow} condition, which allows at most
one occurrence of a defined function symbol in each path from the root to
a variable.

Another direction is to relax it to the \emph{right weakly-shallow}
condition, which requires only the right-hand side of each rule to be weakly
shallow; this is the class of \emph{weak monadic}
TRSs~\cite{Mitsuhashi12}.
The weakly shallow property is used to make non-root reductions between
redexes correspond to reductions at variable positions, and this
correspondence also holds on the left-hand side of each non-$E$-overlapping
rule.
Hence we expect that the right weakly-shallow condition is enough.

\backmatter

\section*{Statements and Declarations}

\begin{itemize}
\item Funding: No funding was received for conducting this study.
\item Data availability: Not applicable. No datasets were generated or
  analysed during the current study.
\end{itemize}

%
%

\bibliography{article}

\setcounter{tocdepth}{2}

\end{document}